\documentclass[prx,aps,superscriptaddress,footinbib,twocolumn]{revtex4-2}

\usepackage{amsthm,amsmath,amssymb}
\usepackage{mathrsfs}

\newcommand{\st}{\,:\,}

\newcommand{\ket}[1]{\vert{ #1 }\rangle}

\usepackage{amsthm}
\newtheorem{theorem}{Theorem}

\newtheorem{lemma}{Lemma}

\theoremstyle{definition}
\newtheorem{definition}{Definition}
\theoremstyle{remark}
\newtheorem*{remark}{Remark}

\usepackage{amssymb}
\usepackage{graphicx}
\usepackage{epstopdf}

\usepackage{algorithm}
\usepackage[noend]{algpseudocode}

\usepackage{natbib}

\usepackage[colorlinks=true,citecolor=blue,linkcolor=magenta]{hyperref}

\usepackage{xcolor}

\usepackage{comment}

\usepackage{multirow}

\begin{document}

\title{Accelerating Fault-Tolerant Quantum Computation with Good qLDPC Codes}

\author{Guo Zhang}
\affiliation{Graduate School of China Academy of Engineering Physics, Beijing 100193, China}

\author{Yuanye Zhu}
\affiliation{Graduate School of China Academy of Engineering Physics, Beijing 100193, China}

\author{Ying Li}
\email{yli@gscaep.ac.cn}
\affiliation{Graduate School of China Academy of Engineering Physics, Beijing 100193, China}

\begin{abstract}
We propose a fault-tolerant quantum computation scheme that is broadly applicable to quantum low-density parity-check (qLDPC) codes. The scheme achieves constant qubit overhead and a time overhead of $O(d^{a+o(1)})$ for any $[[n,k,d]]$ qLDPC code with constant encoding rate and distance $d = \Omega(n^{1/a})$. For good qLDPC codes, the time overhead is minimized and reaches $O(d^{1+o(1)})$. In contrast, code surgery based on gauging measurement and brute-force branching requires a time overhead of $O(dw^{1+o(1)})$, where $d\leq w\leq n$. Thus, our scheme is asymptotically faster for all codes with $a < 2$. This speedup is achieved by developing techniques that enable parallelized code surgery under constant qubit overhead and leverage classical locally testable codes for efficient resource state preparation. These results establish a new paradigm for accelerating fault-tolerant quantum computation on qLDPC codes, while maintaining low overhead and broad applicability. 
\end{abstract}

\maketitle

\section{Introduction}

Quantum error correction is essential to realizing large-scale quantum computation~\cite{Shor1995,Kitaev1997,Gottesman1998,Gaitan2018}. According to the fault-tolerance theorem, when the physical error rate is below a constant threshold, arbitrarily high computational accuracy can be achieved by choosing a quantum error correction code with sufficiently large code distance $d$~\cite{Knill1998,Aharonov2008}. However, fault-tolerant quantum computation (FTQC) inevitably introduces additional resource overheads: each logical qubit requires multiple physical qubits for encoding and for supporting its computational operations, referred to as \emph{qubit overhead}; each round of logical-qubit operations must be implemented through multiple rounds of physical-qubit operations, referred to as \emph{time overhead}. For instance, when computation is performed using the surface code via lattice surgery, the qubit overhead scales as $O(d^2)$, while the time overhead scales as $O(d)$~\cite{Fowler2012,Horsman2012}. 

Substantial effort has been devoted to reducing the resource overhead of FTQC. A major advance in this direction is the development of low-overhead quantum low-density parity-check (qLDPC) codes, such as hypergraph product (HGP) codes, which achieve a constant encoding rate~\cite{Tillich2014}. This property ensures that the logical qubit number $k$ scales linearly with the physical qubit number $n$, i.e.~$k = \Theta(n)$. Recent breakthroughs have led to the discovery of good qLDPC codes, which further improve the code distance from $d = \Theta(n^{1/2})$ to $d = \Theta(n)$ while maintaining a constant encoding rate~\cite{Panteleev2021,Breuckmann2021a,Breuckmann2021,Panteleev2022,Dinur2022}. 

To enable logical computational operations on qLDPC codes, two main strategies have been developed: the concatenated codes combined with gate teleportation (CC+GT) approach~\cite{Gottesman2014,Yamasaki2024,Tamiya2024,Nguyen2024} and the code surgery approach (a generalization of lattice surgery)~\cite{Horsman2012,Cohen2022,Cross2024,Cowtan2024,Zhang2025,Ide2024,Williamson2024,Zhang2025a,He2025,Swaroop2025,Cowtan2025,Baspin2025}. For general qLDPC codes, the most efficient known method is the gauging measurement combined with brute-force branching (GM+BFB)---a form of code surgery that achieves constant qubit overhead and a time overhead scaling as $O(d^{2+o(1)})$ [assuming the logical operator weight $w = O(d)$]~\cite{Zhang2025,Ide2024,Williamson2024,Cowtan2025}. A lower time overhead of $O(d^{1+o(1)})$ can be attained using a refined variant of the CC+GT scheme~\cite{Nguyen2024}; however, its applicability is restricted to (almost) good quantum locally testable codes (qLTCs)---also known as $c^3$-qLTCs---which, beyond having constant encoding rate and constant relative distance, must also possess constant soundness~\cite{Nguyen2024,dinur2024expansion}.

In this work, we propose a new scheme applicable to general qLDPC codes, achieving lower resource overhead than the GM+BFB protocol and extending the $O(d^{1+o(1)})$ time overhead to a broader code families. For qLDPC codes with a constant encoding rate and code distance $d = \Omega(n^{1/a})$, our scheme attains a time overhead of $O(d^{a+o(1)})$ while maintaining constant qubit overhead. Compared with GM+BFB, our method achieves a smaller time overhead for all constant-rate qLDPC codes whose relative distance exceeds that of HGP codes (i.e., $a < 2$). In particular, for good qLDPC codes, the time overhead is further reduced to $O(d^{1+o(1)})$. Although this minimal overhead matches the best performance of the CC+GT approach, our scheme achieves it under weaker conditions—requiring only a constant encoding rate and a constant relative distance. 

Our scheme integrates code surgery with gate teleportation by performing the parity-check measurements required in code surgery through gate teleportation. To minimize resource overhead, we introduce two key techniques: parallelized code surgery (PCS) and locally-testable state preparation (LTSP). In PCS, multiple code blocks share a common ancilla system (a set of mutually coupled qubits) to implement computational operations, thereby reducing the qubit overhead of code surgery to a constant. In LTSP, the resource states required for gate teleportation are prepared with constant qubit and time overhead, employing a classical locally testable code (LTC) that requires only two constants---constant encoding rate and constant soundness~\cite{Leverrier2022LTC,Panteleev2022,lin2022c3LTC}. Leveraging local testability, we can effectively control measurement errors in parity-check measurements: when performing code surgery with resource states prepared via LTSP, only a single round of parity-check measurements is needed, reducing the time overhead of code surgery from $O(d)$ to $O(1)$. The introduction of PCS and LTSP inevitably leads to serialization of logical operations, resulting in a final time overhead of $O(d^{a+o(1)})$. As general overhead-reduction techniques, PCS and LTSP not only improve the asymptotic scaling of FTQC but also offer practical means to lower overhead in near-term, finite-size implementations. 

\begin{figure*}
\begin{minipage}{\linewidth}
\begin{table}[H]
\begin{tabular}{ccccc}
\hline\hline
Code & Protocol & Qubit overhead exponent & Time overhead exponent & Refs. \\
\hline\hline
\multirow{4}{*}{\shortstack{qLDPC codes \\ with $k = \Theta(n)$ \\ and $d = \Omega(n^{1/a})$}} & GM+BFB & $0$ & $\geq 2$ & \cite{Zhang2025a,Cowtan2025}\\
\cline{2-5}
& DS & $1$ & $1$ & \cite{Zhang2025} \\
\cline{2-5}
& polylog CC+GT & $0$ & $\geq 2a$& \cite{Tamiya2024} \\
\cline{2-5}
& PCS+LTSP+GT & $0$ & $a$ & This work \\
\hline\hline
\multirow{2}{*}{Good qLTCs} & log CC+GT & $0$ & $1$ &\cite{Nguyen2024} \\
\cline{2-5}
& PCS+LTSP+GT & $0$ & $1$ & This work \\
\hline\hline
Surface codes & LS & $2$ & $1$ & \cite{Litinski2019} \\
\hline
\end{tabular}
\caption{
Qubit and time overheads in fault-tolerant schemes using quantum low-density parity-check (qLDPC) codes. The exponent $b$ in the overhead expression $O(d^{b+o(1)})$ is referred to as the overhead exponent. Techniques developed for code surgery include gauging measurement (GM)~\cite{Ide2024,Williamson2024}, brute-force branching (BFB)~\cite{Zhang2025}, and devised sticking (DS)~\cite{Zhang2025}. {The overhead estimates for the GM+BFB protocol are predicated on the assumption that $w = O(d)$; any departure from this regime only increases the overhead exponents. In contrast, our protocol and other methods in the table do not rely on this assumption.} Typical protocols based on gate teleportation (GT) utilize concatenated codes (CC)~\cite{Gottesman2014,Nguyen2024,Tamiya2024}. The overhead results for the polylog CC+GT~\cite{Tamiya2024} protocol assume the use of quantum codes possessing single-shot properties. Lattice surgery (LS)~\cite{Litinski2019} is a method for quantum computation on the surface code (SC). This work introduces two new techniques for reducing time overhead: parallelized code surgery (PCS) and locally-testable state preparation (LTSP).}
\label{tab:schemes}
\end{table}
\end{minipage}
\end{figure*}

\section{Related works}

Two general approaches have been developed for fault-tolerant quantum computation with qLDPC codes: CC+GT and code surgery. 

Gate teleportation implements logical gates by consuming pre-prepared resource states~\cite{Gottesman2014,Nguyen2024,Tamiya2024}. Gottesman proposed using gate teleportation to operate on logical qubits encoded in qLDPC codes, with resource states generated via concatenated codes~\cite{Gottesman2014}. This approach is made possible by the existence of well-established fault-tolerant quantum computation protocols in certain concatenated codes, such as the concatenated Steane codes~\cite{Yamasaki2024}. In Gottesman's approach, the resource state is prepared by simulating its preparation circuit using fault-tolerant operations of the concatenated code, producing an encoded resource state. This is then decoded using the decoding circuit of the concatenated code to yield the unencoded resource state needed for gate teleportation. Since the simulation is fault-tolerant and the decoding circuit introduces only local stochastic errors at an asymptotically constant rate, the final resource state can be made sufficiently accurate to enable reliable gate teleportation. 

Several variants of the CC+GT approach have been developed, all achieving constant qubit overhead~\cite{Nguyen2024,Tamiya2024}. Their primary difference lies in the time overhead. Consider a quantum computer consisting of $M$ memory blocks, each encoded using an $[[n, k, d]]$ qLDPC code. To implement a quantum circuit of size $\vert C\vert$ on the $K = Mk$ logical qubits with a target failure rate $\epsilon$, the logical error rate per operation must be reduced to $p_L = O(\epsilon / \vert C\vert)$. Achieving this requires the concatenated code to have sufficiently large distance, ensuring that the encoded resource states are prepared with error rate $O(p_L)$. This requirement imposes a significant qubit cost on resource state production. To maintain a constant total qubit overhead, resource states are provided to only a subset of memory blocks at any given time. As a result, each layer of logical operations must be executed over multiple time steps. The overall time overhead is determined by both the cost of preparing resource states and the complexity due to serialization. In Gottesman's original protocol, and its subsequent improvement~\cite{Gottesman2014,Yamasaki2024}, this overhead scales as $O(\mathrm{poly}(K))$. A more recent advancement~\cite{Tamiya2024} reduces the time overhead to 
\[
O\left(\log^{2a + \gamma_1 + \gamma_2} \left(\frac{\vert C\vert}{\epsilon}\right)\right) = O\left(d^{2a + \gamma_1 + \gamma_2}\right),
\]
where $\gamma_1$ and $\gamma_2$ are parameters characterizing the concatenated code; note that $\vert C\vert / \epsilon = O(\mathrm{poly}(K))$ in typical scenarios. Here, we have re-expressed the time overhead by assuming an exponential suppression of the logical error rate with code distance, i.e.,~$p_L = e^{-\Omega(d)}$, allowing for a direct comparison with code-surgery protocols. 

A further breakthrough was reported in Ref.~\cite{Nguyen2024}, which introduced state distillation protocols capable of achieving an almost constant spacetime overhead. Before distillation, a concatenated code with modest distance is used to produce resource states with a sufficiently small constant error rate. These are subsequently distilled to achieve the target error rate of $O(p_L)$. Due to the efficiency of the distillation process, the time cost of preparing resource states becomes negligible, making serialization cost the dominant contribution to the overall time overhead. Consequently, the total time overhead is reduced to
\[
O\left(\log^{1+o(1)} \left(\frac{\vert C\vert}{\epsilon}\right)\right) = O(d^{1+o(1)}).
\]
However, since the efficient distillation relies on local testability, achieving this $O(d^{1+o(1)})$ overhead requires the underlying code to be an (almost) good qLTC, characterized by triple constants: constant encoding rate, constant relative distance, and constant soundness. 

Code surgery is a generalization of lattice surgery, originally developed for implementing logical operations on surface codes~\cite{Horsman2012,Horsman2012,Cohen2022,Cross2024,Zhang2025,Ide2024,Williamson2024,Zhang2025a,He2025,Cowtan2025,Baspin2025}. In code surgery, the original code is temporarily deformed by coupling it to an ancilla system, enabling the measurement of one or more selected logical Pauli operators. These logical measurements can then be used to implement Clifford gates. Code surgery for qLDPC codes was first proposed in Ref.~\cite{Cohen2022}, known as the CKBB protocol, named after its authors. In this protocol, a logical operator is measured using an ancilla system of $O(dw)$ physical qubits, where $w$ is the weight of the logical operator. Since $w \geq d$, the required ancilla size is comparable to that of a surface code. Subsequent works introduced gauging measurements, which reduced the ancilla size to $O(w^{1+o(1)})$ physical qubits~\cite{Cross2024,Ide2024,Williamson2024}. Gauging measurements achieve this by transforming the relevant checks into those of a repetition code and constructing the ancilla system using an expander graph, thereby achieving robustness against certain errors that the CKBB protocol mitigates using additional qubits. 

Code surgery faces an inherent limitation in time efficiency: the parity-check measurements of the deformed code must be repeated for $\Theta(d)$ rounds to reliably correct measurement errors. As a result, lattice surgery on surface codes incurs a time overhead of $O(d)$. In the CKBB and gauging measurement protocols, when logical operators are measured sequentially---one in each block at a time---the total time overhead for measuring $\Theta(k)$ logical operators (a layer of operations on $k$ logical qubits) is $O(kd) = O(d^{1+a})$. 

To improve time efficiency, general techniques have been developed to construct deformed codes that support the simultaneous measurement of multiple logical operators. Two main approaches are devised sticking and brute-force branching~~\cite{Zhang2025}. In devised sticking, an arbitrary number of logical operators can be measured in parallel using a single ancilla system of $O(nd)$ physical qubits~\footnote{Throughout this paper, when discussing devised sticking, we refer specifically to the case where the logical operators to be measured are (i) expressed in standard form, and (ii) have logical thickness one, meaning they act on mutually disjoint sets of logical qubits~\cite{Zhang2025}.}. This approach has a time overhead of $O(d)$ and a qubit overhead of $O(d)$, assuming the underlying code has constant encoding rate. In brute-force branching, an intermediate deformed code is constructed to distribute the logical operators across disjoint supports, enabling their measurements to be performed in parallel. A protocol that combines gauging measurements with brute-force branching (GM+BFB)~\cite{Cowtan2025} achieves a time overhead of $O(dw^{1+o(1)})$ while reducing the qubit overhead to constant. 

Note that code surgery operations alone can only implement Clifford logical gates. A parallelized magic state injection protocol enables the preparation of magic states with (almost) the same resource scaling as code surgery~\cite{Zhang2025a}. 

In addition to typical CC+GT and code surgery protocols, Ref.~\cite{Tamiya2024} proposes preparing gate-teleportation resource states using surface codes (or alternatives such as color codes, rather than concatenated codes), and then implementing gate teleportation on high-rate qLDPC code blocks via code surgery. 

\section{Scheme}


We now sketch our scheme for FTQC. In our scheme, we adopt the code-surgery approach and propose to implement code surgery via gate teleportation. We introduce two key techniques to minimize qubit and time overheads: PCS and LTSP; see Fig.~\ref{fig:scheme} and Sec.~\ref{sec:CS}. 

{In the PCS framework, we simultaneously operate on multiple logical qubits while maintaining a strictly constant qubit overhead. This is achieved by coupling an ancilla system to the memory code (which encodes the data logical qubits) to form a deformed qLDPC code. The ancilla system is constructed as a hypergraph product of two linear codes. The first linear code determines the specific logical operations to be applied, while the second linear code---referred to as the \textit{R code}---is a classical LDPC code with a constant encoding rate. A central finding of this work is that the codewords of the R code establish parallel channels for logical operations: if the R code has a logical dimension $k_R$, a single ancilla system can be coupled to $k_R$ distinct memory-code blocks, enabling parallel logical operations across all blocks. Since both the memory code and the R code possess constant encoding rates, the ancilla system has a total size of $O(kk_R)$, where $k$ is the number of logical qubits per block. This approach permits the simultaneous execution of logical operations on $O(kk_R)$ logical qubits distributed across the blocks, thereby achieving high parallelism with constant qubit overhead.

To ensure high time efficiency, we propose the LTSP scheme to circumvent the temporal bottlenecks of conventional code surgery. Standard implementations typically require $O(d)$ rounds of repeated stabilizer measurements to achieve fault tolerance against measurement errors. We overcome this requirement by introducing \textit{single-shot} properties. Rather than requiring the memory code or the deformed code itself to be single-shot~\cite{Campbell2019,cowtan2025fast}, we demonstrate that single-shot measurements can be effectively realized by utilizing an independent single-shot code within some state preparation circuits. Specifically, we implement stabilizer measurements in a single round via gate teleportation. By utilizing a carefully prepared gate-teleportation resource state with bounded error weight, a single measurement round is sufficient to reliably extract the stabilizer information. This approach shifts the single-shot requirement from the code surgery circuit to the circuit preparing resource states. To control errors on the resource state, we encode each qubit in the preparation circuit as a logical qubit of an \textit{F code}. By selecting the F code to be a classical LTC with a constant encoding rate and constant soundness, we can prepare resource states with (almost) constant qubit overhead in constant depth.

By combining PCS and LTSP, logical operations can be implemented in parallel with constant qubit overhead and circuit depth. However, the architecture naturally imposes a constraint, requiring identical operations to be performed across multiple blocks simultaneously. For general heterogeneous circuits, this constraint necessitates the \textit{serialization} of logical operations, introducing a time overhead factor of $O(k) = O(d^a)$. Additional polylogarithmic factors $O(d^{o(1)})$ arise from the requirements for robustness against stochastic noise and magic state distillation. Ultimately, through a time-space trade-off, we achieve a strictly constant qubit overhead while maintaining a total time complexity within the $O(d^{a+o(1)})$ class. These results are summarized in the following theorem. }

\begin{theorem}
\label{the:FTQC}
Consider a classical-input/classical-output quantum circuit $C$ composed of Clifford and $T$ gates, with size $|C| \leq W D$, where $W$ denotes the width and $D = \mathrm{poly}(W)$ denotes the depth. There exists a constant $\epsilon_* \in (0,1)$ such that for any $\epsilon_L \in (0,1)$, one can efficiently construct a fault-tolerant version $C_{\mathrm{FT}}$ with width $W_{\mathrm{FT}}$ and depth $D_{\mathrm{FT}}$ satisfying 
\begin{align}
\frac{W_{\mathrm{FT}}}{W} &= O_{W \to \infty}(1), \\
\frac{D_{\mathrm{FT}}}{D} &=O_{W \to \infty}(\mathrm{log}^{a+o(1)}(\frac{|C|}{\epsilon_L})),
\end{align}
where $a$ is a constant associated with the memory code block. If $C_{\mathrm{FT}}$ operates under a local stochastic noise model and physical error rate $p \le \epsilon_*$ at every location, then $C_{\mathrm{FT}}$ simulates $C$ up to output error $\epsilon_L$, in the sense that their output distributions differ by at most $\epsilon_L$.
\end{theorem}

A detailed proof is provided in Appendix~\ref{app:threshold}. 

\begin{remark}
{The theorem assumes a physical qubit network with all-to-all connectivity, where two-qubit gates can be applied between any pair of physical qubits. In this regime, the logical qubit network inherits this all-to-all connectivity, allowing the fault-tolerant circuit $C_{\mathrm{FT}}$ to implement any logical circuit $C$ with the stated overheads. Notably, our protocol can be adapted to much more stringent architectural constraints. Specifically, on a physical network with almost constant vertex degree $D_P = O(d^{o(1)})$, we still achieve an almost constant qubit overhead and a time overhead of $O(d^{a+o(1)})$. In this restricted setting, however, the logical qubit network is characterized by a constant vertex degree $D_L = \Theta(1)$. For further details, see Appendix~\ref{app:almost_constant_degree}. }
\end{remark}

\begin{figure*}[htbp]
\centering
\includegraphics[width=0.9\linewidth]{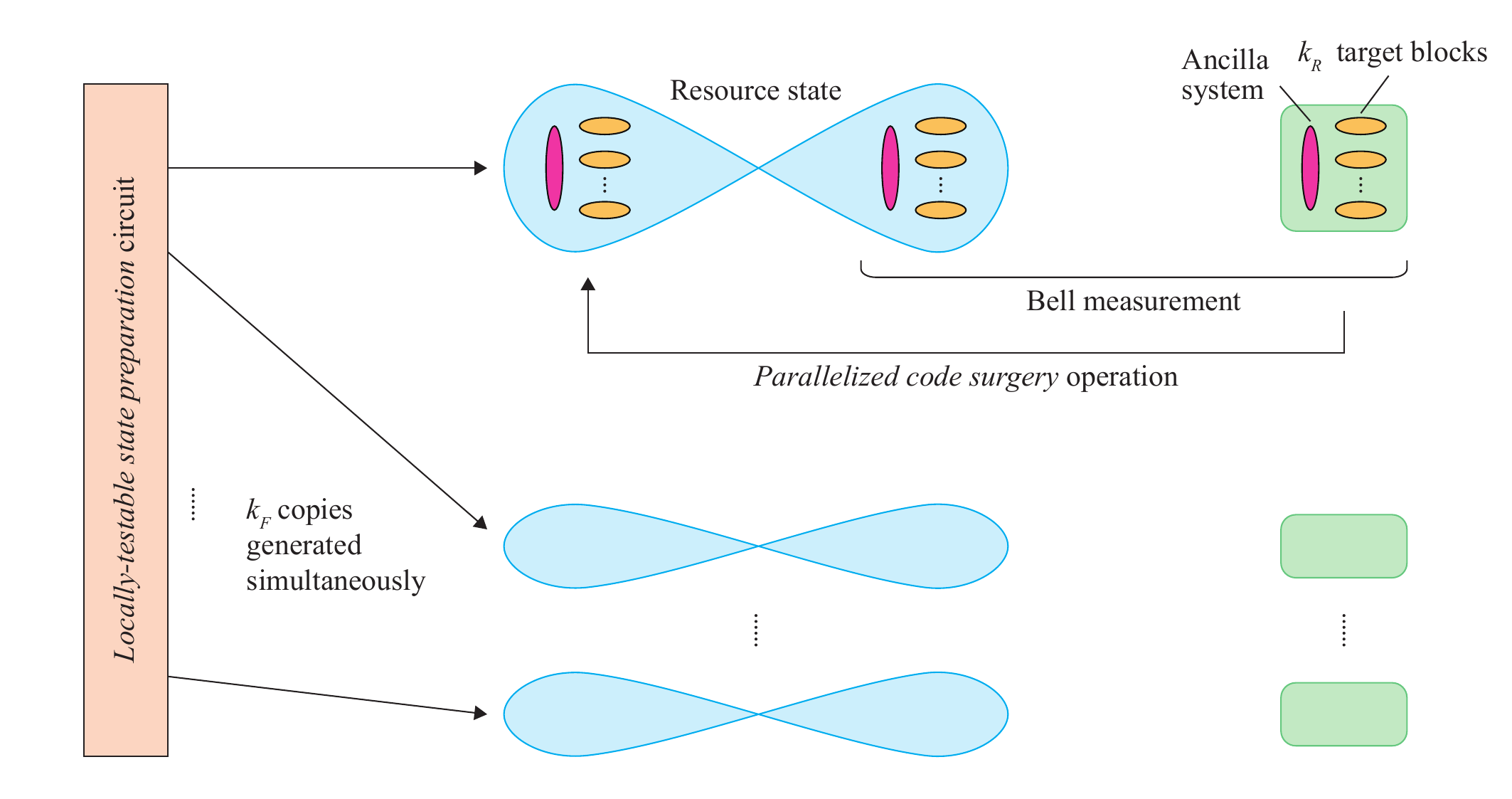}
\caption{
Schematic illustration of our scheme. Each execution of the locally-testable state-preparation circuit produces $k_F$ copies of the resource state, which are then used in gate teleportation to implement parallelized code surgery (PCS). PCS applies logical operations simultaneously on $k_R$ target code blocks, using a single ancilla system shared across the $k_R$ blocks. The locally-testable state-preparation circuit has almost constant qubit and time overhead, while gate-teleportation-based PCS also maintains constant qubit and time overhead.
}
\label{fig:scheme}
\end{figure*}

\begin{figure*}
\begin{minipage}{\linewidth}
\begin{table}[H]
\begin{tabular}{ccc}
\hline\hline
Codes & Parameters & Examples \\
\hline\hline
\multirow{5}{*}{Memory code (Q)} & \multirow{5}{*}{$k = \Theta(n)$ and $d = \Omega(n^{1/a})$} & Hypergraph product, $a = 2$ \cite{Tillich2014} \\
& & Quantum Tanner, $a = 1$ \cite{Leverrier2022} \\
& & Expander LP , $a = 1+\epsilon$ \cite{Panteleev2022} \\
& & DLV, $a = 1+\epsilon$ \cite{dinur2024expansion} \\
& & LH,  $a = 1$ \cite{Lin2022} \\
& & DHLV, $a = 1$ \cite{Dinur2022} \\
\hline
\multirow{2}{*}{R code (C)} & \multirow{2}{*}{$k_R = \Theta(n_R)$ and $d_R = \Omega\left(n_R^{1/a_R}\right)$} & Expander, $a_R = 1$ \cite{556667} \\
& & Tanner, $a_R = 1$ \cite{1056404} \\
\hline
F code (C) & $k_F = \Theta(n_F)$, $d_F = \Omega\left(n_F^{1/a_F}\right)$, and $s = \Theta(1)$  & PK, $a_F = 1$ \cite{Panteleev2022}\\& & ~~Lossless-expander LTCs, $a_F = 1$ \cite{lin2022c3LTC}
\\
\hline
\end{tabular}
\caption{
Quantum (Q) and classical (C) codes utilized in our protocol. The quantum computer is structured into multiple blocks, each encoded in a qLDPC code, called memory code. The quantum computation is performed on logical qubits within these memory-code blocks via code surgery, where each operation corresponds to a specific deformed code. The deformed code is generated utilizing a classical linear code, called R code. To implement the parity-check measurements of a deformed code (or the memory code), we prepare the corresponding gate-teleportation resource states using another classical linear code, referred to as the F code. This code has to be a locally testable code, and $s$ denotes its soundness. In examples, $\epsilon$ denotes an arbitrary positive number. Additionally, we employ a low-distance surface code to suppress errors during resource state preparation and to inject $T$-gate magic states, and a high-distance color code to prepare $S$-gate magic states. We use $d_S$ and $d_C$ to denote distances of the surface and color codes, respectively. 
}
\label{tab:codes}
\end{table}
\end{minipage}
\end{figure*}

\subsection{Codes}

Our scheme employs a set of error correction codes, as listed in Table~\ref{tab:codes}. A qLDPC code, referred to as the memory code, is used to encode logical qubits. We perform operations on these logical qubits via code surgery, which requires the construction of corresponding deformed codes. A classical LDPC code, called the R code, is used in constructing the deformed codes, replacing the role of the repetition code in the CKBB protocol. Another classical LDPC code, referred to as the F code, serves as the core component of the resource-state factory, which produces the resource states required for gate teleportation. The requirements for these codes, along with representative examples, are presented in Table~\ref{tab:codes}. Additionally, our scheme also employs a surface code and a color code for various auxiliary tasks. 

\subsection{Quantum computer}

\begin{figure*}[htbp]
\centering
\includegraphics[width=0.9\linewidth]{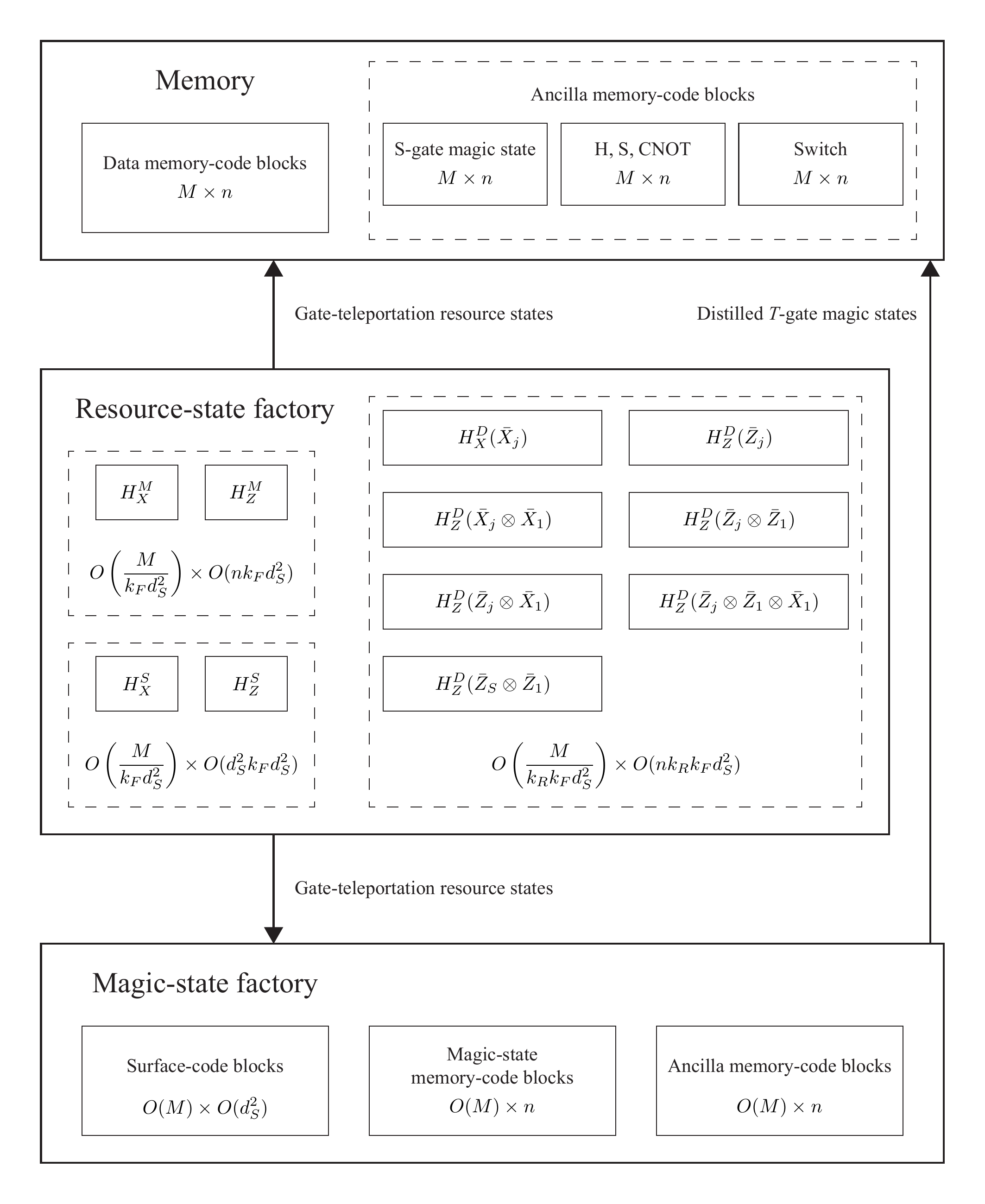}
\caption{
Quantum computer and allocation of physical qubits.
In the memory, the memory-code blocks are divided into four groups according to their functions. Each group contains $M$ blocks, and each block occupies $n$ physical qubits.
The resource-state factory produces resource states specified by the corresponding check matrices, which are organized into three families (indicated by the dashed boxes). In each dashed box, the notation $O(A)\times O(B)$ denotes the total number of physical qubits used for generating resource states within that family: each type of resource state is prepared by a tailored locally-testable state-preparation circuit requiring $O(B)$ physical qubits, and $O(A)$ such circuits are run in parallel.
In the magic-state factory, $O(M)$ surface-code blocks are used to inject $T$-gate magic states into $O(M)$ memory-code blocks, while an additional $O(M)$ ancillary memory-code blocks are employed for magic-state distillation.
}
\label{fig:QC}
\end{figure*}

As depicted in Fig.~\ref{fig:QC}, the quantum computer comprises three main components: the memory, the resource-state factory, and the magic-state factory. The memory consists of multiple memory-code blocks. $M$ of the blocks store the logical information, giving a total of $K = Mk$ data logical qubits. Here, $k$ is the logical dimension of the memory code. The remaining memory-code blocks serve as ancillas. The resource-state factory produces the resource states required for gate teleportation, while the magic-state factory prepares encoded magic states. 

\begin{figure*}
\begin{minipage}{\linewidth}
\begin{table}[H]
\begin{tabular}{cc}
\hline\hline
~Logical operator (set) $\sigma$~ & Resource states \\
\hline\hline
$\bar{X}_j$ & $H^D_Z(\bar{X}_j)$, $H^M_Z \times k_R$ \\
\hline
$\bar{Z}_j$ & $H^D_Z(\bar{Z}_j)$, $H^M_X \times k_R$ \\
\hline
$\bar{X}_j \otimes \bar{X}_1$ & $H^D_Z(\bar{X}_j \otimes \bar{X}_1)$, $H^M_Z \times 2k_R$ \\
\hline
$\bar{Z}_j \otimes \bar{Z}_1$ & $H^D_Z(\bar{Z}_j \otimes \bar{Z}_1)$, $H^M_X \times 2k_R$ \\
\hline
$\bar{Z}_j \otimes \bar{X}_1$ & $H^D_Z(\bar{Z}_j \otimes \bar{X}_1)$, $H^M_X \times k_R$, $H^M_Z \times k_R$ \\
\hline
$\bar{Z}_j \otimes \bar{Z}_1 \otimes \bar{X}_1$ & ~~$H^D_Z(\bar{Z}_j \otimes \bar{Z}_1 \otimes \bar{X}_1)$, $H^M_X \times 2k_R$, $H^M_Z \times 2k_R$~~ \\
\hline
$\bar{Z}_S \otimes \bar{Z}_1$ & $H^D_Z(\bar{Z}_S \otimes \bar{Z}_1)$, $H^S_X \times k_R$, $H^M_X \times k_R$ \\
\hline
$\bar{Z}_C \otimes \bar{Z}_1$ & $H^D_Z(\bar{Z}_C \otimes \bar{Z}_1)$, $H^C_X \times k_R$, $H^M_X \times k_R$ \\
\hline
\end{tabular}
\caption{
List of logical measurements. $\bar{X}_j$ ($\bar{Z}_j$) denotes the $X$ ($Z$) logical operator of the $j$th logical qubit in a memory-code block; $\bar{Z}_S$ is the $Z$ logical operator of a distance-$d_S$ surface-code block; and $\bar{Z}_C$ is the $Z$ logical operator of a distance-$d_C$ color-code block. A resource state is specified by the corresponding check matrix $H$: the state $H$ is used to measure $Z$ stabilizer operators represented by $H$. The notation $H \times m$ denotes $m$ copies of the resource state. Measurements of $X$ stabilizer operators are obtained by applying transversal Hadamard gates to rotate the basis. We use $H^c_P$ with $c = M,D,S,C$ to represent the $P\in \{X,Z\}$ check matrix of the memory, deformed, surface, and color codes, respectively; and $H^D_Z(\sigma)$ denotes the $Z$ check matrix of the deformed code tailored to the measurement of logical operator(s) $\sigma$. Each set of resource states enables the simultaneous application of the same logical measurement on $k_R$ identical systems (target blocks). 
}
\label{tab:measurements}
\end{table}
\end{minipage}
\end{figure*}

\subsection{Logical measurements---Parallelized code surgery}

\begin{figure}[htbp]
\centering
\includegraphics[width=\linewidth]{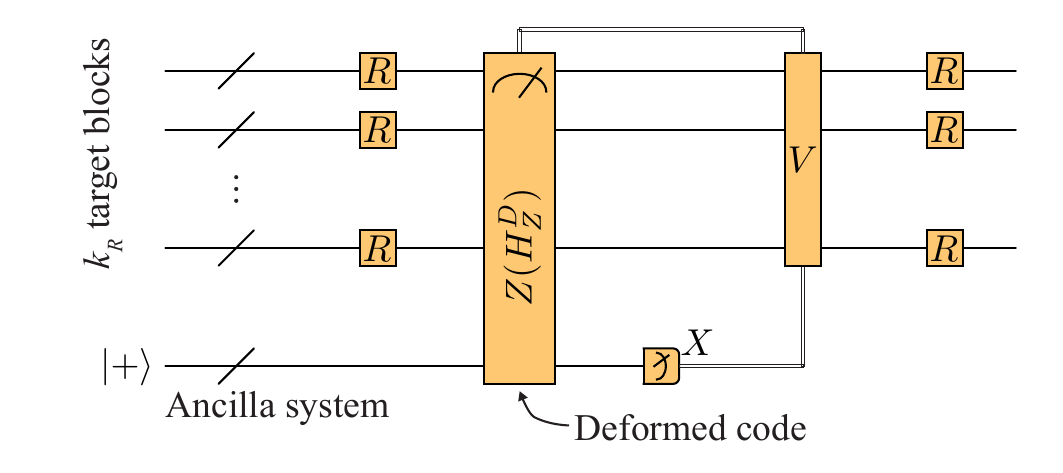}
\caption{
Circuit for parallelized code surgery. The ancilla system is transversally initialized in the state $\ket{+}$ and measured in the $X$ basis. The $R$ gate rotates the logical measurement basis to $Z$: for each memory-code block on which the measured Pauli operator acts as $X$, a transversal Hadamard gate is applied. The central step of code surgery is a single round of parity-check measurements of the deformed code, specifically measuring $Z$ stabilizer operators $Z(H^D_Z)$. Implemented via gate teleportation with a resource state from locally-testable state preparation, this single round suffices to guarantee fault tolerance. The feedback gate $V$ is a Pauli gate. 
}
\label{fig:PCS}
\end{figure}

We achieve universal quantum computation using measurements of logical Pauli operators, supplemented with encoded magic states. Table~\ref{tab:measurements} summarizes all logical measurements used in our scheme. These measurements are realized through PCS; see Fig.~\ref{fig:PCS} for the circuit. Each logical measurement corresponds to a specifically tailored deformed code. The construction of the deformed codes is detailed in Sec.~\ref{sec:PCS}. A key feature of PCS is that the deformed codes require only a constant qubit overhead. 

We refer to the collection of code blocks involved in a logical measurement as a target block. For example, the measurement $\bar{Z}_j \otimes \bar{Z}_1 \otimes \bar{X}_1$ acts on a target block composed of three memory-code blocks. In conventional code surgery, a logical measurement is performed on a single target block at a time. In contrast, PCS performs identical logical measurements simultaneously across $k_R$ identical target blocks, where $k_R$ is the logical dimension of the R code (see Fig.~\ref{fig:scheme}). {In this context, `parallelized' refers specifically to the batching of identical operations applied to multiple blocks. }

\subsection{Gate teleportation---Locally-testable state preparation}

\begin{figure}[htbp]
\centering
\includegraphics[width=\linewidth]{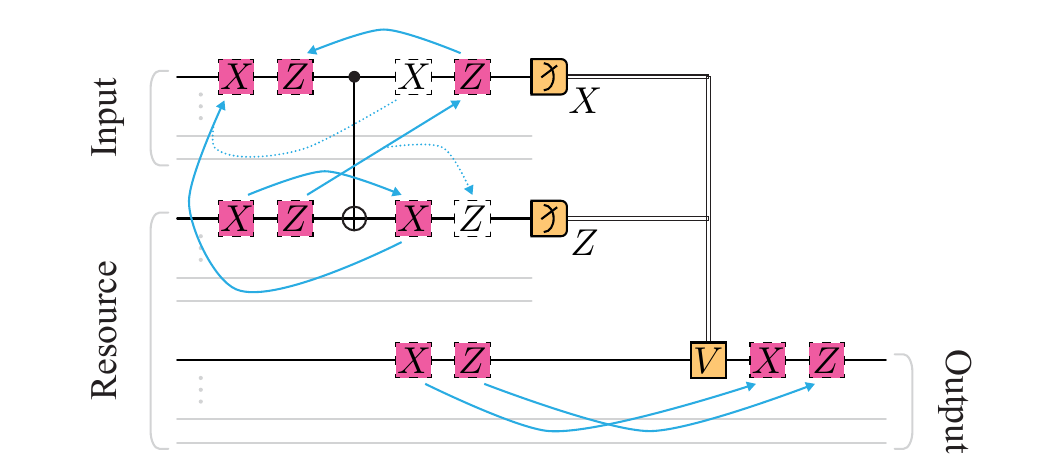}
\caption{
Gate teleportation. The circuit involves three sets of qubits: the first set carries the input state, while the second and third sets encode the resource state. The diagram illustrates operations on a representative qubit from each set (black lines), with the same operations applied in parallel to the remaining qubits (gray lines). Depending on measurement outcomes, a Pauli gate $V$ is applied on the output qubit. Dashed boxes indicate Pauli errors, and white boxes denote trivial errors. Arrows trace the paths of error propagation through the circuit. 
}
\label{fig:teleportation}
\end{figure}

The fundamental operations in error correction and code surgery are parity-check measurements. In our scheme, we perform parity-check measurements on various quantum codes: the memory code, surface code, color code, and deformed codes. In the standard CSS code formalism, each code corresponds to two sets of stabilizer operators---$X$-type and $Z$-type---represented by the corresponding check matrices $H_X$ and $H_Z$. We measure the $X$ and $Z$ generators separately. For each of $H_X$ and $H_Z$, the parity-check measurements are implemented via gate teleportation, each consuming a dedicated resource state; see Fig.~\ref{fig:teleportation} for the circuit. 

\begin{figure}[htbp]
\centering
\includegraphics[width=\linewidth]{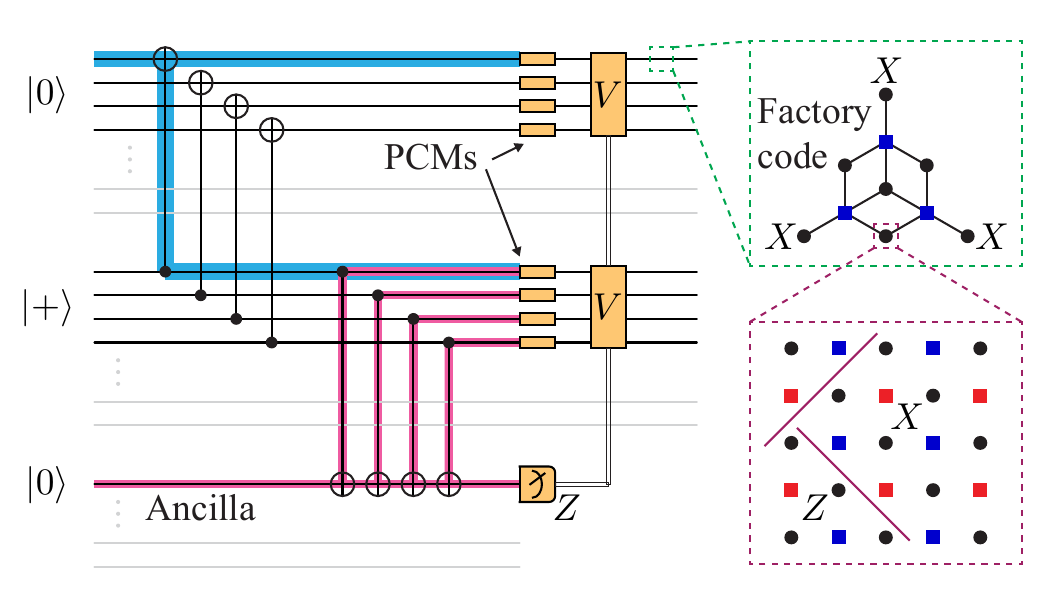}
\caption{
Circuit for generating resource states used to measure $Z$ stabilizers. The diagram illustrates operations on selected representative qubits (black lines), with similar operations applied in parallel to the remaining qubits (gray lines). The procedure begins by applying transversal controlled-NOT gates between two sets of qubits to prepare Bell states. An ancilla qubit is then used to measure each $Z$-stabilizer operator of the code, with the illustrated example assuming a weight-four stabilizer for the representative ancilla qubit. Measurement outcomes of the ancilla qubits determine Pauli gates $V$, which are subsequently applied. Each qubit in this circuit is encoded using a F code, which has only $Z$-type stabilizers. An example Tanner graph is shown, where blue (red) squares represent $Z$ ($X$) stabilizers. Note that the depicted graph is from the $[7,4,3]$ Hamming code; in practice, the F code should be an LTC with large enough distance. Furthermore, each qubit of the F code is itself encoded in a low-distance surface code, forming a two-level encoding scheme. All operations are performed via transversal gates within both the factory and surface codes. After preparing the encoded resource states, decoding is performed to obtain unencoded resource states. This is achieved by measuring selected qubits in the $X$ and $Z$ bases as illustrated in the figure and applying Pauli gates to the remaining qubits based on the outcomes. The measurement pattern is determined by the standard form of the generator matrices. The entire decoding process can be completed in $O(1)$ time, producing $k_F$ copies of unencoded resource states per circuit run. To correct errors, parity-check measurements (PCMs) of the F code are performed before decoding. Since the F code is an LTC, measurement errors in PCMs are equivalent to low-weight Pauli errors, and a single round of PCMs suffices to reliably detect and correct $X$ errors (including measurement errors on ancilla qubits) by exploiting the propagation of F-code $Z$ stabilizers (highlighted in bold cyan and magenta lines). As a result, each resource state copy has only low-weight residual errors. 
}
\label{fig:LTSP}
\end{figure}

We prepare the gate-teleportation resource states using LTSP. Each type of resource state is prepared by a tailored state-preparation circuit. Each execution of such a circuit generates $k_F$ identical copies of the corresponding resource state, where $k_F$ denotes the logical dimension of the F code; see Figs.~\ref{fig:scheme}~and~\ref{fig:LTSP}. We note that each qubit in the circuit is encoded via the concatenation of the F code and a surface code to protect against both types of Pauli errors. The detailed construction of these circuits is provided in Sec.~\ref{sec:LTSP}. 

A key feature of LTSP is that, by encoding with a constant-soundness LTC, the prepared resource states possess a crucial property: the parity-check measurements performed using these resource states are effectively free of measurement errors. Consequently, only a single round of measurements is required in PCS (and other code-surgery protocols). 

\begin{figure*}
\begin{minipage}{\linewidth}
\begin{table}[H]
\begin{tabular}{ccc}
\hline\hline
Logical operations & Logical measurements & Extra resource states \\
\hline\hline
Initialization & & $H^M_X$ \\
\hline
$Mea_j$ & $\bar{Z}_j$ & \\
\hline
$H_j$ & $\bar{Z}_j \otimes \bar{Z}_1$, $\bar{X}_j$, $\bar{Z}_j \otimes \bar{X}_1$, $\bar{Z}_1$ & $H^M_Z$ \\
\hline
$S_j$ & $\bar{Z}_1 \otimes \bar{Z}_1$, $\bar{Z}_j \otimes \bar{Z}_1 \otimes \bar{X}_1$, $\bar{X}_1$ & $H^M_Z$ \\
\hline
$T_j$ & ~$\bar{Z}_j \otimes \bar{Z}_1$, $\bar{X}_1 \times 2$, $\bar{Z}_1 \otimes \bar{Z}_1$, $\bar{Z}_j \otimes \bar{Z}_1 \otimes \bar{X}_1$~ & ~$T$-gate magic state, $H^M_Z$~ \\
\hline
$CNOT_{a,b}$ & $\bar{Z}_{a} \otimes \bar{Z}_1$, $\bar{X}_{b} \otimes \bar{X}_1$, $\bar{Z}_1$ & $H^M_Z$ \\
\hline\hline
Noisy $T$-gate magic state preparation & $\bar{Z}_S \otimes \bar{Z}_1$ & $H^S_X$, $H^S_Z$, $H^M_Z$\\
\hline
~Fault-tolerant $S$-state magic state preparation~ & $\bar{Z}_C \otimes \bar{Z}_1$ & $H^C_X$, $H^C_Z$, $H^M_Z$\\
\hline
\end{tabular}
\caption{
Universal set of logical operations employed in our scheme. 
Initialization prepares all logical qubits of a memory-code block transversally in the state $\ket{0}$. 
$Mea_j$, $H_j$, $S_j$, and $T_j$ denote, respectively, the $Z$-basis measurement, Hadamard gate, $S$ gate, and $T$ gate acting on the $j$th logical qubit of a memory-code block. Note that $Mea_j$ can be used to reinitialize a logical qubit. $CNOT_{a,b}$ is the controlled-NOT gate with control the $a$th logical qubit of block $A$ and target the $b$th logical qubit of block $B$, where $A$ and $B$ may be the same or different blocks. Each logical operation is implemented through a set of logical measurements and therefore consumes corresponding resource states; the table also lists any additional resource states consumed in addition to those used in the logical measurements. A $T$ gate consumes one copy of the $T$-gate magic state, whereas an $S$ gate requires one copy of the $S$-gate magic state as input but does not consume it. 
In addition to the universal operation set, we also list the magic-state preparation operations. 
}
\label{tab:operations}
\end{table}
\end{minipage}
\end{figure*}

\subsection{Universal quantum computation}

We adopt the following universal set of logical gates: Hadamard gate, $S$ gate, $T$ gate, and controlled-NOT gate. Using logical measurements, we can realize controlled-NOT and Hadamard gates directly, while $S$ and $T$ gates are implemented via logical measurements assisted by their corresponding magic states. In addition to the logical gates, initialization and measurement of memory-code logical qubits are also required. Table~\ref{tab:operations} summarizes all logical operations, with relevant circuits given in Appendix~\ref{app:operations}. 

For the $S$-gate magic state, we prepare it using a color code with a sufficiently large code distance $d_C$. The state is first created on a color-code block using the transversal $S$ gate and then transferred to a memory-code block via logical measurements. Because the operations are fault-tolerant, the resulting magic state attains fault-tolerant–level fidelity, eliminating the need for further magic-state distillation. 

For the $T$-gate magic state, we prepare it using a surface code with a small code distance $d_S$. A noisy magic state is first generated on a surface-code block following the protocol in Ref.~\cite{Li2014AMS}, and then transferred to a memory-code block. Since the prepared state is noisy, further distillation is required, which is performed using a constant-spacetime-overhead protocol~\cite{Nguyen2024}. 

Each logical $S$ gate requires one $S$-gate magic state as input. Since the magic state is preserved after use, it can be reused. To exploit this, we prepare $M$ copies of the $S$-gate magic state prior to computation, with each copy stored in an ancilla memory-code block. These ancilla blocks serve $S$ gates for data memory-code blocks throughout the entire computation. 

In contrast, each $T$ gate consumes one copy of the $T$-gate magic state. Therefore, these magic states must be generated on demand during computation. 

\subsection{Cost analysis}
\label{sub:cost analysis}
We now present a scheme for qubit allocation and logical circuit compilation, followed by an analysis of the associated resource overhead. While the scheme leaves significant room for further optimization, it is sufficient to establish our main conclusions: a constant qubit overhead and a time overhead of $O(kd^{o(1)}) = O(d^{a+o(1)})$. 

The allocation of qubits in the quantum computer is illustrated in Fig.~\ref{fig:QC}. In addition to data memory-code blocks, the memory includes ancilla memory-code blocks organized into three sectors: one stores $S$-gate magic states, another supports the implementation of $H$, $S$, and controlled-NOT gates, and the third manages slot switching in PCS. For the latter, when logical operations are performed on $q < k_R$ target blocks, ancilla blocks are used to fill the remaining $k_R - q$ slots. Each sector contains $M$ ancilla blocks. Since the memory code has constant encoding rate, the total number of physical qubits in memory is $O(Mk)$, where $Mk$ is the total number of logical qubits in the data memory-code blocks. 

{The resource-state and magic-state factories also require physical qubits. Assume the F code has a constant encoding rate. To serve all $M$ data blocks simultaneously, the factories would naively require $O(Mkd_S^2)$ physical qubits, where the factor $d_S^2$ arises from the encoding with the surface code (see Fig.~\ref{fig:LTSP}). To recover a \textit{strictly constant} qubit overhead, we implement a time-space trade-off: by using fewer physical qubits in the factories and serving data blocks in batches of $O(M/d_S^2)$, the qubit requirement for the factories is reduced to $O(Mk)$ (see Fig.~\ref{fig:QC}). This trade-off incurs an additional time overhead factor of $O(d_S^2) = O(\mathrm{polylog}(d))$, which is absorbed into the $O(d^{o(1)})$ term of our final complexity results. Altogether, the quantum computer utilizes $O(Mk)$ physical qubits, achieving the target of constant qubit overhead.}

{We now address the time overhead of the protocol. As we will show in Section~\ref{sec:CS}, the LTSP scheme generates resource states with almost constant qubit overhead and circuit depth. Utilizing these resource states, the PCS gadget can then be implemented with strictly constant qubit overhead and circuit depth. Therefore, the integration of PCS and LTSP enables the parallel processing of $O(Mk)$ logical qubits using almost $O(Mk)$ physical qubits (without the time-space trade-off) in almost constant depth, potentially achieving a nearly constant spacetime cost. 

However, PCS and LTSP impose specific compilation constraints: each execution of PCS implements identical operations across $k_R$ code blocks, and each run of the LTSP circuit generates $k_F$ identical resource state copies. These characteristics introduce a single-instruction-multiple-data constraint on the logical circuit compilation. For structured circuits with repetitive gate motifs---such as the Trotterized simulation of Hamiltonians with translational symmetry---it may be possible to approach nearly constant spacetime overhead. Conversely, for general circuits with heterogeneous patterns, these constraints necessitate the \textit{serialization} of logical operations. This serialization process is the primary source of the $O(d^{a+o(1)})$ time overhead reported in our final results. }

Logical operations are serialized as follows. Consider compiling a single layer of heterogeneous logical operations acting on disjoint data logical qubits. These operations are serialized into multiple sub-layers so that, in each sub-layer, every data memory-code block performs at most one logical operation. Consequently, the number of operations per sub-layer is at most $M$. This serialization reduces to a graph coloring problem, and a full layer of operations can be scheduled within $O(k)$ sub-layers; see Ref.~\cite{Nguyen2024} and Appendix~\ref{app:serialization}. If the time overhead of each sub-layer is $T$, then the total time overhead is $O(kT)$. This introduce a time overhead factor of $O(k) = O(d^a)$, which represents the primary contributor to the overall time complexity. Next, we analyze the time overhead of each sub-layer. 

The physical circuit depth of all logical operations is $O(1)$, assuming access to sufficient gate-teleportation resource states and distilled magic states. The circuits of logical operations are detailed in Appendix~\ref{app:operations}, where logical measurements are implemented using the PCS circuit shown in Fig.~\ref{fig:PCS}, and the parity-check measurements in PCS are realized via the gate-teleportation circuit in Fig.~\ref{fig:teleportation}. Consequently, the time overhead $T$ is entirely determined by the efficiency of the two state factories. 

Gate-teleportation resource states used in logical operations fall into two categories: those for parity checks of the memory code and those for parity checks of deformed codes. In PCS, a deformed code can simultaneously act on $k_R$ code blocks. Accordingly, to implement up to $M$ operations in a sub-layer, the number of required memory-code resource states is $O(M)$, while that of deformed-code resource states is $O(M/k_R)$. In LTSP, each execution of the preparation circuit produces $k_F$ identical copies of a resource state. For memory-code resource states, the preparation circuit requires $O(k k_F d_S^2)$ physical qubits; for deformed-code resource states, it requires $O(k k_R k_F d_S^2)$ qubits. Here, $d_S = \mathrm{polylog}(d)$ is the surface code distance, which will be justified in Sec.~\ref{sec:LTSP}. In both cases, the circuit depth is $O(d_S)$. Therefore, with $O(Mk)$ physical qubits allocated to the resource-state factory (consider the time-space trade-off), one can generate all gate-teleportation resource states needed for a sub-layer within time $O(d_S^3)$. 

For each sub-layer of logical operations, at most $M$ copies of distilled $T$-gate magic states are required. The preparation of these states proceeds in two stages. In the first stage, we prepare $O(M)$ copies of noisy magic states on distance-$d_S$ surface-code blocks, which are subsequently transferred to memory-code blocks---each copy is transferred from a surface-code block to the first logical qubit of a memory-code block. The noisy-state preparation circuit requires $O(Mk)$ physical qubits and has a depth of $O(1)$, while the necessary resource states are generated in the resource-state factory with a time overhead of $O(d_S^3)$. In the second stage, we distill the noisy magic states. Using a constant-overhead distillation protocol~\cite{Nguyen2024}, we can obtain $M$ copies of distilled magic states from the noisy inputs. The distillation stage incurs a time overhead of $O(d^{o(1)} d_S^3)$, where $d^{o(1)}$ originates from the distillation circuit, and the $d_S^3$ factor arises from the resource-state generation. 

Based on the above analysis, the total time overhead for each sub-layer is $O(d^{o(1)}d_S^3)$ {(the $d^{o(1)}$ factor originates from the magic state distillation protocol~\cite{Nguyen2024}; see Theorem~\ref{the:constantmagic})}, resulting in a time overhead of $O(kd^{o(1)}d_S^3) = O(d^{a+o(1)})$ for an entire layer. Note that $d_S = \mathrm{polylog}(d)$ {(an explicit analytical expression for $d_S$ is provided in Appendix~\ref{app:SC_overhead})}. In addition to executing logical operations, the computation also requires the preparation of $S$-gate magic states and the initialization of data memory-code blocks prior to computation. These tasks can be completed in time $O(kd_S^3)$---the same order as that of an operation layer---given $O(Mk)$ physical qubits. A detailed analysis is provided in Appendix~\ref{app:cost}. 

\section{Low-overhead code surgery}
\label{sec:CS}

In the resource-efficient code surgery protocol, specifically devised sticking, the overall spacetime overhead scales as $O(d^2)$. This overhead arises from two sources. One factor of $d$ comes from the qubit cost of constructing the ancilla system in code surgery. In devised sticking, the ancilla size scales as $O(kd)$ to measure $\Theta(k)$ logical operators simultaneously. The other factor of $d$ reflects the inherent time-efficiency limitation of conventional code surgery, which requires $\Theta(d)$ rounds of parity-check measurements to suppress measurement errors. In this work, we develop an approach that eliminates both $d$-scaling factors, enabling a constant-spacetime-overhead implementation of code surgery. However, it imposes constraints on the operation set, resulting in an overall time overhead of $O(d^{a+o(1)})$ when compiling a general logical circuit. 

We address the qubit cost with PCS. In this approach, code surgery is applied simultaneously across multiple memory blocks using a single ancilla system. Consequently, with an ancilla of size $O(k\mathrm{poly}(d))$, we can operate on $\Theta(k\mathrm{poly}(d))$---the same polynomial---logical qubits in parallel, resulting in a qubit overhead that remains constant. 

We overcome the inherent time-efficiency limitation of code surgery through the combination of gate teleportation and LTSP. The central challenge is to reduce the time overhead without compromising the constant qubit overhead~\cite{Gottesman2014,Cowtan2025}. To this end, we utilize gate teleportation to perform parity-check measurements, and propose generating the required resource states using a classical LTC. Leveraging the local testability property, we demonstrate that the overhead from state preparation contributes only an almost constant factor to the qubit cost, while reducing the time overhead of code surgery to almost constant. 

Although the $d^2$ factor is eliminated, the use of PCS and LTSP imposes constraints on the compilation of logical circuits. In particular, only an operation set of size $O(1)$ can be applied to each memory block, similar to the CC+GT protocol using state distillation~\cite{Nguyen2024}. This limited operation set restricts intra-block parallelization and necessitates serialization of logical operations. As a result, the overall time overhead scales as $O(kd^{o(1)}) = O(d^{a+o(1)})$. 

\subsection{Parallelized code surgery}
\label{sec:PCS}

\begin{figure}[htbp]
\centering
\includegraphics[width=\linewidth]{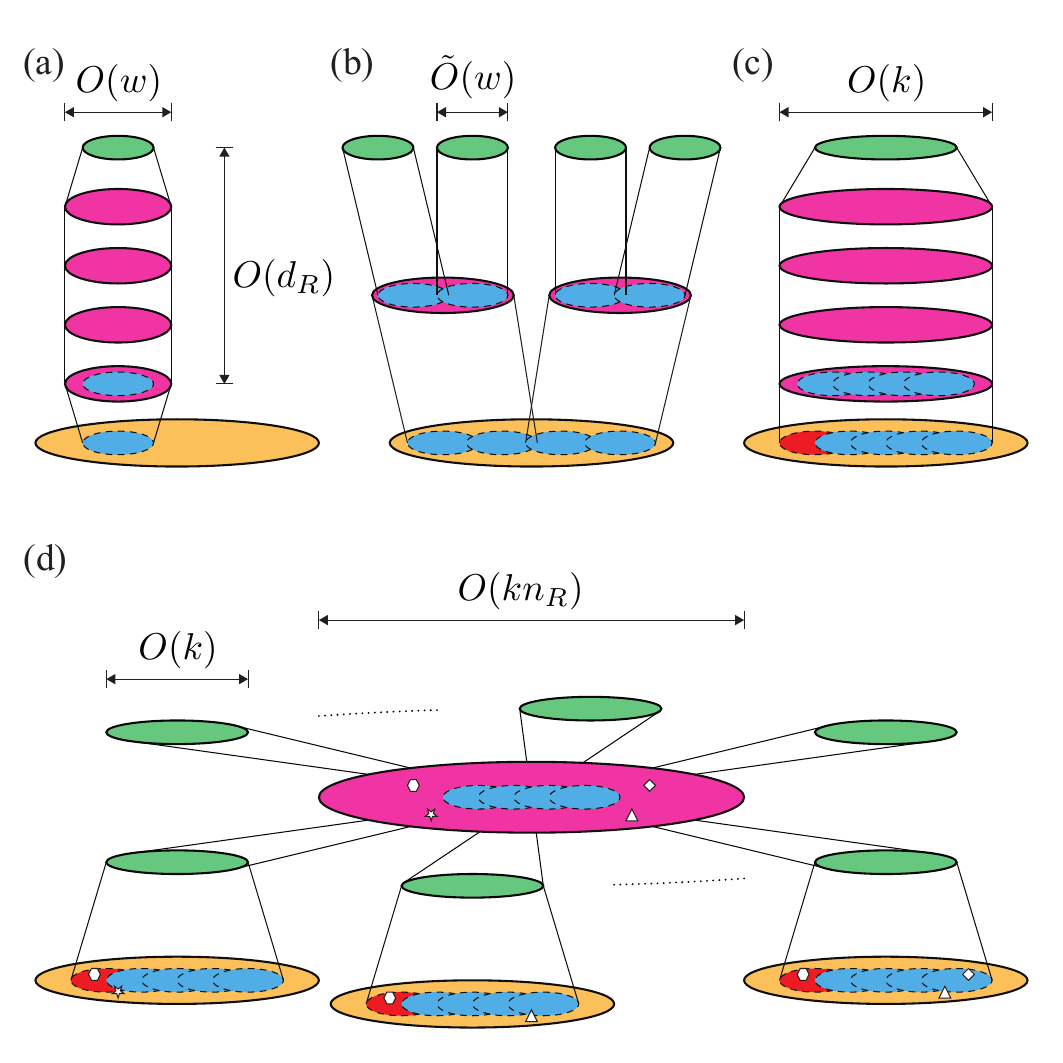}
\caption{
Code surgery schemes. Each orange circle represents a memory block. Magenta and green circles denote two types of subsystems that constitute the ancilla system. Dashed cyan circles indicate logical operators to be measured.
(a) CKBB protocol. The ancilla system comprises multiple subsystems and is coupled to the support of the target logical operator. 
(b) GM+BFB protocol. The ancilla system consists of two components: branch-sticker subsystems (magenta), which directly couple to memory blocks to isolate logical operators, and gauging-measurement subsystems (green), which couple to the branch stickers to facilitate readout. Although only one layer of branch stickers is shown, measuring $q$ logical operators in parallel requires $O(\log d)$ layers in practice. 
(c) Devised sticking. The ancilla system is coupled to a subset of the memory block that may overlap with logical operators beyond the targeted ones. By appropriately designing the glue code, any subset of logical operators can be selectively measured, while unmeasured logical operators (dashed red circles) are preserved. 
(d) Parallelized code surgery. Logical operators on multiple memory blocks are measured simultaneously. The ancilla system comprises two parts: the large magenta circle corresponds to the second columns of $H^D_X$ and $H^D_Z$ in Eqs.~(\ref{eq:HDX2})~and~(\ref{eq:HDZ2}); the green circles correspond to the third columns. There are $n_R$ green-circle subsystems in total, with $k_R$ of them coupled to memory blocks. Errors on the magenta subsystem (indicated by white-face markers) are equivalent to errors on the memory blocks. 
}
\label{fig:stickers}
\end{figure}

Before presenting our scheme, we first review the qubit overhead in several existing code surgery protocols; see Fig.~\ref{fig:stickers}. In the CKBB protocol, the ancilla system is constructed as a hypergraph product code, which is a product of two linear codes, referred to as the glue code and the R code throughout this paper. To measure a logical operator via code surgery, a layer of qubits along the glue-code direction [horizontal direction in Fig.~\ref{fig:stickers}(a)] is coupled to the support of the operator to extract its eigenvalue. Consequently, the glue code has size $\Theta(w)$, where $w$ is the weight of the logical operator. The R code is chosen as a repetition code of length $\Theta(d)$ to suppress errors in the ancilla, thereby setting the size of the ancilla in the other direction [vertical direction in Fig.~\ref{fig:stickers}(a)]. Together, these yield an overall ancilla size of $O(wd)$ per measured logical operator. In the gauging measurement protocol, the R-code dimension is minimized by using an expander-based construction that controls errors through graph-theoretic properties, reducing the ancilla size to $O(w^{1+o(1)})$ per measured operator. In devised sticking, the ancilla is also constructed as a hypergraph product code. However, the glue-code layer is coupled to multiple logical operators simultaneously. By carefully designing the glue code, the ancilla can be programmed to measure an arbitrary subset of logical operators. When measuring up to $k$ logical operators in parallel, the glue code has size $O(k)$, resulting in an overall ancilla size of $O(kd)$. Thus, the qubit overhead per logical operator is reduced to $O(d)$. 

In our approach, we further reduce the qubit overhead by coupling a single ancilla system to multiple memory blocks; see Fig.~\ref{fig:stickers}(d). Specifically, we select the R code to be a constant-rate LDPC code (see Table~\ref{tab:codes}). 
We require the R code to have distance $d_R = \Omega\left(n_R^{1/a_R}\right)$ for some finite constant $a_R$, ensuring that its block length satisfies $n_R = \mathrm{poly}(d)$ when $d_R = \Theta(d)$. Each independent codeword of the R code functions as a dedicated channel for reading out the eigenvalues from a corresponding memory block. Consequently, a R code of length $n_R$ enables simultaneous measurement of $\Theta(n_R)$ memory blocks and $\Theta(kn_R)$ logical operators. This construction reduces the eventual qubit overhead per logical operator to $O(1)$. 

The idea of coupling an ancilla system to the memory multiple times may have broader applications. It has been discussed as a potential strategy for measuring high-weight logical operators~\footnote{In a discussion with Armands Strikis (unpublished data, 2024)}. This technique may also offer a pathway to reducing qubit overhead when measuring multiple operators within a single code block. However, such applications remain underdeveloped due to the absence of a rigorous justification of fault tolerance---specifically, there is no guarantee that the resulting deformed code maintains a well-lower-bounded distance. In contrast, the approach presented in this work provides a systematic framework for constructing the ancilla system and designing its coupling to memory blocks, ensuring a provable code distance while faithfully implementing the desired logical operations. 

We now present the scheme of PCS. Following the standard formalism for CSS codes, we represent the target code by the tuple $(H_X, H_Z, J_X, J_Z)$. This corresponds to the code of the target block on which the logical measurement is applied. The target block could be a single memory-code block or a composite of multiple blocks. To simplify notation, we use $[[n, k, d]]$ to denote the parameters of the target code, although this notation has also been used for the memory code. Here, $H_X$ and $H_Z$ ($J_X$ and $J_Z$) are the check (generator) matrices, representing $X$ and $Z$ stabilizer (logical) operators of the code, respectively. Formal definitions of these matrices and notation can be found in Appendix~\ref{app:notations}. These matrices satisfy the standard conditions: $H_X H_Z^\mathrm{T} = H_X J_Z^\mathrm{T} = H_Z J_X^\mathrm{T} = 0$ and $J_X J_Z^\mathrm{T} = E_k$, where $E_k$ is the $k \times k$ identity matrix. 

In PCS, we perform simultaneous logical measurements on multiple target blocks. On each block, we measure $q$ independent $Z$ logical operators, which can be compactly represented by the matrix $\alpha J_Z$. Here, $\alpha \in \mathbb{F}_2^{q \times k}$ is a full-rank matrix, with each row specifying a $Z$ logical operator. Given a R code with logical dimension $k_R$, the same single-block measurement specified by $\alpha J_Z$ can be applied in parallel to $k_R$ target blocks. While we focus on single-block measurements of $Z$ logical operators, the generalization to multi-block measurements and $X$ logical operators is straightforward. 


To realize the parallelized measurement, we employ a deformed code defined by the following check matrices: 
\begin{eqnarray}
H^D_X &=& \left(\begin{array}{ccc}
E_{k_R} \otimes H_X & 0 & {G_R^\mathrm{r}}^\mathrm{T} \otimes T \\
0 & E_{r_R} \otimes H_G & H_R \otimes E_{r_G}
\end{array}\right)
\label{eq:HDX2}
\end{eqnarray}
and 
\begin{eqnarray}
H^D_Z &=& \left(\begin{array}{ccc}
E_{k_R} \otimes H_Z & 0 & 0 \\
G_R^\mathrm{r} \otimes S & H_R^\mathrm{T} \otimes E_{n_G} & E_{n_R} \otimes H_G^\mathrm{T}
\end{array}\right).
\label{eq:HDZ2}
\end{eqnarray}
Here, $r_X$ is the number of rows in the matrix $H_X$, and the superscript `r' denotes the right inverse of a matrix. The matrices $H_R$, $G_R$, $H_G$, $S$, and $T$ will be explained later. In each check matrix, the first column corresponds to the $k_R$ target blocks, while the second and third columns represent the ancilla system. In the code surgery procedure, the ancilla-system qubits are initialized in the $\ket{+}$ state. Parity-check measurements associated with the deformed code are then performed, followed by $X$-basis measurements on the ancilla-system qubits. This protocol implements the desired $Z$ logical measurements across the target blocks. 

The matrices $H_R$ and $G_R$ define the R code, where $H_R$ is the check matrix and $G_R$ is the generator matrix. We choose $G_R$ in its standard form, so that its right inverse takes the form ${G_R^\mathrm{r}}^{\mathrm{T}} = \left(\begin{array}{cc} E_{k_R} & 0 \end{array}\right)$, which is sparse. The ancilla system used in the protocol is a hypergraph product code generated from the R code and the glue code, the latter specified by the check matrix $H_G$. The $X$ and $Z$ check matrices of the hypergraph product code are given by the last two columns in the second row of $H^D_X$ and $H^D_Z$, respectively. 

The matrices $H_G$, $S$, and $T$ are constructed following the framework of devised sticking~\cite{Zhang2025}, such that they satisfy the following conditions: i) $H_X S^\mathrm{T} = T H_G$; ii) there exists a matrix $R$ such that $\alpha J_Z R S = \alpha J_Z$ and $H_G (\alpha J_Z R)^\mathrm{T} = 0$; and iii) there exists a matrix $\beta$ such that $\alpha_\perp J_X S^\mathrm{T} = \beta H_G$. Here, $\alpha_\perp \in \mathbb{F}_2^{(k - q) \times k}$ is a full-rank matrix satisfying $\alpha_\perp \alpha^\mathrm{T} = 0$. These conditions ensure that the ancilla system is correctly programmed and coupled to the target blocks: the intended logical operators are measured, while the unmeasured operators are preserved throughout the code surgery procedure. Furthermore, the matrices $H_G$, $S$, and $T$ are constructed to be sparse: their row and column weights are bounded above by a constant, ensuring that the resulting deformed code remains a qLDPC code. In particular, $S$ is designed to have both row and column weights exactly equal to one, a property that contributes to maintaining a favorable distance in the deformed code. Finally, when the generator matrices $J_X$ and $J_Z$ are expressed in standard form, and assuming that the logical operators selected for measurement act on mutually disjoint sets of logical qubits (i.e.~the logical thickness is one), the number of rows and columns in $H_G$ scales as $O(k)$. This ensures that the qubit overhead of code surgery is constant. 

We can measure the target logical operators because they belong to the stabilizer group of the deformed code. This is demonstrated by the relation 
\begin{eqnarray}
\left(\begin{array}{ccc} E_{k_R}\otimes(\alpha J_Z) & 0 & 0 \end{array}\right) &=& \left(\begin{array}{cc} 0 & G_R\otimes(\alpha J_Z R) \end{array}\right) H^D_Z.
\end{eqnarray}
This equation also specifies how to extract the eigenvalues of the target logical operators from the measurement outcomes of the stabilizers. 

Next, we sketch the proof for a lower bound on the distance of the deformed code. The key idea is to use the technique that establishes an equivalence between errors on the ancilla system and errors on the target blocks; see Fig.~\ref{fig:stickers}(d). On each target block, the weight of the equivalent error is upper bounded by the weight of the original error on the ancilla system. Consequently, the distance of the deformed code is lower bounded by the distance of the target code. We note, however, that an error on the ancilla system may be equivalent to multiple distinct errors on different target blocks, resulting in error amplification. In the present protocol, such amplification is tolerable because error correction is performed independently on each target block. Finally, we emphasize that the validity of the error equivalence argument assumes that the error weight is below the distance of the R code. Therefore, the R code must have sufficiently large distance to ensure fault tolerance. 

\begin{lemma}
Suppose the target code is a qLDPC code with constant encoding rate, and the R code is an LDPC code also with constant encoding rate. Then, the resulting deformed code supports PCS with constant qubit overhead. In particular, the deformed code defined by Eqs.~(\ref{eq:HDX2})~and~(\ref{eq:HDZ2}) is a qLDPC code {whose length is upper bounded by $(k_R+2n_R)n = O(k_R k)$}. This code enables the intended operation: the parallel measurement of an identical set of up to $k$ logical operators, specified by $\alpha J_Z$, across $k_R$ target blocks. Moreover, if $H_R$ is full-rank, the distance of the deformed code satisfies $d_D \geq \min\{d, d_R\}$. 
\label{lem:PCS}
\end{lemma}

See Appendix~\ref{app:PCS} for the proof. 

\subsection{Locally-testable state preparation}
\label{sec:LTSP}

The inherent time-efficiency limitation of code surgery arises from the need to correct measurement errors. In conventional code surgery, a subset of logical operators is measured by coupling the target block to an ancilla system, resulting in a deformed code whose stabilizer group includes the logical operators to be measured. Eigenvalues of these operators are then extracted by performing parity-check measurements of the deformed code, i.e.~measuring its stabilizer operators. However, due to the presence of measurement errors, these measurements must be repeated for $\Theta(d)$ rounds to ensure accurate error correction and reliable eigenvalue extraction. Consequently, if parity-check measurements could be performed without measurement errors, this would eliminate the $\Theta(d)$ repetition factor and significantly reduce the overall time overhead. 

Parity-check measurements free of measurement errors can be effectively realized through gate teleportation. The key insight is that, in gate teleportation, all errors can be interpreted as acting either on the input or output state, as illustrated in Fig.~\ref{fig:teleportation}. This leads to the following effective picture: the input state may carry some error, followed by a single round of error-free parity-check measurements, after which additional errors may occur on the output state. Although the measurements are effectively error-free in this picture, it remains crucial to ensure that the effective errors on both the input and output states are low-weight and therefore correctable. In particular, this requires that the resource state used in gate teleportation contains only low-weight errors. 

Errors in the resource state can be managed either by using a high-distance concatenated code or by applying state distillation~\cite{Gottesman2014,Yamasaki2024,Tamiya2024,Nguyen2024}. While high-distance concatenated codes incur substantial spacetime overhead, constant-overhead state distillation requires the quantum memory to be encoded in a good qLTC. Since our objective is a general method of reducing spacetime overhead for qLDPC codes, an alternative approach is necessary. Furthermore, even if the memory were encoded in a good qLTC, constructing a deformed code that also satisfies local testability is a highly nontrivial task, further limiting the practical applicability of constant-overhead state distillation in our context. 

Our approach, LTSP, offers a generic method to control measurement errors in qLDPC-code quantum error correction, while requiring only almost constant overhead; see Fig.~\ref{fig:LTSP}. Central to this approach is generating resource states using a classical LTC, which we refer to as the F code because it serves as the core component of the resource state factory (see Table~\ref{tab:codes}). In our scheme, each resource state is dedicated to measuring only $Z$ stabilizer operators---$X$ stabilizer operators can be measured by rotating the basis. When measuring $Z$ stabilizers, we simulate the corresponding measurement circuit by encoding logical qubits in the F code. Specifically, we treat the classical code as a quantum code that protects only $Z$ operators. Then, the measurement circuit is simulated with transversal operations. This simulation produces an encoded resource state, which can be efficiently transformed into the eventual unencoded resource state. Our method reliably produces resource states meeting the requirement that errors remain low-weight, as summarized in Lemma~\ref{lem:LTSP}. Importantly, this approach is broadly applicable to any qLDPC code. 

\begin{lemma}
Consider the generation of gate-teleportation resource states for measuring $Z$ stabilizer operators of a deformed code. Suppose the F code is an LDPC code with constant soundness. Let $\vert \bar{e}^{sp}_X \vert$ and $\vert \bar{e}^{sp}_Z \vert$ denote the weights of $X$ and $Z$ errors occurring during the state preparation circuit, respectively. Let $\vert e^{RS}_X \vert$ and $\vert e^{RS}_Z \vert$ denote the weights of $X$ and $Z$ errors present on one copy of the prepared resource state. If these errors are undetectable and satisfy $\vert \bar{e}^{sp}_X \vert < d_F / C_1$, then the following bounds hold: $\vert e^{RS}_X \vert \leq C_2 \vert \bar{e}^{sp}_X \vert$ and $\vert e^{RS}_Z \vert \leq C_3 \vert \bar{e}^{sp}_Z \vert$, where $C_1$, $C_2$, and $C_3$ are constants determined by the weight and soundness parameters of the involved codes. The state preparation circuit uses $O(n_D n_F d_S^2)$ physical qubits, has depth $O(d_S)$, and produces $n_F$ copies of the resource state. While $d_S$ can be arbitrary positive integer, in practice it is chosen to scale as $\mathrm{polylog}(n_D n_F)$. These conclusions naturally generalize from deformed codes to arbitrary CSS-type qLDPC codes. 
\label{lem:LTSP}
\end{lemma}

See Appendix~\ref{app:LTSP} for a detailed explanation of the protocol and the corresponding proof. {Explicit expressions for constants $C_1$, $C_2$, and $C_3$ are provided in Lemma~\ref{lem:sp}.}

{Now, we explain why errors on the resource state remain low-weight. Recall that our goal is to measure stabilizer operators of a qLDPC code (the deformed code). Consequently, Pauli errors generated directly by circuit operations remain low-weight. The challenge arises in the feedback operations: to prepare the resource state within the logical subspace of the deformed code, we apply Pauli gates conditioned on measurement outcomes identifying the subspace sectors. However, these measurements can be faulty. Since these feedback Pauli gates may have large weight, even a low-weight measurement error can propagate into a high-weight Pauli error on the output resource state. 

We address this vulnerability by protecting the stabilizer-measurement circuit using the F code. We treat the F code as a quantum CSS code that only protects against $X$ errors (assuming the circuit measures $Z$ stabilizers). Then, each qubit in the circuit is encoded as a logical qubit of the F code. Operations within the circuit are then implemented as transversal logical operations on the F code. This protection ensures that measurement errors---and the subsequent propagation into high-weight Pauli errors---are suppressed. As established in Lemma~\ref{lem:LTSP}, as long as the weight of the physical errors in the circuit is bounded by $\Theta(d_F)$ (specifically, $|\bar{e}_X^{sp}| < d_F / C_1$), measurement errors do not occur. This prevents the triggering of dense feedback and ensures that the errors on the output resource state remain bounded and proportional to the weight of the physical errors. 

However, introducing an additional level of encoding could potentially increase the circuit depth. Standard implementations of error correction for the F code would typically require $O(d_F)$ rounds of stabilizer measurements, thereby amplifying the circuit depth by a factor proportional to the distance. We circumvent this bottleneck by selecting a constant-soundness LTC as the F code. Such codes are characterized by their \textit{single-shot} property and have linear soundness functions~\cite{Campbell2019}. Due to this single-shot nature, a single round of stabilizer measurement is sufficient to achieve the required error suppression, ensuring that the circuit depth remains constant. }


Next, we analyze the qubit costs. To reduce the qubit overhead for generating resource states, we propose using a double-constant classical LTC with constant encoding rate and constant soundness~\cite{Leverrier2022LTC,Panteleev2022,lin2022c3LTC}. Similar to the R code, we require the F code to have distance $d_F = \Omega\left(n_F^{1/a_F}\right)$ for some finite constant $a_F$, ensuring that its block length satisfies $n_F = \mathrm{poly}(d)$ when $d_F = \Theta(d)$. Because the encoding rate is constant, we can produce $\Theta(n_F)$ copies of the resource state simultaneously, ensuring constant qubit overhead. 

Finally, we address a critical issue neglected in the above discussion. Because the F code is classical, it can only protect against one type of error, allowing the other type to accumulate. In fact, it cannot correct either error type to the fault-tolerant level up to our circuit design. This error accumulation typically results in a vanishing fault-tolerant threshold. To overcome this, we further encode each physical qubit of the F code into a logical qubit of a low-distance surface code (or any other low-distance code supporting a certain logical operation set). The measurement circuit has size $O(n_D)$, where $n_D$ is the length of the deformed code. When encoding this circuit with the F code, the circuit size becomes $O(n_D n_F)$. To control error accumulation, it suffices to choose a surface code distance $d_S = \mathrm{polylog}(n_D n_F)$. Incorporating the surface code increases the circuit size to $O(n_D n_F \mathrm{poly}(d_S))$, thus contributing only polylogarithmic overhead factors to the qubit overhead and circuit depth. 

\section{Discussions}

We have presented a general scheme for FTQC on qLDPC codes that combines code surgery with gate teleportation and introduces optimizations in deformed-code construction and resource-state preparation. Our approach achieves constant qubit overhead while significantly reducing time overhead. For all constant-rate qLDPC codes with the code distance grows faster than HGP codes, the time overhead is lower than that of the GM+BFB protocol; for good qLDPC codes, it reaches the $O(d^{1+o(1)})$ scaling. These results improve the asymptotic resource efficiency of FTQC.

{While the primary objective of this work is to optimize scaling behavior in the limit of large code distance, the techniques developed for this purpose---specifically PCS and LTSP---are applicable for lowering the time overhead in near-term, intermediate-distance implementations, offering promising avenues for future investigation (see Appendix~\ref{app:performance}). PCS is uniquely suited for implementing identical logical operations across multiple code blocks, which arise frequently in practical applications such as the simulation of many-body quantum systems with translational invariance. In these scenarios, PCS enables the simultaneous manipulation of all logical qubits with constant qubit overhead, even for code families like HGP codes, where the asymptotic advantage of our protocol over the GM+BFB method for general heterogeneous circuits vanishes. The resource efficiency of LTSP is governed by the parameters of the chosen LTC, underscoring the importance of identifying code instances with optimized parameters at moderate distances to fully realize the performance gains. Notably, LTSP is a versatile framework; for example, utilizing a repetition code with an expander-graph-based check matrix as the F code provides a mechanism for trading qubit overhead for enhanced temporal efficiency. Overall, compared with existing approaches, our protocol offers distinct advantages in scenarios involving large code distances, highly symmetric gate patterns, or time-critical applications.
}

\begin{acknowledgments}
This work is supported by the National Natural Science Foundation of China (Grant Nos. 12225507, 12088101) and NSAF (Grant No. U1930403). 
\end{acknowledgments}

\appendix

\begin{widetext}

\section{Preliminaries}

\subsection{Notations}
\label{app:notations}

We use $\vert \bullet \vert$ and $\Vert \bullet \Vert$ to denote the Hamming weight and the matrix norm induced by Hamming weight, respectively, i.e.~the norm of a matrix $A \in \mathbb{F}_2^{m \times n}$ is 
\begin{eqnarray}
\Vert A \Vert = \max\left\{ \frac{\vert uA \vert}{\vert u \vert} \st u \in \mathbb{F}_2^{m}-\{0\} \right\}.
\end{eqnarray}

We denote by $A_{j,\bullet}$ and $A_{\bullet,j}$ the $j$-th row and $j$-th column of the matrix $A$, respectively, and we denote by $A^\mathrm{r}$ the right inverse. 

Let $b = (b_1,b_2,\ldots,b_n) \in \mathbb{F}_2^n$. We use 
\begin{eqnarray}
\sigma(b) &=& \sigma^{b_1} \otimes \sigma^{b_2} \otimes \cdots \otimes \sigma^{b_n}
\end{eqnarray}
to represent $n$-qubit $\sigma = X,Z$ Pauli operators. When $B \in \mathbb{F}_2^{m \times n}$, 
\begin{eqnarray}
\sigma(B) &=& \left(\begin{array}{c}
\sigma(B_{1,\bullet}) \\
\sigma(B_{2,\bullet}) \\
\vdots \\
\sigma(B_{m,\bullet})
\end{array}\right)
\end{eqnarray}
is a list of Pauli operators. 

An $[[n,k,d]]$ CSS code is represented by four matrices $H_X \in \mathbb{F}_2^{r_X \times n}$, $H_Z \in \mathbb{F}_2^{r_Z \times n}$ and $J_X,J_Z \in \mathbb{F}_2^{k \times n}$ that satisfy conditions $H_X H_Z^\mathrm{T} = H_X J_Z^\mathrm{T} = J_X H_Z^\mathrm{T} = 0$, $J_X J_Z^\mathrm{T} = E_k$. The stabilizer of the code is generated by operators $X(H_X)$ and $Z(H_Z)$. The $X$ and $Z$ logical operators of the $j$-th logical qubit are $X((J_X)_{j,\bullet})$ and $Z((J_Z)_{j,\bullet})$, respectively. 

We use $(H^c_X,H^c_Z,J^c_X,J^c_Z)$ with $c = M,D,S,C$ to represent the memory, deformed, surface, and color codes, respectively. Regarding parameters of a quantum code, we use notations $n_c$, $k_c$, $d_c$, $r^c_X$, and $r^c_Z$. Check and generator matrices of the R and F codes are $H_c \in \mathbb{F}_2^{r_c \times n_c}$ and $G_c \in \mathbb{F}_2^{k_c \times n_c}$ with $c = R,F$, respectively, which satisfy $H_c G_c^\mathrm{T} = 0$; and we also use $d_c$ to denote the distance of a classical linear code. 

\subsection{Primitive operations, spacetime locations, and error models}

\begin{figure}[htbp]
\centering
\includegraphics[width=\linewidth]{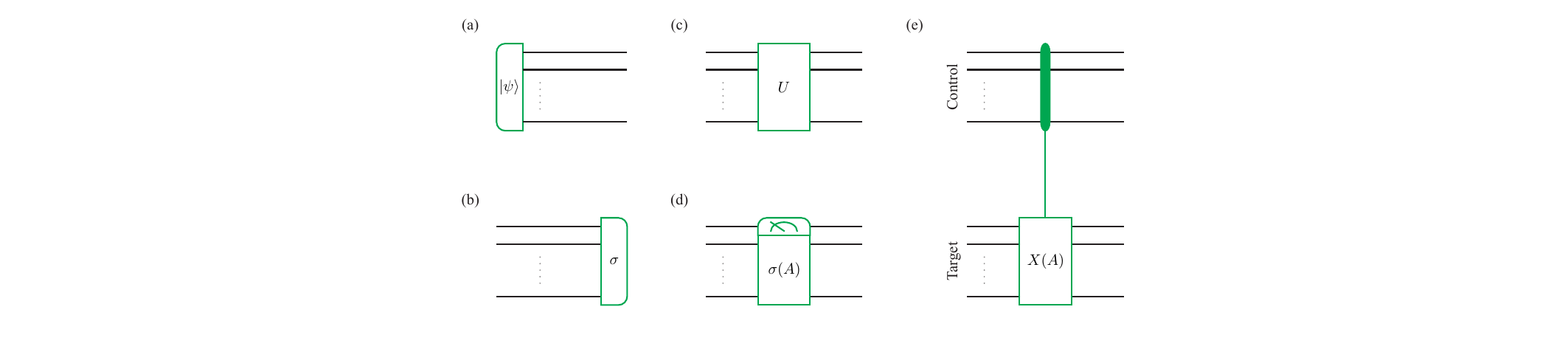}
\caption{
(a) Transversal initialization that prepares each qubit in the state $\ket{\psi}$. 
(b) Transversal measurement that measures each qubit in the basis $\sigma = X,Z$. 
(c) Gate $U$. Here, $U$ is a layer of single-qubit gates, i.e.~$U = U_1 \otimes U_2 \otimes \cdots$, where $U_1,U_2,\ldots$ are single-qubit gates on qubits $1,2,\ldots$, respectively. 
(d) Projective measurement on $n$ qubits that measures Pauli operators $\sigma(A)$, where $A \in \mathbb{F}_2^{m \times n}$. 
(e) Generalized transversal controlled-NOT gate. The control (target) is a set of $m$ ($n$) qubits. The transformation $X(A_{j,\bullet})$ is (not) applied on the target when the $j$-th control qubit is in the state $\ket{1}$ ($\ket{0}$). 
We assume that the matrix $A$ has rows and columns of weight at most $\omega_{max}$, where $\omega_{max}$ is a positive constant. 
}
\label{fig:primitive_operations}
\end{figure}

The primitive operations are illustrated in Fig.~\ref{fig:primitive_operations}. Given a quantum circuit, there are two sets of spacetime locations. Quantum locations are the qubits during time intervals between primitive operations; see Figs.~\ref{fig:state_preparation}, \ref{fig:measurement}, and \ref{fig:lattice_surgery} for examples, in which red boxes represent the quantum locations. Classical locations are outcomes of projective measurements: In the measurement of $\sigma(A)$, the outcome is a vector $\nu \in \mathbb{F}_2^m$, and $(-1)^{\nu_j}$ is the eigenvalue of $\sigma\left(A_{j,\bullet}\right)$, i.e.~each element corresponds to an operator; such a measurement introduces $m$ locations, and each location corresponds to one element in the outcome. 

We assume that errors can occur independently at any location. We now justify the error model. For initializations and gates with errors, each of them can be modeled as an ideal operation with errors occurring at the locations immediately after the operation. For qubit measurements with errors, they can be modeled as ideal measurements with errors occurring at locations immediately before the measurements: these errors effectively flip measurement outcomes. For a projective measurement with errors, it can be modeled as a measurement that projects the state onto the correct subspace with errors occurring in the measurement outcome and at locations immediately after the measurement: the errors at locations after the measurement may induce an effective incorrect projection. Compared with taking spacetime locations in circuits that projective measurements and generalized transversal controlled-NOT gates are decomposed into standard single-qubit and two-qubit operations (see Lemma~\ref{lem:depth_error}), our choice of spacetime locations could overestimate effective circuit distances and the fault-tolerance threshold by a constant factor, which does not change our conclusions. 

\begin{lemma}
Projective measurements and generalized transversal controlled-NOT gates specified by a matrix $A$ can be implemented using a circuit composed of single-qubit and two-qubit operations. If the row and column weights of $A$ are bounded above by a positive constant $\omega_{max}$, then the circuit has depth at most $\omega_{max}+2$. Moreover, an error occurring at a single location in this circuit may propagate to at most $\omega_{\max}$ locations. 
\label{lem:depth_error}
\end{lemma}

\begin{proof}
We only need to consider generalized transversal controlled-NOT gates because projective measurements can be implemented using them. Specifically, to measure the operators $X(A)$, we initialize the control qubits in the $\ket{+}$ state, apply the generalized transversal controlled-NOT gate defined by the matrix $A$, and then measure the control qubits in the $X$ basis. Similarly, to measure $Z(A^\mathrm{T})$, we initialize the target qubits in the $\ket{0}$ state, apply the same gate, and then measure the target qubits in the $Z$ basis. 

A generalized transversal controlled-NOT gate can be implemented by a circuit of depth at most $\omega_{\mathrm{max}}$, since scheduling the controlled-NOT gates corresponds to an edge-coloring problem on a bipartite graph~\cite{West2001,Tremblay2022}. 

Each qubit in the circuit is involved in at most $\omega_{\max}$ controlled-NOT gates. On a control (target) qubit, an $X$ ($Z$) error can propagate to the corresponding target (control) qubits, while $Z$ ($X$) errors do not propagate. Therefore, a single-qubit error may only spread to qubits within the corresponding support set. As a result, the weight of the propagated error is upper bounded by $\omega_{\max}$. 
\end{proof}

\subsection{Locally testable codes}

\begin{definition}
{\bf $(\omega,s)$-locally testable codes~\cite{Leverrier2022LTC}.} A linear code $C \subset \mathbb{F}_2^n$ is called $(\omega,s)$-locally testable if it has a check matrix $H \in \mathbb{F}_2^{r \times n}$ with rows of weight at most $\omega$ such that for any vector $u \in \mathbb{F}_2^n$ we have 
\begin{eqnarray}
\frac{1}{r} \vert H u^\mathrm{T} \vert &\geq& \frac{s}{n} \min_{c \in C} \vert u-c \vert.
\end{eqnarray}
The parameters $\omega$ and $s$ are positive real numbers called the locality and soundness, respectively. 
\end{definition}

\begin{lemma}
If $H \in \mathbb{F}_2^{r \times n}$ is the check matrix of an $(\omega,s)$-locally testable code, for every vector $v^\mathrm{T} \in \mathrm{colsp} H$, there exists a vector $u_\star \in \mathbb{F}_2^n$ such that $v^\mathrm{T} = H u_\star^\mathrm{T}$ and 
\begin{eqnarray}
\vert u_\star \vert &\leq& \frac{n}{rs} \vert v \vert.
\end{eqnarray}
\label{lem:LTC}
\end{lemma}

\begin{proof}
Because $v^\mathrm{T} \in \mathrm{colsp} H$, there exists $w \in \mathbb{F}_2^n$ such that $v^\mathrm{T} = H w^\mathrm{T}$. Then, there exists $c_\star \in \mathrm{ker} H$ such that 
\begin{eqnarray}
\frac{1}{r} \vert v \vert &\geq& \frac{s}{n} \vert w-c_\star \vert.
\end{eqnarray}
The lemma is proved by taking $u_\star = w-c_\star$. 
\end{proof}

\section{Logical operations}
\label{app:operations}

The circuits implementing logical operations are shown in Fig.~\ref{fig:operations}, while the circuits for magic-state preparation are shown in Fig.~\ref{fig:injections}. 

\begin{figure}[htbp]
\centering
\includegraphics[width=\linewidth]{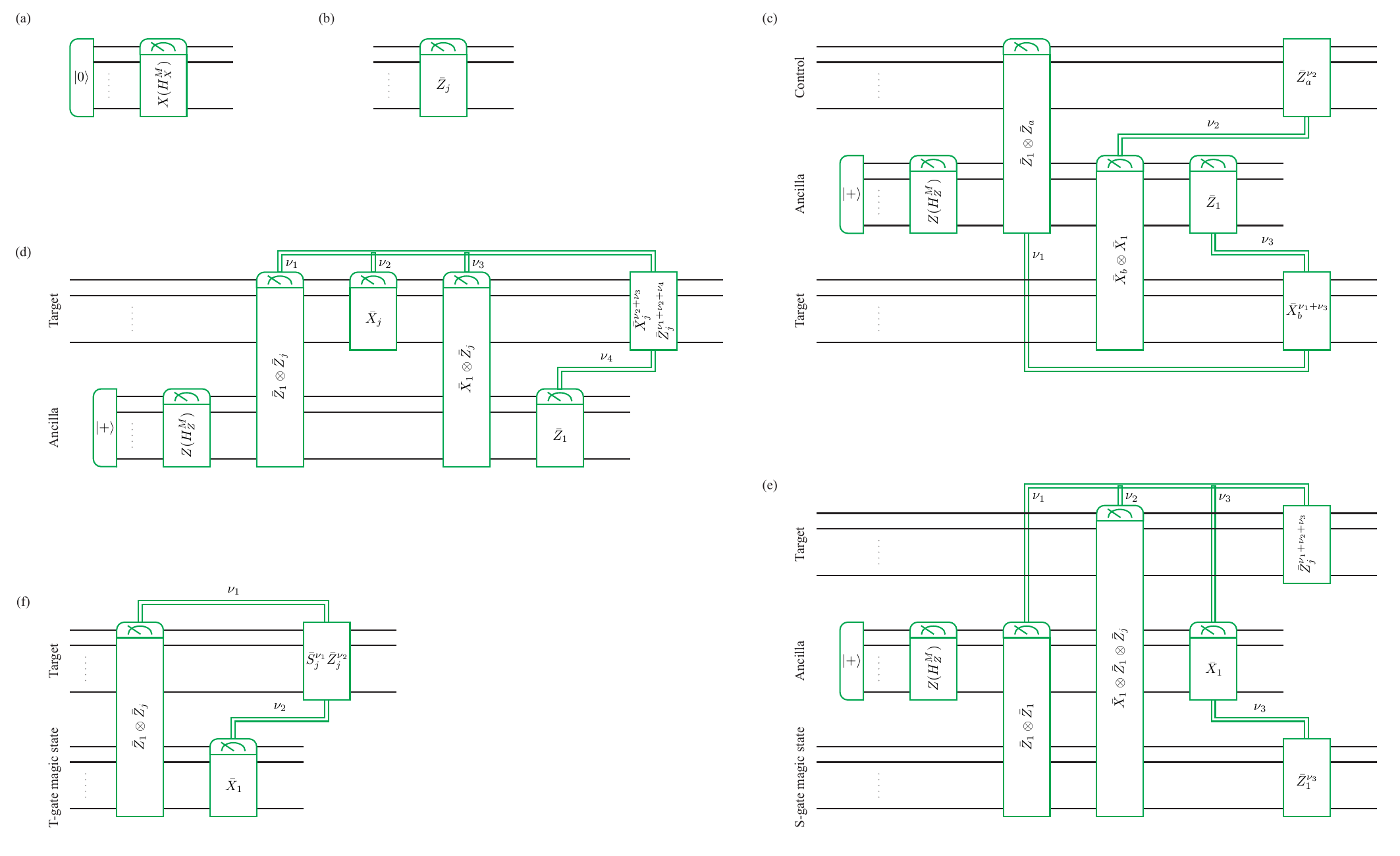}
\caption{
Logical operations. (a) Initialization. (b) Measurement $Mea_j$. (c) Controlled-NOT gate $CNOT_{a,b}$. (d) Hadamard gate $H_j$. (e) $S$ gate $S_j$. (f) $T$ gate $T_j$. Each group of horizontal lines represents a memory-code block. Single-qubit gates are applied to the target memory-code block. Each of the gates $CNOT_{a,b}$, $H_j$, and $S_j$ requires an ancilla memory-code block initialized in the $\ket{+}$ state. Magic states are stored in the first logical qubit of their respective memory-code blocks. 
}
\label{fig:operations}
\end{figure}

\begin{figure}[htbp]
\centering
\includegraphics[width=\linewidth]{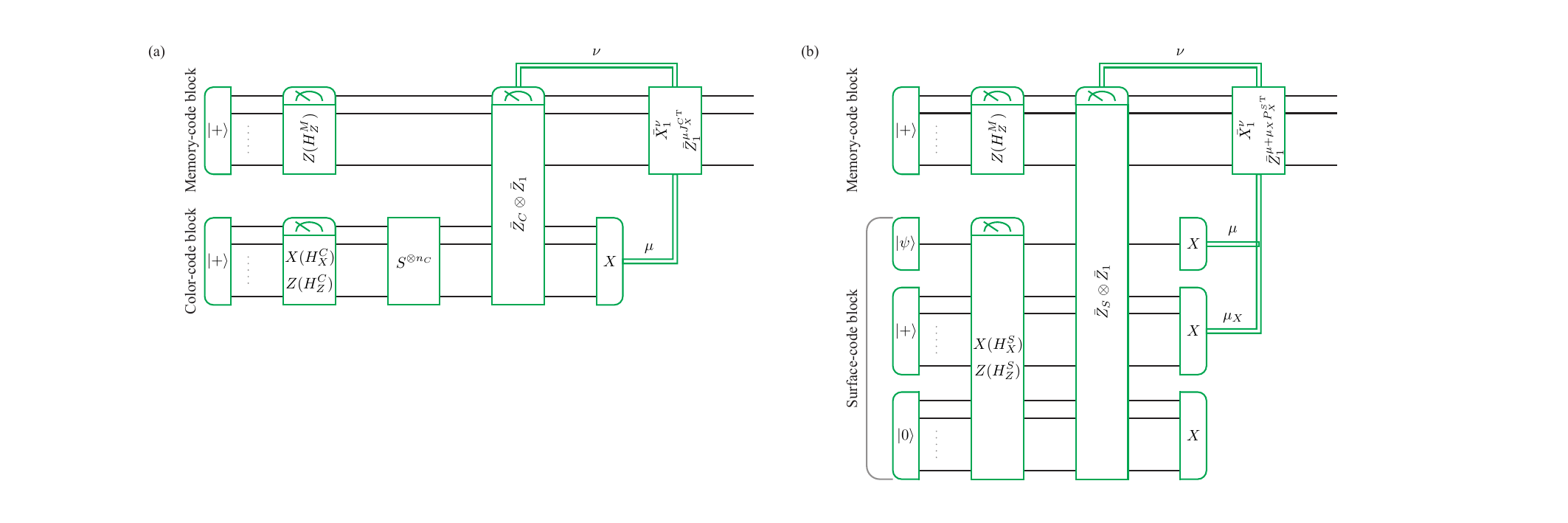}
\caption{
(a) Preparation of the $S$-gate magic state: a fault-tolerant copy is prepared on the first logical qubit of the memory-code block. (b) Preparation of the $T$-gate magic state: a noisy copy is prepared on the first logical qubit of the memory-code block. 
}
\label{fig:injections}
\end{figure}

\section{Serialization}
\label{app:serialization}

We use a pair $(B,F)$ to denote a logical operation: the operation $F \in \{Mea_j, H_j, S_j, T_j, CNOT_{a,b}\}$ applied on the block(s) $B$. Specifically, if $F \in \{Mea_j, H_j, S_j, T_j\}$ is applied on the $u$-th block, then $B = u$; if $F = CNOT_{a,b}$ is applied on the $a$-th logical qubit in the $u$-th block and the $b$-th logical qubit in the $v$-th block, then $B = (u,v)$. 

We define the \emph{qubit support} of the operation $(B,F)$ as 
\begin{eqnarray}
\mathrm{supp}(B,F) =
\begin{cases}
\{(u,j)\}, & \text{if } B = u,\ F \in \{Mea_j, H_j, S_j, T_j\},\\
\{(u,a),(v,b)\}, & \text{if } B = (u,v),\ F = CNOT_{a,b}.
\end{cases}
\end{eqnarray}
Similarly, we define the \emph{block support} of $(B,F)$ as
\begin{eqnarray}
\mathrm{supp}(B) =
\begin{cases}
\{u\}, & \text{if } B = u,\\
\{u,v\}, & \text{if } B = (u,v).
\end{cases}
\end{eqnarray}

Consider an operation set $OS = \{(B_l,F_l) \mid l = 1,2,\ldots,L\}$. We say that $OS$ is \emph{qubit-disjoint} if and only if the qubit supports are disjoint, i.e., 
\begin{eqnarray}
\mathrm{supp}(B_{l_1},F_{l_1}) \cap \mathrm{supp}(B_{l_2},F_{l_2}) = \emptyset
\quad \forall\, l_1 \neq l_2.
\end{eqnarray}
Similarly, $OS$ is \emph{block-disjoint} if and only if the block supports are disjoint, i.e., 
\begin{eqnarray}
\mathrm{supp}(B_{l_1}) \cap \mathrm{supp}(B_{l_2}) = \emptyset
\quad \forall\, l_1 \neq l_2.
\end{eqnarray}

\begin{lemma}
For any qubit-disjoint operation set $OS = \{(B_l,F_l) \st l = 1,2,\ldots,L\}$ (an operation layer) acting on code blocks, each encoding $k$ logical qubits, there exists a partition $P$ of $OS$ such that $\vert P \vert = O(k)$, and every element of $P$ (a sub-layer) is block-disjoint. 
\end{lemma}

\begin{proof}
We construct the partition as follows (see also Ref.~\cite{Tamiya2024,Nguyen2024}). 

First, we map the operation set $OS$ to a multigraph with loops $G = (V, E, r)$. Suppose the operation set acts on $M$ code blocks; then the vertex set is $V = \{1, 2, \ldots, M\}$. The edge set $E = \{1, 2, \ldots, L\}$ has the same cardinality as the operation set, and $r(l) = \mathrm{supp}(B_l)$ specifies the endpoint vertices of the $l$-th edge. 

Next, we find a proper edge coloring of $G$, i.e., a mapping $f: E \to C$, where $C$ is the set of colors, such that any two edges incident to the same vertex receive different colors. Since $OS$ is qubit-disjoint, the maximum degree of $G$ is $k$. Then, the minimum number of colors required for the coloring, the chromatic index of $G$, is $O(k)$~\cite{Nguyen2024,Arjomandi1982}. 

Finally, we construct the desired partition as $P = \{U_c \st c\in C\}$, where $U_c = \{(B_l,F_l) \st f(l) = c\}$. Each $U_c$ is block-disjoint by construction, and $\vert P \vert = O(k)$, completing the proof. 
\end{proof}

\section{Detailed cost analysis}
\label{app:cost}

In this section, we analyze the time overhead under the constraint of a constant qubit overhead. We assume that the number of logical qubits $k$ is small compared to $M$. Specifically, we assume that the memory code distance scales as $d = O(\mathrm{polylog}(Mk))$, which determines both the block length $n$ and the corresponding logical qubit number $k$. This distance is sufficient to achieve a logical error rate per operation $p_L$ satisfying $1/p_L = O(\mathrm{poly}(Mk))$. If the logical circuit size is polynomial, i.e.~$\lvert C\rvert = O(\mathrm{poly}(Mk))$, then this choice of code distance is also sufficient to achieve any target total logical error rate $\epsilon$ satisfying $1/\epsilon = O(\mathrm{poly}(Mk))$. 

\begin{lemma}
Suppose $d,d_R,d_F = O(\mathrm{polylog}(Mk))$. Assuming a sufficient supply of distilled magic states, any block-disjoint operation set (sub-layer) $OS$ can be implemented in time $O(d_S^3)$, requiring $O(Mk)$ physical qubits in the resource-state factory. 
\label{lem:operations}
\end{lemma}

\begin{proof}
Let $\mathrm{num}(F)$ denote the number of occurrences of operation $F \in {Mea_j, H_j, S_j, T_j, CNOT_{a,b}}$ in $OS$. Because $OS$ is block-disjoint, $\sum_F \mathrm{num}(F) = \vert OS \vert \leq M$. We define 
\begin{eqnarray}
\mathrm{batch}(F) = \left\lceil\frac{1}{d_S^2} \left\lceil\frac{1}{k_F} \left\lceil\frac{1}{k_R} \mathrm{num}(F) \right\rceil\right\rceil\right\rceil.
\end{eqnarray}
Since 
\begin{eqnarray}
\mathrm{batch}(F) \leq \frac{1}{d_S^2} \left[\frac{1}{k_F} \left(\frac{1}{k_R} \mathrm{num}(F) + 1 \right) + 1 \right] + 1,
\end{eqnarray}
we obtain 
\begin{eqnarray}
\sum_F \mathrm{batch}(F) \leq \frac{M + \Theta(k^2)(k_R + k_Rk_F + k_Rk_Fd_S^2)}{k_Rk_Fd_S^2},
\end{eqnarray}
where we have used that $\vert {Mea_j, H_j, S_j, T_j, CNOT_{a,b}} \vert = \Theta(k^2)$. 

The operation set $OS$ is implemented in $d_S^2$ steps, where each step contains at most $\mathrm{batch}(F)k_Rk_F$ copies of operation $F$ acting on different code blocks. 

To illustrate, consider $F = Mea_j$. This operation requires two types of resource states: $H^D_Z(\bar{Z}_j)$ and $H^M_X$. For each step, at most $\mathrm{batch}(F)k_F$ copies of $H^D_Z(\bar{Z}_j)$ are needed. These can be generated by running $\mathrm{batch}(F)$ copies of the corresponding LTSP circuit in parallel, requiring at most $\mathrm{batch}(F) \times O(nk_Rk_Fd_S^2)$ physical qubits. Similarly, up to $\mathrm{batch}(F)k_Rk_F$ copies of $H^M_X$ are needed, which can be produced by running $\mathrm{batch}(F)k_R$ copies of the corresponding LTSP circuit in parallel, requiring at most $\mathrm{batch}(F)k_R \times O(nk_Fd_S^2)$ physical qubits. Therefore, the total number of physical qubits required for operation $F$ is $\mathrm{batch}(F)O(nk_Rk_Fd_S^2)$. 

The same scaling applies to all operations, yielding 
\begin{eqnarray}
N_{RSF,1} &\leq & \sum_F \mathrm{batch}(F) O(nk_Rk_Fd_S^2) \\
&\leq & O(Mk) + O(k^3k_Rk_Fd_S^2).
\end{eqnarray}
Assuming $d,d_R,d_F = O(\mathrm{polylog}(Mk))$, we have $k^3k_Rk_Fd_S^2 = O(\mathrm{polylog}(Mk))$, under assumptions on code rate and distance. Consequently, $N_{RSF,1} = O(Mk)$, i.e.,~the total number of physical qubits required in the resource-state factory to implement a block-disjoint operation set scales linearly with the logical qubit count. 

Each step involves running LTSP circuits in parallel, which requires a time overhead of $O(d_S)$. Since there are $d_S^2$ steps, the total time overhead for implementing a block-disjoint operation set is $O(d_S^3)$. 
\end{proof}

\begin{lemma}
Suppose $d,d_R,d_F = O(\mathrm{polylog}(Mk))$ and $d_S = O(\mathrm{polylog}(d))$. Then $M$ copies of distilled $T$-gate magic states can be generated in time $O(d^{o(1)}d_S^3)$, requiring $O(Mk)$ physical qubits in the resource-state and magic-state factories. 
\end{lemma}

\begin{proof}
The generation of distilled magic states consists of two stages. First, we prepare noisy magic states, and then we distill them. 

The circuit for noisy magic-state preparation is shown in Fig.~\ref{fig:injections}(b). We begin by preparing a noisy magic state on a surface-code block, which is subsequently transferred to the first logical qubit of a memory-code block. Using the constant spacetime-overhead distillation protocol, to obtain $M$ copies of distilled magic states, we require $O(M)$ copies of noisy magic states~\cite{Li2014AMS}. Consequently, $O(M)$ surface-code blocks and $O(M)$ memory-code blocks are needed to store these states, corresponding to $O(Md_S^2)$ and $O(Mk)$ physical qubits in the magic-state factory, respectively. 

The $M' = O(M)$ copies of noisy magic states are prepared over $d_S^2$ steps, where each step produces at most $m'k_Rk_F$ copies, with 
\begin{eqnarray}
m' = \left\lceil\frac{1}{d_S^2} \left\lceil\frac{1}{k_F} \left\lceil\frac{1}{k_R} M' \right\rceil\right\rceil\right\rceil.
\end{eqnarray}
We require four types of gate-teleportation resource states: $H^M_Z$, $H^S_X$, $H^S_Z$, and $H^D_Z(\bar{Z}_S \otimes \bar{Z}_1)$. In each step, up to $m'k_Rk_F$ ($m'k_Rk_F$, $m'k_Rk_F$, and $m'k_F$) copies of $H^M_Z$ ($H^S_X$, $H^S_Z$, and $H^D_Z(\bar{Z}_S \otimes \bar{Z}_1)$) are required. These resource states can be generated by running $m'k_R$ ($m'k_R$, $m'k_R$, and $m'$) copies of the corresponding LTSP circuits in parallel, which demand at most
$m'k_R \times O(nk_Fd_S^2)$ ($m'k_R \times O(d_S^2k_Fd_S^2)$, $m'k_R \times O(d_S^2k_Fd_S^2)$, and $m' \times O(nk_Rk_Fd_S^2)$) physical qubits. Since $d_S = O(\mathrm{polylog}(n))$, the total number of physical qubits is $O(m'nk_Rk_Fd_S^2) = O(M'n) = O(Mk)$. The time cost of each step is $O(d_S)$, therefore, the overall time overhead of preparing noisy magic states is $O(d_S^3)$. 

To prepare $M$ copies, the distillation circuit requires $O(M)$ logical qubits, has a depth of $O(\mathrm{polylog}(d))$, and is repeated $O(d^{o(1)})$ times~\cite{Nguyen2024}. We propose to utilize only the first logical qubit of each memory-code block for distillation. Consequently, $O(M)$ memory-code blocks are required in the magic-state factory, corresponding to $O(Mk)$ physical qubits. 

Each layer of the distillation circuit constitutes a block-disjoint operation set, as operations are restricted to the first logical qubit of each block. Therefore, the entire circuit can be executed in time $O(d^{o(1)}d_S^3)$, requiring $O(Mk)$ physical qubits in the resource-state factory (Lemma~\ref{lem:operations}). 

Note that the implementation of the distillation circuits also requires ancillary memory-code blocks; however, this contributes only a constant factor to the overall qubit overhead. 
\end{proof}

\begin{lemma}
Suppose $d,d_R,d_F = O(\mathrm{polylog}(Mk))$, $d_C = O(d)$, and $d_S = O(\mathrm{polylog}(d))$. The quantum computer can be initialized in time $O(kd_S^3)$, requiring $O(Mk)$ physical qubits. 
\end{lemma}

\begin{proof}
There are two tasks required to initialize the quantum computer: preparing $M$ copies of high-fidelity $S$-gate magic states and storing them in $M$ memory-code blocks---each stored in the first logical qubit of a block; and initializing the $Mk$ logical qubits in the data memory-code blocks in the state $\ket{0}$. 

As shown in Fig.~\ref{fig:injections}(a), a high-fidelity $S$-gate magic state is first prepared on a color-code block, which is subsequently transferred to the first logical qubit of a memory-code block. We prepare the $M$ copies in $k$ steps, with at most $M'' = \lceil M/k \rceil$ copies prepared in each step. Thus, we require $M''$ color-code blocks, corresponding to at most $M'' \times O(d_C^2) = O(Mk)$ physical qubits, noting that $d_C = O(k)$. 

As shown in Fig.~\ref{fig:injections}(a), a high-fidelity $S$-gate magic state is first prepared on a color code block, which is subsequently transferred to the first logical qubit of a memory-code block. We prepare the $M$ copies in $k$ steps with at most $M'' = \lceil M/k \rceil\}$ copies prepared in each step; then, we need $M''$ color-code blocks, i.e.~at most $M'' \times O(d_C^2) = O(Mk)$ physical qubits. Note that $d_C = O(k)$. 

In each step, we prepare the $M''$ copies over $d_S^2$ sub-steps, where each sub-step produces at most $m''k_Rk_F$ copies, with 
\begin{eqnarray}
m'' = \left\lceil\frac{1}{d_S^2} \left\lceil\frac{1}{k_F} \left\lceil\frac{1}{k_R} M'' \right\rceil\right\rceil\right\rceil.
\end{eqnarray}
Similar to the preparation of noisy $T$-gate magic states, we require four types of gate-teleportation resource states: $H^M_Z$, $H^C_X$, $H^C_Z$, and $H^D_Z(\bar{Z}_C \otimes \bar{Z}_1)$. In each step, up to $m''k_Rk_F$ ($m''k_Rk_F$, $m''k_Rk_F$, and $m''k_F$) copies of $H^M_Z$ ($H^C_X$, $H^C_Z$, and $H^D_Z(\bar{Z}_C \otimes \bar{Z}_1)$) are required. These resource states can be generated by running $m''k_R$ ($m''k_R$, $m''k_R$, and $m''$) copies of the corresponding LTSP circuits in parallel, requiring at most
$m''k_R \times O(nk_Fd_S^2)$ ($m''k_R \times O(d_C^2k_Fd_S^2)$, $m''k_R \times O(d_C^2k_Fd_S^2)$, and $m'' \times O((n + d_C^2)k_Rk_Fd_S^2)$) physical qubits. Since $d_S = O(\mathrm{polylog}(n))$, the total number of physical qubits is $O(m''n^2k_Rk_Fd_S^2) = O(M''n^2) = O(Mk)$, noting that $d_C = O(n)$. 

According to the above analysis, the overall time overhead for preparing high-fidelity $S$-gate magic states is $O(kd_S^3)$. 

The initialization of the data memory-code blocks is similar to a block-disjoint operation set and can therefore be accomplished in time $O(d_S^3)$ with $O(Mk)$ physical qubits. 
\end{proof}

{\section{Implementation on almost-constant-degree qubit networks}}
\label{app:almost_constant_degree}

In this section, we adapt our protocol to a physical qubit network characterized by an almost constant vertex degree:
\begin{itemize}
    \item \textbf{Physical qubit network:} The physical network is modeled as a graph where each vertex represents a physical qubit. Two-qubit physical gates can be implemented between any pair of qubits connected by an edge. The maximum vertex degree of this graph is $D_{P} = O(d^{o(1)})$, meaning that each physical qubit is coupled to at most an almost constant number of neighbors.
    
    \item \textbf{Logical qubit network:} On this physical network, we realize a fault-tolerant quantum computer defined by a \textit{logical coupling graph}. On this graph, a memory-code block is assigned to each vertex. Inter-block logical operations are supported between any two blocks connected by an edge. The maximum vertex degree of the logical graph is strictly bounded by a constant, $D_L = \Theta(1)$.
    
    \item \textbf{Qubit overhead:} Under these connectivity constraints, the protocol maintains an almost constant qubit overhead of $O(d^{o(1)})$.
    
    \item \textbf{Time overhead and logical operations:} The time overhead remains $O(d^{a+o(1)})$ for the implementation of logical circuits consisting of the following operations:
    \begin{enumerate}
        \item \textbf{Single-qubit operations:} These include state initialization and measurement in the $X$ or $Z$ basis, single-qubit Clifford gates, and the preparation of noisy magic states for each logical qubit.
        \item \textbf{Intra-block operations:} Two-qubit gates can be implemented between any arbitrary pair of logical qubits residing within the same memory-code block.
        \item \textbf{Inter-block operations:} Two-qubit gates can be implemented between logical qubits belonging to distinct blocks, provided those blocks are adjacent on the logical graph.
    \end{enumerate}
\end{itemize}

\subsection{Operation diversity}
\label{app:challenge}

\begin{figure}
    \centering
    \includegraphics[width=\linewidth]{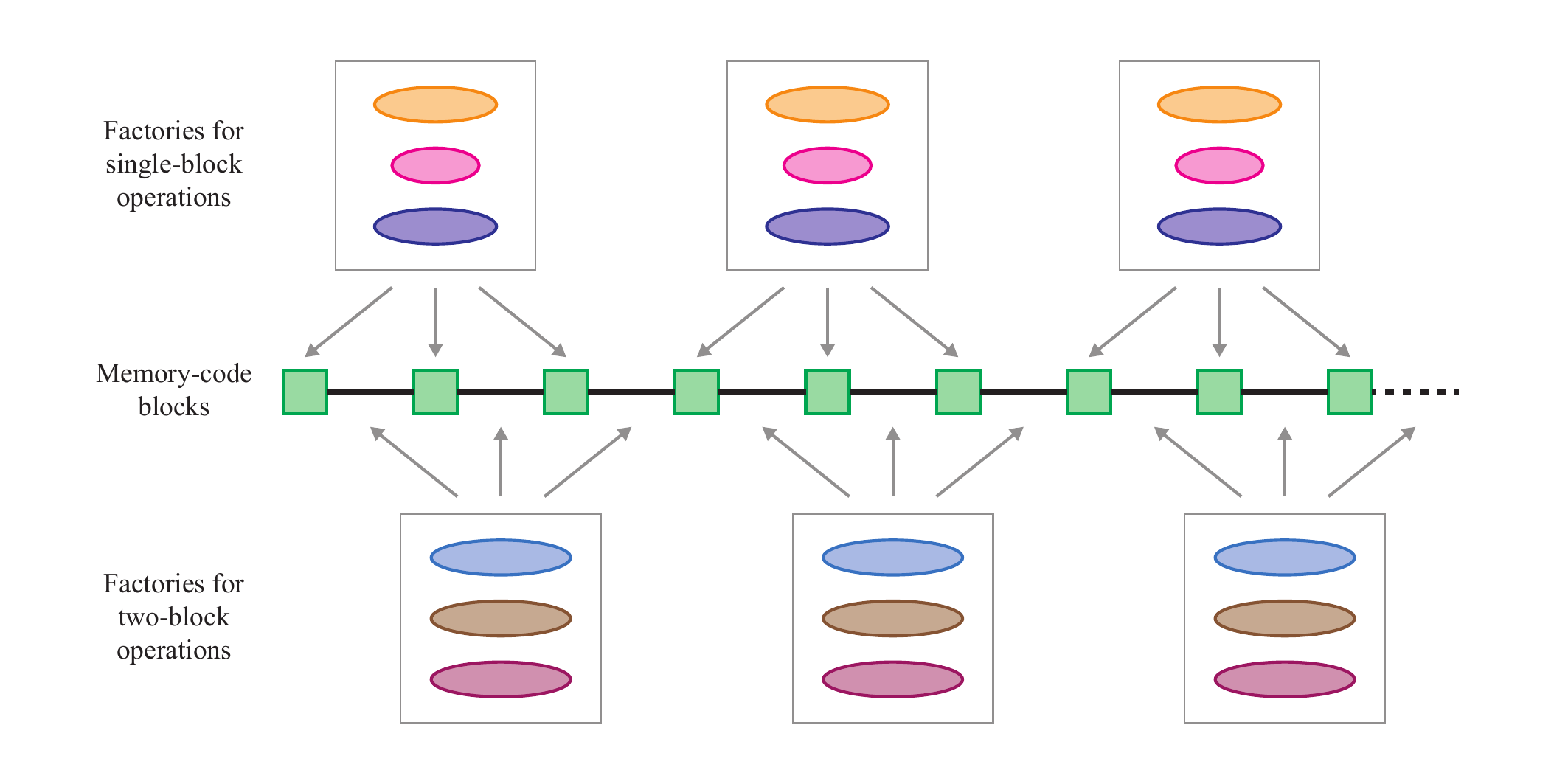}
    \caption{
    Network of memory-code blocks and resource-state factories. The architecture is mapped onto an almost-constant-degree physical qubit network where two-block measurements are restricted to pairs connected by black edges. Although illustrated as a one-dimensional chain, the network can be generalized to any graph with a constant vertex degree. Each distinct single-block and two-block measurement type is assigned a dedicated resource-state factory. Consequently, each physical qubit within a data block must potentially interface with every relevant specialized factory, leading to a physical vertex degree that scales linearly with the number of logical measurement types.
    }
    \label{fig:blocks_factories}
\end{figure}

Since our protocol exclusively utilizes (q)LDPC codes, one might intuitively expect it to require only a constant-degree physical qubit network, i.e.,~$D_P = O(1)$. In practice, however, a single physical qubit often participates in multiple distinct logical operations across successive circuit layers. This necessitates interactions with various resource states at different stages of the protocol. Consequently, the total set of required physical couplings for a given qubit can accumulate beyond $O(1)$, leading to a total vertex degree that exceeds a strict constant.

As detailed in Table~\ref{tab:measurements}, the number of distinct logical measurement types, denoted by $\sigma$, scales as $O(k)$. A naive implementation would allocate dedicated ancilla qubits and resource-state preparation circuits for each measurement type, referred to as a resource-state factory in what follows. As illustrated in Fig.~\ref{fig:blocks_factories}, each code block would require multiple dedicated factories to support single-block measurement types. Similarly, inter-block measurements would require additional specialized factories. This leads to two major challenges: 
\begin{itemize}
    \item \textbf{Qubit overhead:} Supporting $O(k)$ distinct measurement types with dedicated factories would incur a total qubit overhead of $O(k)\times O(Mkd_S^2)$, assuming no reuse of factory qubits. This scaling fundamentally violates the target of constant qubit overhead.
    \item \textbf{Degree explosion:} Since each code block must interface with every specialized factory (as indicated by the arrows in Fig.~\ref{fig:blocks_factories}), the physical vertex degree $D_P$ scales as $O(k)$, preventing the realization on a constant-degree qubit network.
\end{itemize}
Next, we show how to address these issues. 

\subsection{Rotation of logical qubit array}

\begin{figure}
    \centering
    \includegraphics[width=\linewidth]{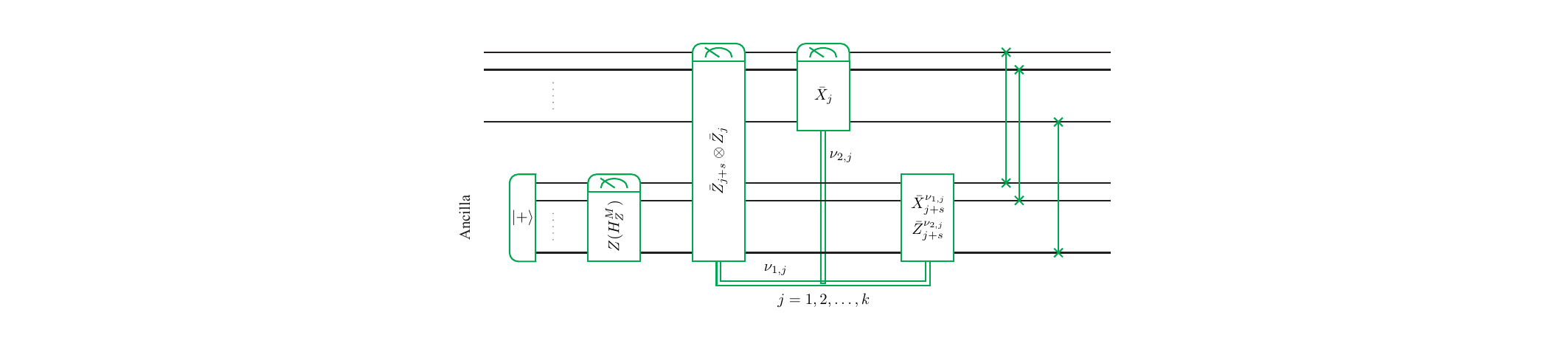}
    \caption{
    Logical rotation operation. The state of each $j$-th logical qubit is transferred to the $(j + s \pmod k)$th logical qubit.
    }
    \label{fig:Rot}
\end{figure}

We can significantly reduce the number of operation types by introducing a \textit{logical qubit array rotation} scheme within each memory-code block, which is proposed in Ref.~\cite{Nguyen2024}. Initially, $O(k)$ distinct operation types are required to access every logical qubit in a block. However, the same functionality can be achieved using only $O(\mathrm{polylog}(k))$ operation types. 

To implement this, we model the logical qubits within a block as a one-dimensional ring with indices from $1$ to $k$. In this configuration, logical measurements are restricted to the logical qubit at the first position. To operate on any other logical qubit, we rotate the ring until the target qubit reaches the first position. By utilizing only $O(\mathrm{polylog}(k))$ elementary rotation types (e.g.,~rotations by powers of two), we can realize any arbitrary rotation of the ring in at most $O(\mathrm{polylog}(k))$ steps. This introduces an additional $O(d^{o(1)})$ factor to the time overhead but reduces the variety of required physical resource-state factories. As illustrated in Fig.~\ref{fig:Rot}, these logical rotations can be implemented efficiently via code surgery.

By adopting this rotation-based protocol, we reduce the total number of operation types to $O(\mathrm{polylog}(k)) = O(d^{o(1)})$, which is almost constant with respect to the code distance.

\subsection{Overhead and Connectivity Scaling}

We define a logical measurement as the pair $(\sigma, B)$, where $\sigma$ denotes the measurement type and $B$ specifies the participating code blocks. Let $LM_\sigma = \{(\sigma, B)\}$ be the set of logical measurements of type $\sigma$. Since the logical coupling graph is characterized by a constant degree, it follows that $|LM_\sigma| = O(M)$. Then, we allocate physical qubits as follows:
\begin{itemize}
    \item \textbf{Data and auxiliary blocks:} We utilize $M$ data blocks alongside $O(M)$ auxiliary blocks. Because the underlying memory code has a constant encoding rate, the total number of physical qubits required for these blocks is $O(Mk)$.
    
    \item \textbf{Ancilla qubits for deformed codes:} For each measurement type $\sigma$, we partition the set $LM_\sigma$ into $O(|LM_\sigma|/k_R)$ groups, where each group contains up to $k_R$ logical measurements. Each group is assigned a dedicated set of physical qubits to construct the corresponding ancilla system for the deformed code. The number of physical qubits per group is $O(k_R k)$. Given that the rotation scheme reduces the number of distinct measurement types to $O(\mathrm{polylog}(k))$, the total number of physical qubits for all ancilla systems is $O(Mk\mathrm{polylog}(k))$.
    
    \item \textbf{Resource-state factories:} For each measurement type $\sigma$, we allocate a specialized factory to every set of $k_F$ groups. This requires $O(|LM_\sigma|/(k_F k_R))$ factories to serve all measurements of a given type. Including the qubit cost of surface-code encoding ($d_S^2$), each factory requires $O(k_F k_R k d_S^2)$ physical qubits. Summing over all types, the total qubit requirement for the factories is $O(Mkd_S^2 \mathrm{polylog}(k))$.
\end{itemize}
As a result, the total qubit overhead for the protocol is $O(d_S^2 \mathrm{polylog}(k)) = O(d^{o(1)})$, which is almost constant with respect to the code distance. 

The time overhead incurs an additional $O(\mathrm{polylog}(k))$ factor from the rotation operations, maintaining an overall complexity of $O(d^{a+o(1)})$. 

The connectivity is summarized below:
\begin{itemize}
    \item \textbf{Data and auxiliary blocks:} These physical qubits are coupled to their respective factory qubits. Since the number of distinct measurement types is $O(\mathrm{polylog}(k))$, the vertex degree for these qubits is bounded by $O(\mathrm{polylog}(k))$.
    
    \item \textbf{Ancilla qubits for deformed codes:} Each ancilla qubit is coupled to exactly one resource-state factory; thus, its vertex degree is $O(1)$.
    
    \item \textbf{Resource-state factories:} Each factory is dedicated to its corresponding code blocks, resulting in a constant vertex degree of $O(1)$.
\end{itemize}
Consequently, the maximum vertex degree of the physical qubit network is $D_P = O(\mathrm{polylog}(k)) = O(d^{o(1)})$, satisfying the requirement for an almost constant-degree architecture.

\section{Parallelized code surgery - Proof of Lemma~\ref{lem:PCS}}
\label{app:PCS}

\begin{lemma}{\cite{Zhang2025a}}
Consider a CSS code $(H_X, H_Z, J_X, J_Z)$. Let $Z(\alpha J_Z)$ be the set of logical operators to be measured, where $\alpha \in \mathbb{F}_2^{q \times k}$ is a full-rank matrix, and $q\leq k$. We can realize the measurement using a deformed code, which is a CSS code with check matrices 
\begin{eqnarray}
H^D_X &=& \left(\begin{array}{cc}
H_X & T \\
0 & H_M
\end{array}\right)
\label{eq:HDX}
\end{eqnarray}
and 
\begin{eqnarray}
H^D_Z &=& \left(\begin{array}{cc}
H_Z & 0 \\
S & H_G^\mathrm{T}
\end{array}\right),
\label{eq:HDZ}
\end{eqnarray}
where $S \in \mathbb{F}_2^{n_G \times n}$, $T \in \mathbb{F}_2^{r_X \times r_G}$, $H_G \in \mathbb{F}_2^{r_G \times n_G}$ and $H_G \in \mathbb{F}_2^{r_M \times r_G}$ satisfy the following conditions: 
\begin{itemize}
\item[i)] $H_X S^\mathrm{T} = T H_G$; 
\item[ii)] There exists a matrix $R \in \mathbb{F}_2^{n \times n_G}$ such that $\alpha J_Z R S = \alpha J_Z$ and $H_G (\alpha J_Z R)^\mathrm{T} = 0$; 
\item[iii)] There exists a matrix $\beta \in \mathbb{F}_2^{(k-q) \times r_G}$ such that $\alpha_\perp J_X S^\mathrm{T} = \beta H_G$, where $\alpha_\perp \in \mathbb{F}_2^{(k-q) \times k}$ is a full-rank matrix satisfying $\alpha_\perp \alpha^\mathrm{T} = 0$; and 
\item[iv)] $H_M H_G = 0$. 
\end{itemize}
The unmeasured $X$ and $Z$ logical operators are $X(\alpha_\perp J_X)$ and $Z({\alpha_\perp^\mathrm{r}}^\mathrm{T} J_Z)$, respectively. 
\label{lem:homological_measurement}
\end{lemma}

\begin{lemma}{\cite{Zhang2025a}}
Given an arbitrary set of $Z$ logical operators $Z(\alpha J_Z)$ to be measured, we can generate matrices $S$, $T$, and $H_G$ such that code surgery conditions i), ii), and iii) are satisfied, following the devised sticking approach. Furthermore, row and column weights of the matrices $H_G$, $S$, and $T$ are bounded above by a constant, and $\Vert S \Vert = 1$. When the generator matrices $J_X$ and $J_Z$ are expressed in standard form, and assuming that the logical operators selected for measurement act on mutually disjoint sets of logical qubits (i.e.~the logical thickness is one), the number of rows and columns in $H_G$ scales as $O(k)$. 
\label{lem:devised_sticking}
\end{lemma}

According to Lemma~\ref{lem:devised_sticking}, the check matrices $H^D_X$ and $H^D_Z$ given by Eqs.~(\ref{eq:HDX2})~and~(\ref{eq:HDZ2}) has the following properties: 
\begin{itemize}
\item[i)] Row and column weights are bounded above by a constant, i.e.~the deformed code is a qLDPC code; 
\item[ii)] The number of rows and columns scales as $O(k_R k)$. According to Eqs.~\eqref{eq:HDX2} and~\eqref{eq:HDZ2}, the row numbers of $H_X^D$ and $H_X^D$ are $k_R r_X+r_R r_G$ and $k_R r_Z+n_R n_G$, respectively, and their column numbers are both $k_R n+r_R n_G+n_R r_G$. Since $n_G,r_G\leq n$, $r_R\leq n_R$, and $r_X,r_Z\leq n$ (assuming the check matrices are full-rank), all these numbers are upper bounded by $(k_R+2n_R)n$. Hence, using $n=O(k)$ and $n_R=O(k_R)$, we conclude that the numbers of rows and columns are $O(k_R k)$.
\end{itemize}

\subsection{Logical measurements associated with the deformed codes}

Now, we verify that the deformed code given by Eqs.~(\ref{eq:HDX2})~and~(\ref{eq:HDZ2}) is valid for realizing the desired measurement. Let's define matrices 
\begin{eqnarray}
\tilde{H}_X &=& E_{k_R} \otimes H_X, \\
\tilde{H}_Z &=& E_{k_R} \otimes H_Z, \\
\tilde{J}_X &=& E_{k_R} \otimes J_X, \\
\tilde{J}_Z &=& E_{k_R} \otimes J_Z, \\
\tilde{\alpha} &=& E_{k_R} \otimes \alpha, \\
\tilde{\alpha}_\perp &=& E_{k_R} \otimes \alpha_\perp, \\
\end{eqnarray}
and 
\begin{eqnarray}
\tilde{S} &=& G_R^\mathrm{r} \otimes S, \\
\tilde{T} &=& \left(\begin{array}{cc}
0 & {G_R^\mathrm{r}}^\mathrm{T} \otimes T
\end{array}\right), \\
\tilde{H}_G &=& \left(\begin{array}{c}
H_R \otimes E_{n_G} \\
E_{n_R} \otimes H_G
\end{array}\right), \\
\tilde{H}_M &=& \left(\begin{array}{cc}
E_{r_R} \otimes H_G & H_R \otimes E_{r_X}
\end{array}\right), \\
\tilde{R} &=& G_R \otimes R, \\
\tilde{\beta} &=& \left(\begin{array}{cc}
0 & {G_R^\mathrm{r}}^\mathrm{T} \otimes \beta
\end{array}\right).
\end{eqnarray}
Then, we can find that check matrices in Eqs.~(\ref{eq:HDX2})~and~(\ref{eq:HDZ2}) are consistent with Eqs.~(\ref{eq:HDX})~and~(\ref{eq:HDZ}), and the matrices satisfy the code surgery conditions i), ii), iii), and iv), i.e. 
\begin{eqnarray}
\tilde{H}_X \tilde{S}^\mathrm{T} &=& \tilde{T} \tilde{H}_G, \label{eq:HXS} \\
\tilde{\alpha} \tilde{J}_Z \tilde{R} \tilde{S} &=& \tilde{\alpha} \tilde{J}_Z, \\
\tilde{H}_G (\tilde{\alpha} \tilde{J}_Z \tilde{R})^\mathrm{T} &=& 0, \\
\tilde{\alpha}_\perp \tilde{J}_X \tilde{S}^\mathrm{T} &=& \tilde{\beta} \tilde{H}_G, \label{eq:JXS} \\
\tilde{H}_M \tilde{H}_G &=& 0.
\end{eqnarray}
Therefore, according to Lemma~\ref{lem:homological_measurement}, with this deformed code, we can simultaneously measure logical operators $Z(\tilde{\alpha} \tilde{J}_Z)$. 

\subsection{Distances of deformed codes}

Without loss of generality, we consider the R-code generator matrix in the form 
\begin{eqnarray}
G_R &=& \left(\begin{array}{cc}
E_{k_R} & P_R
\end{array}\right),
\end{eqnarray}
and then 
\begin{eqnarray}
{G_R^\mathrm{r}}^\mathrm{T} &=& \left(\begin{array}{cc}
E_{k_R} & 0
\end{array}\right).
\end{eqnarray}

\begin{lemma}
Suppose $H_R$ is full-rank. For every vector $v \in \mathbb{F}_2^{r_R n_G + n_R r_X}$ satisfying $\tilde{H}_M v^\mathrm{T} = 0$ and $\vert v \vert < d_R$, there exists a vector $u_\star \in \mathbb{F}_2^{k_R n_G}$ such that $v^\mathrm{T} = \tilde{H}_G u_\star^\mathrm{T}$ and 
\begin{eqnarray}
\left\vert u_\star [(G_R^\mathrm{r})_{\bullet,l} \otimes S] \right\vert \leq \Vert S \Vert \times \vert v \vert
\label{eq:u_star}
\end{eqnarray}
for all $l = 1,2,\ldots,k_R$. 
\end{lemma}

\begin{proof}
Let $v = \left(\begin{array}{cc}a & b\end{array}\right)$. Then, the condition $\tilde{H}_M v^\mathrm{T} = 0$ becomes 
\begin{eqnarray}
(E_{r_R} \otimes H_G) a^\mathrm{T} &=& (H_R \otimes E_{r_X}) b^\mathrm{T}.
\end{eqnarray}
Let's introduce matrices $A$, $B$, and $U$ such that 
\begin{eqnarray}
\mathrm{vec}(A) &=& a^\mathrm{T}, \\
\mathrm{vec}(B) &=& b^\mathrm{T}, \\
\mathrm{vec}(U) &=& u_\star^\mathrm{T}.
\end{eqnarray}
We can rewrite the condition as 
\begin{eqnarray}
H_G A &=& B H_R^\mathrm{T}.
\label{eq:AB}
\end{eqnarray}
Similarly, we can rewrite $v^\mathrm{T} = \tilde{H}_G u_\star^\mathrm{T}$ as 
\begin{eqnarray}
A &=& U H_R^\mathrm{T}, \label{eq:A} \\
B &=& H_G U. \label{eq:B}
\end{eqnarray}

Since $H_R$ is full-rank, there exists $H_R^\mathrm{r} \in \mathbb{F}_2^{n_R \times r_R}$ such that $H_R H_R^\mathrm{r} = E_{r_R}$: We remark that $n_R > r_R$ when $H_R$ is full-rank and $k_R$ is not zero. Then, we can take 
\begin{eqnarray}
U &=& A {H_R^\mathrm{r}}^\mathrm{T}
\label{eq:U}
\end{eqnarray}
to satisfy Eq.~(\ref{eq:A}). Such a matrix $U$ also satisfies the inequality (\ref{eq:u_star}). Note that 
\begin{eqnarray}
[(G_R^\mathrm{r})_{\bullet,l} \otimes S]^\mathrm{T} u_\star^\mathrm{T} &=& \mathrm{vec}\left(S^\mathrm{T} U (G_R^\mathrm{r})_{\bullet,l}\right) = (S^\mathrm{T} U)_{\bullet,l} = S^\mathrm{T} U_{\bullet,l} = S^\mathrm{T} A ({H_R^\mathrm{r}}^\mathrm{T})_{\bullet,l}.
\end{eqnarray}
Therefore, 
\begin{eqnarray}
\left\vert [(G_R^\mathrm{r})_{\bullet,l} \otimes S]^\mathrm{T} u_\star^\mathrm{T} \right\vert \leq \Vert S \Vert \times \vert A ({H_R^\mathrm{r}}^\mathrm{T})_{\bullet,l} \vert \leq \Vert S \Vert \times \vert a \vert \leq \Vert S \Vert \times \vert v \vert.
\end{eqnarray}

Now, we prove that the matrix $U$ in Eq.~(\ref{eq:U}) also satisfies Eq.~(\ref{eq:B}) when $\vert v \vert < d_R$. Using the expression of $U$, we have 
\begin{eqnarray}
H_G U &=& H_G A {H_R^\mathrm{r}}^\mathrm{T} = B H_R^\mathrm{T} {H_R^\mathrm{r}}^\mathrm{T}.
\end{eqnarray}
Here, we have used Eq.~(\ref{eq:AB}). If 
\begin{eqnarray}
B = B H_R^\mathrm{T} {H_R^\mathrm{r}}^\mathrm{T},
\label{eq:B2}
\end{eqnarray}
Eq.~(\ref{eq:U}) holds. Note that the right inverse matrix of $H_R$ is not unique. Next, we will prove that there always exists a right inverse matrix of $H_R$ such that Eq.~(\ref{eq:B2}) is true. 

Because $\vert v \vert < d_R$, the matrix $B$ has at most $d_R-1$ non-zero columns. Suppose the indexes of the $d_R-1$ columns are $i_1,i_2,\ldots,i_{d_R-1}$ and all other columns are zero. We introduce the matrices 
\begin{eqnarray}
B^{(0)} &=& \left(\begin{array}{cccc}
B_{\bullet,i_1} & B_{\bullet,i_2} & \cdots & B_{\bullet,i_{d_R-1}}
\end{array}\right), \\
H_R^{(0)} &=& \left(\begin{array}{cccc}
(H_R)_{\bullet,i_1} & (H_R)_{\bullet,i_2} & \cdots & (H_R)_{\bullet,i_{d_R-1}}
\end{array}\right).
\end{eqnarray}
Later, we will prove that the $d_R-1$ columns in $H^{(0)}_R$ are linearly independent, and then $r_R \geq d_R-1$. Because $H_R$ is full-rank and has $r_R$ linearly-independent columns, we can find additional $r_R-(d_R-1)$ columns to constitute a complete set together with columns in $H^{(0)}_R$. Let $i_{d_R},i_{d_R+1},\ldots,i_{r_R}$ be indexes of the additional $r_R-(d_R-1)$ columns, we define 
\begin{eqnarray}
H_R^{(1)} &=& \left(\begin{array}{cccc}
(H_R)_{\bullet,i_{d_R}} & (H_R)_{\bullet,i_{d_R+1}} & \cdots & (H_R)_{\bullet,i_{r_R}}
\end{array}\right).
\end{eqnarray}
Then, there exists a permutation matrix $\pi \in \mathbb{F}_2^{n_R \times n_R}$ such that matrices $B$ and $H_R$ are in the forms 
\begin{eqnarray}
B &=& \left(\begin{array}{ccc}
B^{(0)} & 0 & 0
\end{array}\right) \pi, \\
H_R &=& \left(\begin{array}{ccc}
H^{(0)}_R & H^{(1)}_R & H^{(2)}_R
\end{array}\right) \pi.
\end{eqnarray}
Because $H^{(0,1)}_R = \left(\begin{array}{cc}H^{(0)}_R & H^{(1)}_R\end{array}\right)$ is invertible, a right inverse matrix of $H_R$ is 
\begin{eqnarray}
H_R^\mathrm{r} &=& \pi^{-1} \left(\begin{array}{c}
{H^{(0,1)}_R}^{-1} \\
0
\end{array}\right) = \pi^{-1} \left(\begin{array}{c}
({H^{(0,1)}_R}^{-1})_{1:d_R-1,\bullet} \\
({H^{(0,1)}_R}^{-1})_{d_R:r_R,\bullet} \\
0
\end{array}\right).
\end{eqnarray}
Note that $\pi^{-1} = \pi^\mathrm{T}$. Taking this right inverse matrix, we have 
\begin{eqnarray}
B H_R^\mathrm{T} {H_R^\mathrm{r}}^\mathrm{T} = \left(\begin{array}{ccc}
B^{(0)} & 0 & 0
\end{array}\right) \left(\begin{array}{ccc}
E_{d_R-1} & 0 & 0 \\
0 & E_{r_R-(d_R-1)} & 0 \\
{H^{(2)}_R}^\mathrm{T} [({H^{(0,1)}_R}^{-1})_{1:d_R-1,\bullet}]^\mathrm{T} & {H^{(2)}_R}^\mathrm{T} [({H^{(0,1)}_R}^{-1})_{d_R:r_R,\bullet}]^\mathrm{T} & 0
\end{array}\right) \pi = B.
\end{eqnarray}
Eq.~(\ref{eq:B2}) has been proved. 

Lastly, we prove that columns in $H^{(0)}_R$ are linearly independent. Suppose columns in $H^{(0)}_R$ are not linearly independent, there exists a vector $x \in \mathbb{F}_2^{d_R-1}-\{0\}$ such that $H^{(0)}_R x^\mathrm{T} = 0$. Then, we can construct a vector $y \in \mathbb{F}_2^{n_R}$, 
\begin{eqnarray}
y &=& \left(\begin{array}{ccc}
x & 0 & 0
\end{array}\right) \pi.
\end{eqnarray}
We have $H_R y^\mathrm{T} = 0$, $y \neq 0$, and $\vert y \vert < d_R$. This is in contradiction with the code distance of the R code. The lemma has been proved. 
\end{proof}

\begin{lemma}
For the deformed code with check matrices in Eqs.~(\ref{eq:HDX2})~and~(\ref{eq:HDZ2}), the generator matrices of $X$ and $Z$ logical operators are 
\begin{eqnarray}
J^D_X = \left(\begin{array}{cc}
\tilde{\alpha}_\perp \tilde{J}_X & \tilde{\beta}
\end{array}\right)
\end{eqnarray}
and 
\begin{eqnarray}
J^D_Z = \left(\begin{array}{cc}
\tilde{\alpha}_\perp^\mathrm{r}{}^\mathrm{T} \tilde{J}_Z & 0
\end{array}\right),
\end{eqnarray}
respectively. Let $d$ be the distance of the code $(H_X, H_Z, J_X, J_Z)$. If $H_R$ is full-rank, the distance of the deformed code satisfies 
\begin{eqnarray}
d_D &\geq& \min\{d / \Vert S \Vert, d_R\}.
\end{eqnarray}
\label{lem:lattice_surgery}
\end{lemma}

\begin{proof}
The generator matrices satisfy conditions $H^D_X {J^D_Z}^\mathrm{T} = H^D_Z {J^D_X}^\mathrm{T} = 0$ and $J^D_X  {J^D_Z}^\mathrm{T} = E_{k_R(k-q)}$, i.e.~they are valid generator matrices of the deformed code. 

Let $e = \left(\begin{array}{cc}u & v\end{array}\right)$ be an $X$ logical error, i.e.~$H^D_Z e^\mathrm{T} = 0$ and $J^D_Z e^\mathrm{T} \neq 0$. Then, we have 
\begin{eqnarray}
\tilde{H}_Z u^\mathrm{T} &=& 0
\end{eqnarray}
and 
\begin{eqnarray}
\tilde{\alpha}_\perp^\mathrm{r}{}^\mathrm{T} \tilde{J}_Z u^\mathrm{T} &\neq& 0.
\end{eqnarray}
Let's introduce the matrix $U$ such that 
\begin{eqnarray}
\mathrm{vec}(U) &=& u^\mathrm{T}.
\end{eqnarray}
Then, 
\begin{eqnarray}
H_Z U &=& 0
\end{eqnarray}
and 
\begin{eqnarray}
{\alpha_\perp^\mathrm{r}}^\mathrm{T} J_Z U &\neq& 0.
\end{eqnarray}
Suppose the $m$-th column in ${\alpha_\perp^\mathrm{r}}^\mathrm{T} J_Z U$ is not zero, we have 
\begin{eqnarray}
H_Z U_{\bullet,m} &=& 0
\end{eqnarray}
and 
\begin{eqnarray}
{\alpha_\perp^\mathrm{r}}^\mathrm{T} J_Z U_{\bullet,m} &\neq& 0.
\end{eqnarray}
According to the code distance $d$, 
\begin{eqnarray}
\vert e\vert &\geq& \vert U_{\bullet,m} \vert \geq d.
\end{eqnarray}

Let $e = \left(\begin{array}{cc}u & v\end{array}\right)$ be a $Z$ logical error, i.e.~$H^D_X e^\mathrm{T} = 0$ and $J^D_X e^\mathrm{T} \neq 0$. Then, we have 
\begin{eqnarray}
\tilde{H}_X u^\mathrm{T} &=& \tilde{T} v^\mathrm{T}, \\
\tilde{H}_M v^\mathrm{T} &=& 0,
\end{eqnarray}
and 
\begin{eqnarray}
\tilde{\alpha}_\perp \tilde{J}_X u^\mathrm{T} + \tilde{\beta} v^\mathrm{T} &\neq& 0.
\end{eqnarray}
If $\vert v \vert < d_R$, there exists a vector $u_\star \in \mathbb{F}_2^{k_R n_G}$ such that $v^\mathrm{T} = \tilde{H}_G u_\star^\mathrm{T}$ and the inequality (\ref{eq:u_star}) holds. Then, 
\begin{eqnarray}
\tilde{H}_X u^\mathrm{T} &=& \tilde{T} \tilde{H}_G u_\star^\mathrm{T}
\end{eqnarray}
and 
\begin{eqnarray}
\tilde{\alpha}_\perp \tilde{J}_X u^\mathrm{T} + \tilde{\beta} \tilde{H}_G u_\star^\mathrm{T} &\neq& 0.
\end{eqnarray}
Furthermore, we have 
\begin{eqnarray}
\tilde{H}_X (u + u_\star \tilde{S})^\mathrm{T} &=& 0
\end{eqnarray}
and 
\begin{eqnarray}
\tilde{\alpha}_\perp \tilde{J}_X (u + u_\star \tilde{S})^\mathrm{T} &\neq& 0.
\end{eqnarray}
Here, we have used Eqs.~(\ref{eq:HXS})~and~(\ref{eq:JXS}). Let's introduce the matrix $U$ such that 
\begin{eqnarray}
\mathrm{vec}(U) &=& (u + u_\star \tilde{S})^\mathrm{T}.
\end{eqnarray}
Then, 
\begin{eqnarray}
H_X U &=& 0
\end{eqnarray}
and 
\begin{eqnarray}
\alpha_\perp J_X U &\neq& 0.
\end{eqnarray}
Suppose the $m$-th column in $\alpha_\perp J_X U$ is not zero, we have 
\begin{eqnarray}
H_X U_{\bullet,m} &=& 0
\end{eqnarray}
and 
\begin{eqnarray}
\alpha_\perp J_X U_{\bullet,m} &\neq& 0.
\end{eqnarray}
According to the code distance $d$, 
\begin{eqnarray}
\vert U_{\bullet,m} \vert \geq d.
\end{eqnarray}
Note that 
\begin{eqnarray}
U_{\bullet,m} &=& U (E_{k_R})_{\bullet,m} = [(E_{k_R})_{m,\bullet} \otimes E_n] (u + u_\star \tilde{S})^\mathrm{T} \notag \\
&=& [(E_{k_R})_{m,\bullet} \otimes E_n] u^\mathrm{T} + [(G_R^\mathrm{r})_{\bullet,m} \otimes S]^\mathrm{T} u_\star^\mathrm{T},
\end{eqnarray}
therefore, 
\begin{eqnarray}
\Vert S \Vert \times \vert e \vert &\geq& \vert u \vert + \Vert S \Vert \times \vert v \vert \geq \vert [(E_{k_R})_{m,\bullet} \otimes E_n] u^\mathrm{T} \vert + \vert [(G_R^\mathrm{r})_{\bullet,m} \otimes S]^\mathrm{T} u_\star^\mathrm{T} \vert \geq d.
\end{eqnarray}
\end{proof}

\section{Locally-testable state preparation - Proof of Lemma~\ref{lem:LTSP}}
\label{app:LTSP}

In LTSP, we generate resource states used to measure one type of stabilizer operator---either $X$ or $Z$---of a CSS code. Without loss of generality, we focus on the measurement of $Z$ stabilizer operators, specifically the operators $Z(H^D_Z)$ in the deformed code. The conclusions are straightforwardly applicable to $X$ stabilizer operators or to other CSS codes, such as the memory code. 

\subsection{Stabilizer of the resource state}

\begin{figure}[htbp]
\centering
\includegraphics[width=\linewidth]{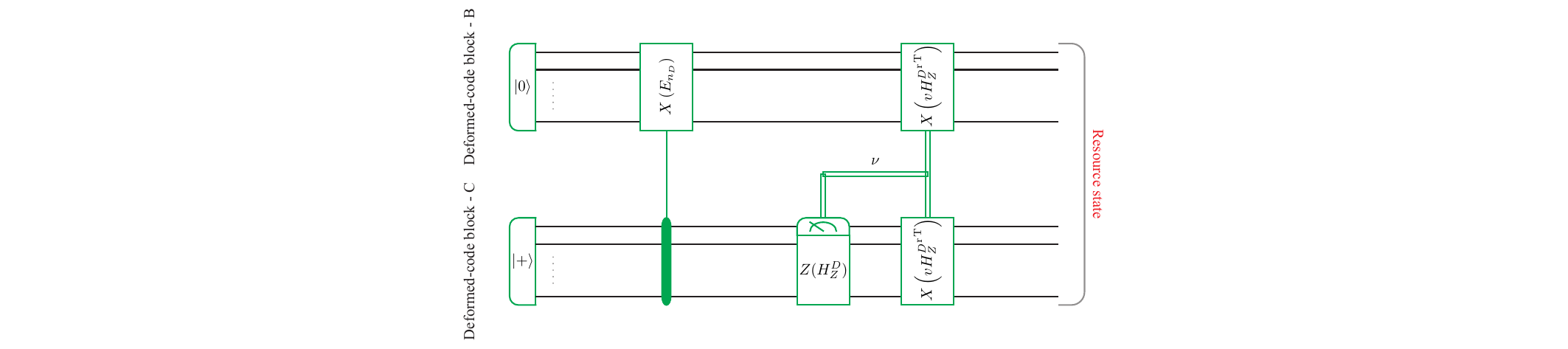}
\caption{
Circuit that defines the resource state for the $Z(H^D_Z)$ measurement. Outcome of the $Z(H^D_Z)$ measurement in the circuit is $\nu \in \mathbb{F}_2^{r^D_Z}$. 
}
\label{fig:resource_state}
\end{figure}

To implement the measurement of $Z(H^D_Z)$ using gate teleportation, we need to prepare a resource state encoded in two blocks of the deformed code, denoted by $B$ and $C$. Each block contains $n_D$ qubits. Ideally, the resource state can be prepared as follows; see Fig.~\ref{fig:resource_state}.  First, we generate $n_D$ copies of the Bell state between the two blocks, which constitute a $2n_D$-qubit stabilizer state, and the stabilizer is generated by operators $X(H^{BS}_X)$ and $Z(H^{BS}_Z)$, where 
\begin{eqnarray}
H^{BS}_X &=& H^{BS}_Z = \left(\begin{array}{cc}
E_{n_D} & E_{n_D}
\end{array}\right).
\end{eqnarray}
We can rewrite the matrices as 
\begin{eqnarray}
\tilde{H}^{BS}_X &=& \left(\begin{array}{cc}
H^D_X & H^D_X \\
J^D_X & J^D_X \\
{{H^D_Z}^\mathrm{r}}^\mathrm{T} & {{H^D_Z}^\mathrm{r}}^\mathrm{T}
\end{array}\right)
\end{eqnarray}
and 
\begin{eqnarray}
\tilde{H}^{BS}_Z &=& \left(\begin{array}{cc}
{{H^D_X}^\mathrm{r}}^\mathrm{T} & {{H^D_X}^\mathrm{r}}^\mathrm{T} \\
J^D_Z & J^D_Z \\
H^D_Z & H^D_Z
\end{array}\right),
\end{eqnarray}
such that $\mathrm{rowsp}H^{BS}_{X/Z} = \mathrm{rowsp}\tilde{H}^{BS}_{X/Z}$. Without loss generality, we can choose ${H^D_X}^\mathrm{r}$ and ${H^D_Z}^\mathrm{r}$ such that $J^D_X {H^D_X}^\mathrm{r} = J^D_Z {H^D_Z}^\mathrm{r} = {{H^D_X}^\mathrm{r}}^\mathrm{T} {H^D_Z}^\mathrm{r} = 0$. Then, we apply the measurement of operators $Z(H^D_Z)$ on block-$C$, followed by a feedback Pauli gate. Note that $X\left(\nu {{H^D_Z}^\mathrm{r}}^\mathrm{T}\right) Z\left((H^D_Z)_{j,\bullet}\right) X\left(\nu {{H^D_Z}^\mathrm{r}}^\mathrm{T}\right) = (-1)^{\nu_j} Z\left((H^D_Z)_{j,\bullet}\right)$. These operations transform Bell states into another stabilizer sate with the stabilizer generated by operators $X(H^{RS}_X)$ and $Z(H^{RS}_Z)$, where 
\begin{eqnarray}
H^{RS}_X &=& \left(\begin{array}{cc}
H^D_X & H^D_X \\
J^D_X & J^D_X
\end{array}\right)
\end{eqnarray}
and 
\begin{eqnarray}
H^{RS}_Z &=& \left(\begin{array}{cc}
E_{n_D} & E_{n_D} \\
0 & H^D_Z
\end{array}\right).
\end{eqnarray}
The resource state to be prepared is the stabilizer state defined by the generators $X(H^{RS}_X)$ and $Z(H^{RS}_Z)$. 

\subsection{State preparation using a classical locally testable code}

\begin{figure}[htbp]
\centering
\includegraphics[width=\linewidth]{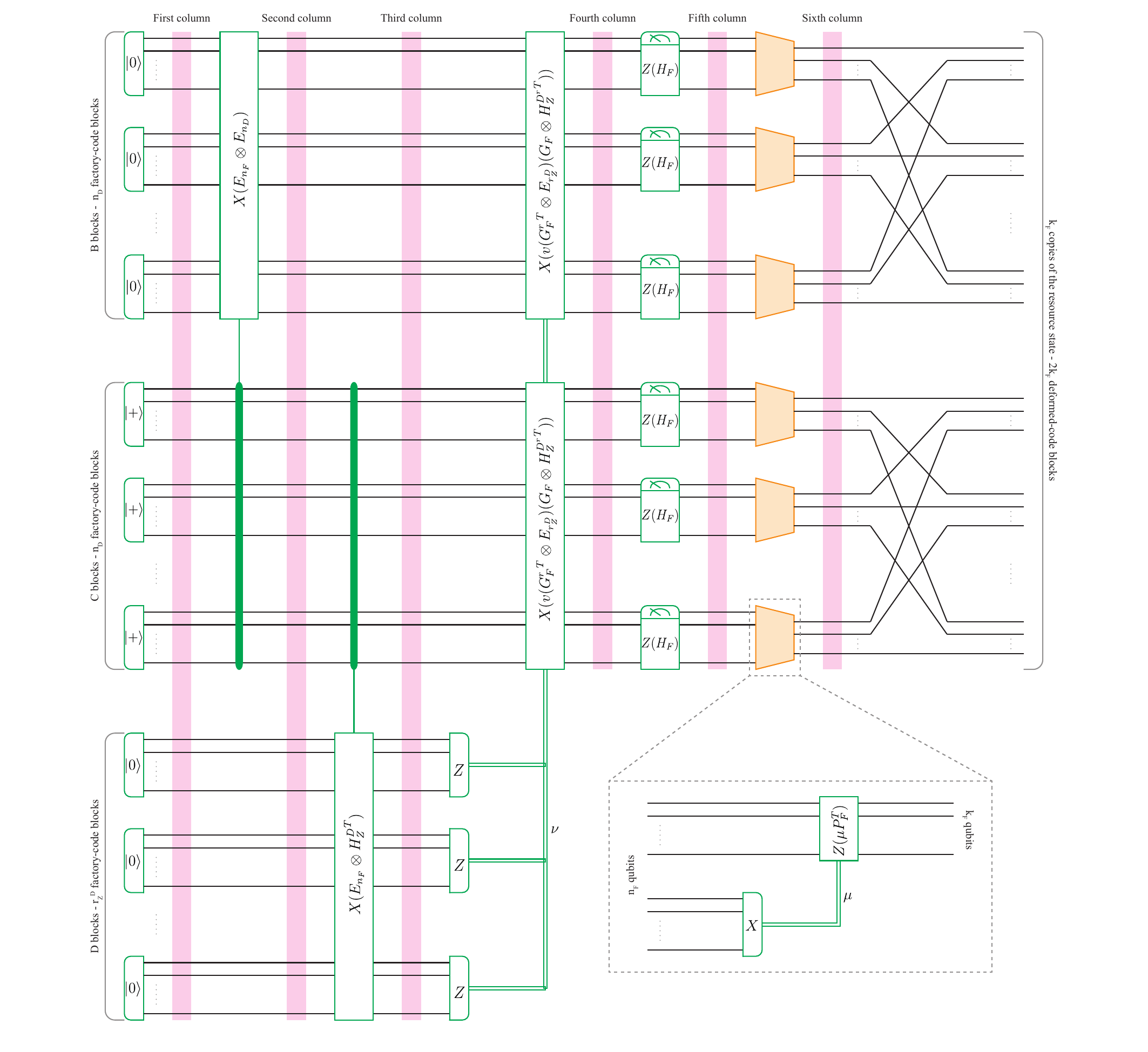}
\caption{
Locally-testable state preparation. The outcome of the transversal measurement on $D$ blocks is $\nu = \mathrm{vec}(V)^\mathrm{T}$, where $V \in \mathbb{F}_2^{r^D_Z \times n_F}$, and $V_{i,j}$ is the outcome of the measurement on the $j$-th qubit in the $i$-th block. The outcome of the partial transversal measurement on one of $B$ and $C$ blocks in decoding is $\mu \in \mathbb{F}_2^{n_F-k_F}$, and $\mu_j$ is the outcome of the measurement on the $(k_F+j)$-th qubit in the block. 
}
\label{fig:state_preparation}
\end{figure}

We prepare the resource state with the circuit shown in Fig.~\ref{fig:state_preparation}. The circuit constitutes of three sets of F-code blocks, denoted by $B$, $C$, and $D$, containing $n_D$, $n_D$, and $r^D_Z$ blocks, respectively. In each F-code block, we encode $k_F$ logical qubits according to the following CSS code: Stabilizer operators are $X(0)$ and $Z(H_F)$, and logical operators are $X(G_F)$ and $Z({G_F^\mathrm{r}}^\mathrm{T})$. Without loss of generality, we consider the F-code generator matrix in the form 
\begin{eqnarray}
G_F &=& \left(\begin{array}{cc}
E_{k_F} & P_F
\end{array}\right),
\end{eqnarray}
and then 
\begin{eqnarray}
{G_F^\mathrm{r}}^\mathrm{T} &=& \left(\begin{array}{cc}
E_{k_F} & 0
\end{array}\right).
\end{eqnarray}
This circuit yields $n_F$ copies of the resource state and is supported by the following lemmas. 

\begin{lemma}
Let $\vert e^{sp}_Z \vert$ denote the total weight of $Z$ errors occurring at the spacetime locations in the state preparation circuit shown in Fig.~\ref{fig:state_preparation}. Then, each copy of the prepared resource state contains $Z$ errors with total weight at most $\vert e^{RS}_Z \vert \leq \vert e^{sp}_Z \vert$. 
\label{lem:spX}
\end{lemma}

\begin{lemma}
Let $\vert e^{sp}_X \vert$ denote the total weight of $X$ errors (including incorrect outcomes in projective measurements of $Z$ operators) occurring at the spacetime locations in the state preparation circuit shown in Fig.~\ref{fig:state_preparation}. Suppose the F code is $(\omega,s)$-locally testable, and that $H^D_Z$ has rows and columns of weights at most $\omega^D_Z$. If the spacetime $X$ errors are undetectable by the two sets of checks illustrated in Fig.~\ref{fig:LTSP}, and 
\begin{eqnarray}
\vert e^{sp}_X \vert &<& \frac{d_F}{\omega^D_Z \max\left\{1,\frac{n_F}{r_Fs}\right\}},
\label{eq:espX}
\end{eqnarray}
then each copy of the prepared resource state contains $X$ errors with total weight at most 
\begin{eqnarray}
\vert e^{RS}_X \vert &\leq& \max\left\{1,\frac{n_F}{r_Fs}\right\} \times \vert e^{sp}_X \vert.
\label{eq:eRSX}
\end{eqnarray}
\label{lem:spZ}
\end{lemma}

\subsubsection{Proof of Lemma~\ref{lem:spX} - $X$ operators and $Z$ errors}

Focusing on the $j$-th copy of the prepared resource state, i.e.~$j$-th logical qubits in $B$ and $C$ F-code blocks, the propagation of $X$ operators in spacetime is described by the generator matrix 
\begin{eqnarray}
J^{sp}_{X,j} &=& \left(\begin{array}{ccc}
J^{sp}_{X,B,j} & J^{sp}_{X,C,j} & J^{sp}_{X,D,j}
\end{array}\right),
\end{eqnarray}
where 
\begin{eqnarray}
J^{sp}_{X,B,j} &=& \left(\begin{array}{cccccc}
0 & (G_F)_{j,\bullet} \otimes H^D_X & (G_F)_{j,\bullet} \otimes H^D_X & (G_F)_{j,\bullet} \otimes H^D_X & (G_F)_{j,\bullet} \otimes H^D_X & (E_{k_F})_{j,\bullet} \otimes H^D_X \\
0 & (G_F)_{j,\bullet} \otimes J^D_X & (G_F)_{j,\bullet} \otimes J^D_X & (G_F)_{j,\bullet} \otimes J^D_X & (G_F)_{j,\bullet} \otimes J^D_X & (E_{k_F})_{j,\bullet} \otimes J^D_X
\end{array}\right),
\end{eqnarray}
\begin{eqnarray}
J^{sp}_{X,C,j} &=& \left(\begin{array}{cccccc}
(G_F)_{j,\bullet} \otimes H^D_X & (G_F)_{j,\bullet} \otimes H^D_X & (G_F)_{j,\bullet} \otimes H^D_X & (G_F)_{j,\bullet} \otimes H^D_X & (G_F)_{j,\bullet} \otimes H^D_X & (E_{k_F})_{j,\bullet} \otimes H^D_X \\
(G_F)_{j,\bullet} \otimes J^D_X & (G_F)_{j,\bullet} \otimes J^D_X & (G_F)_{j,\bullet} \otimes J^D_X & (G_F)_{j,\bullet} \otimes J^D_X & (G_F)_{j,\bullet} \otimes J^D_X & (E_{k_F})_{j,\bullet} \otimes J^D_X
\end{array}\right),~~~
\end{eqnarray}
and 
\begin{eqnarray}
J^{sp}_{X,D,j} &=& \left(\begin{array}{ccc}
0 & 0 & 0 \\
0 & 0 & 0
\end{array}\right),
\end{eqnarray}
correspond to spacetime locations on $B$, $C$, and $D$ blocks, respectively: The six columns correspond to spacetime locations as illustrated in Fig.~\ref{fig:state_preparation}. We remark that operators on $C$ blocks do not propagate to $D$ blocks by the second transversal controlled-NOT gate because of $[(G_F)_{j,\bullet} \otimes H^D_X][E_{n_F} \otimes {H^D_Z}^\mathrm{T}] = [(G_F)_{j,\bullet} \otimes J^D_X][E_{n_F} \otimes {H^D_Z}^\mathrm{T}] = 0$. This generator matrix describes that operators $X\left((G_F)_{j,\bullet} \otimes H^D_X\right)$ and $X\left((G_F)_{j,\bullet} \otimes J^D_X\right)$ on $C$ blocks at the beginning are transformed into operators $X\left((G_F)_{j,\bullet} \otimes H^D_X\right) \otimes X\left((G_F)_{j,\bullet} \otimes H^D_X\right)$ and $X\left((G_F)_{j,\bullet} \otimes J^D_X\right) \otimes X\left((G_F)_{j,\bullet} \otimes J^D_X\right)$ on $B$ and $C$ blocks at the end, respectively. Because qubits in $C$ blocks are initialized in the state $\ket{+}$, the final state is an eigenstate of operators $X\left((G_F)_{j,\bullet} \otimes H^D_X\right) \otimes X\left((G_F)_{j,\bullet} \otimes H^D_X\right)$ and $X\left((G_F)_{j,\bullet} \otimes J^D_X\right) \otimes X\left((G_F)_{j,\bullet} \otimes J^D_X\right)$ on $B$ and $C$ blocks with the eigenvalue $+1$. 

We represent $Z$ errors in spacetime with a vector 
\begin{eqnarray}
e^{sp}_Z &=& \left(\begin{array}{ccc}
e^{sp}_{Z,B} & e^{sp}_{Z,C} & e^{sp}_{Z,D}
\end{array}\right),
\end{eqnarray}
where 
\begin{eqnarray}
e^{sp}_{Z,B/C} &=& \left(\begin{array}{cccccc}
u_{B/C,1} & u_{B/C,2} & u_{B/C,3} & u_{B/C,4} & u_{B/C,5} & u_{B/C,6}
\end{array}\right)
\end{eqnarray}
and 
\begin{eqnarray}
e^{sp}_{Z,D} &=& \left(\begin{array}{ccc}
u_{D,1} & u_{D,2} & u_{D,3}
\end{array}\right).
\end{eqnarray}
The errors flip operators during the propagation according to 
\begin{eqnarray}
J^{sp}_{X,j}{e^{sp}_Z}^\mathrm{T} &=& \left(\begin{array}{c}
\mathrm{vec}\left(H^D_X U [(G_F)_{j,\bullet}]^\mathrm{T}\right) \\
\mathrm{vec}\left(J^D_X U [(G_F)_{j,\bullet}]^\mathrm{T}\right)
\end{array}\right) + \left(\begin{array}{c}
\mathrm{vec}\left(H^D_X U' [(E_{k_F})_{j,\bullet}]^\mathrm{T}\right) \\
\mathrm{vec}\left(J^D_X U' [(E_{k_F})_{j,\bullet}]^\mathrm{T}\right)
\end{array}\right) \notag \\
&=& \left(\begin{array}{c}
H^D_X u_{eff}^\mathrm{T} \\
J^D_X u_{eff}^\mathrm{T}
\end{array}\right) = H^{RS}_X \left(\begin{array}{c}
0 \\
u_{eff}^\mathrm{T}
\end{array}\right),
\label{eq:JXeZ_sp}
\end{eqnarray}
where 
\begin{eqnarray}
\mathrm{vec}(U) &=& (u_{B,2} + u_{B,3} + u_{B,4} + u_{B,5} + u_{C,1} + u_{C,2} + u_{C,3} + u_{C,3} + u_{C,4} + u_{C,5})^\mathrm{T}, \\
\mathrm{vec}(U') &=& (u_{B,6} + u_{C,6})^\mathrm{T}, \\
u_{eff}^\mathrm{T} &=& U [(G_F)_{j,\bullet}]^\mathrm{T} + U' [(E_{k_F})_{j,\bullet}]^\mathrm{T}. \label{eq:ueff}
\end{eqnarray}

We can find that the vector 
\begin{eqnarray}
e^{RS}_Z &=& \left(\begin{array}{cc}
0 & u_{eff}
\end{array}\right) \in \mathbb{F}_2^{2n_D}
\end{eqnarray}
satisfies $H^{RS}_X {e^{RS}_Z}^\mathrm{T} = J^{sp}_X {e^{sp}_Z}^\mathrm{T}$ as illustrated in Eq.~(\ref{eq:JXeZ_sp}). Therefore, $e^{RS}_Z$ represents $Z$ errors on the prepared resource state. According to Eq.~(\ref{eq:ueff}), it satisfies the inequality $\vert e^{RS}_Z \vert \leq \vert e^{sp}_Z \vert$. 

\subsubsection{Proof of Lemma~\ref{lem:spZ} - $Z$ operators and $X$ errors}

Focusing on the $j$-th copy of the prepared resource state, i.e.~$j$-th logical qubits in $B$ and $C$ F-code blocks, the propagation of $Z$ operators in spacetime is described by the generator matrix 
\begin{eqnarray}
J^{sp}_{Z,j} &=& \left(\begin{array}{cccc}
J^{sp}_{Z,B,j} & J^{sp}_{Z,C,j} & J^{sp}_{Z,D,j} & J^{sp}_{Z,mea,j}
\end{array}\right),
\end{eqnarray}
where 
\begin{eqnarray}
J^{sp}_{Z,B,j} &=& \left(\begin{array}{ccc}
({G_F^\mathrm{r}}^\mathrm{T})_{j,\bullet} \otimes E_{n_D} & ({G_F^\mathrm{r}}^\mathrm{T})_{j,\bullet} \otimes E_{n_D} & ({G_F^\mathrm{r}}^\mathrm{T})_{j,\bullet} \otimes E_{n_D} \\
0 & 0 & 0
\end{array}\right. \notag \\
&& \left.\begin{array}{ccc}
({G_F^\mathrm{r}}^\mathrm{T})_{j,\bullet} \otimes E_{n_D} & ({G_F^\mathrm{r}}^\mathrm{T})_{j,\bullet} \otimes E_{n_D} & (E_{k_F})_{j,\bullet} \otimes E_{n_D} \\
0 & 0 & 0
\end{array}\right),
\end{eqnarray}
\begin{eqnarray}
J^{sp}_{Z,C,j} &=& \left(\begin{array}{cccccc}
0 & ({G_F^\mathrm{r}}^\mathrm{T})_{j,\bullet} \otimes E_{n_D} & ({G_F^\mathrm{r}}^\mathrm{T})_{j,\bullet} \otimes E_{n_D} & ({G_F^\mathrm{r}}^\mathrm{T})_{j,\bullet} \otimes E_{n_D} & ({G_F^\mathrm{r}}^\mathrm{T})_{j,\bullet} \otimes E_{n_D} & (E_{k_F})_{j,\bullet} \otimes E_{n_D} \\
0 & 0 & ({G_F^\mathrm{r}}^\mathrm{T})_{j,\bullet} \otimes H^D_Z & ({G_F^\mathrm{r}}^\mathrm{T})_{j,\bullet} \otimes H^D_Z & ({G_F^\mathrm{r}}^\mathrm{T})_{j,\bullet} \otimes H^D_Z & (E_{k_F})_{j,\bullet} \otimes H^D_Z
\end{array}\right),
\end{eqnarray}
and 
\begin{eqnarray}
J^{sp}_{Z,D,j} &=& \left(\begin{array}{ccc}
0 & 0 & 0 \\
({G_F^\mathrm{r}}^\mathrm{T})_{j,\bullet} \otimes E_{r^D_Z} & ({G_F^\mathrm{r}}^\mathrm{T})_{j,\bullet} \otimes E_{r^D_Z} & ({G_F^\mathrm{r}}^\mathrm{T})_{j,\bullet} \otimes E_{r^D_Z}
\end{array}\right),
\end{eqnarray}
correspond to spacetime locations on $B$, $C$, and $D$ blocks, respectively, and 
\begin{eqnarray}
J^{sp}_{Z,mea,j} &=& \left(\begin{array}{cc}
0 & 0 \\
0 & 0
\end{array}\right)
\end{eqnarray}
corresponds to measurements of $Z(H_F)$: The two columns correspond to measurements on $B$ and $C$ blocks, respectively. Note that $[({G_F^\mathrm{r}}^\mathrm{T})_{j,\bullet} \otimes H^D_Z][(G_F^\mathrm{r} \otimes E_{r^D_Z}) (G_F \otimes {{H^D_Z}^\mathrm{r}}^\mathrm{T})]^\mathrm{T} = ({G_F^\mathrm{r}}^\mathrm{T})_{j,\bullet} \otimes E_{r^D_Z}$, which justifies the propagation through the feedback gate depending on the outcome $\nu$. This generator matrix describes that operators $Z\left(({G_F^\mathrm{r}}^\mathrm{T})_{j,\bullet} \otimes E_{n_D}\right)$ on $B$ blocks and and $Z\left(({G_F^\mathrm{r}}^\mathrm{T})_{j,\bullet} \otimes E_{r^D_Z}\right)$ on $D$ blocks at the beginning are transformed into operators $Z\left((E_{k_F})_{j,\bullet} \otimes E_{n_D}\right) \otimes Z\left((E_{k_F})_{j,\bullet} \otimes E_{n_D}\right)$ on $B$ and $C$ blocks and $Z\left((E_{k_F})_{j,\bullet} \otimes H^D_Z\right)$ on $C$ blocks at the end, respectively. Because qubits in $C$ and $D$ blocks are initialized in the state $\ket{0}$, the final state is an eigenstate of operators $Z\left((E_{k_F})_{j,\bullet} \otimes E_{n_D}\right) \otimes Z\left((E_{k_F})_{j,\bullet} \otimes E_{n_D}\right)$ on $B$ and $C$ blocks and $Z\left((E_{k_F})_{j,\bullet} \otimes H^D_Z\right)$ on $C$ blocks with the eigenvalue $+1$. 

We detect errors with two sets of checks. First, let $\eta \in \mathbb{F}_2^{r_F}$ be the outcome of a $Z(H_F)$ measurement, it satisfies $\eta \in \mathrm{colsp} H_F$ if the circuit is error-free. Therefore, we can detect errors that flip the measurement outcome through $H^F_M \eta$, where $H^F_M$ is the check matrix of the code $\mathrm{colsp} H_F$. Second, because qubits in $C$ and $D$ blocks are initialized in the state $\ket{0}$, outcomes of $Z(H_F)$ measurements on $B$ and $C$ blocks as well as transversal $Z$ measurements on $D$ blocks satisfy certain equations if the circuit is error-free. Accordingly, the check matrix is 
\begin{eqnarray}
H^{sp}_Z &=& \left(\begin{array}{cccc}
H^{sp}_{Z,B} & H^{sp}_{Z,C} & H^{sp}_{Z,D} & H^{sp}_{Z,mea}
\end{array}\right),
\end{eqnarray}
where 
\begin{eqnarray}
H^{sp}_{Z,B} &=& \left(\begin{array}{cccccc}
0 & 0 & 0 & 0 & 0 & 0 \\
0 & 0 & 0 & 0 & 0 & 0 \\
H_F \otimes E_{n_D} & H_F \otimes E_{n_D} & H_F \otimes E_{n_D} & H_F \otimes E_{n_D} & 0 & 0 \\
0 & 0 & 0 & 0 & 0 & 0
\end{array}\right), \\
H^{sp}_{Z,C} &=& \left(\begin{array}{cccccc}
0 & 0 & 0 & 0 & 0 & 0 \\
0 & 0 & 0 & 0 & 0 & 0 \\
0 & H_F \otimes E_{n_D} & H_F \otimes E_{n_D} & H_F \otimes E_{n_D} & 0 & 0 \\
0 & 0 & H_F \otimes H^D_Z & H_F \otimes H^D_Z & 0 & 0
\end{array}\right), \\
H^{sp}_{Z,D} &=& \left(\begin{array}{ccc}
0 & 0 & 0 \\
0 & 0 & 0 \\
0 & 0 & 0 \\
H_F \otimes E_{r^D_Z} & H_F \otimes E_{r^D_Z} & H_F \otimes E_{r^D_Z}
\end{array}\right),
\end{eqnarray}
and
\begin{eqnarray}
H^{sp}_{Z,mea} &=& \left(\begin{array}{cc}
H^F_M \otimes E_{n_D} & 0 \\
0 & H^F_M \otimes E_{n_D} \\ 
E_{r_F} \otimes E_{n_D} & E_{r_F} \otimes E_{n_D} \\
0 & E_{r_F} \otimes H^D_Z
\end{array}\right).
\end{eqnarray}
We represent $X$ errors in spacetime with a vector 
\begin{eqnarray}
e^{sp}_X &=& \left(\begin{array}{cccc}
e^{sp}_{X,B} & e^{sp}_{X,C} & e^{sp}_{X,D} & e^{sp}_{X,mea}
\end{array}\right),
\end{eqnarray}
where 
\begin{eqnarray}
e^{sp}_{X,B/C} &=& \left(\begin{array}{cccccc}
u_{B/C,1} & u_{B/C,2} & u_{B/C,3} & u_{B/C,4} & u_{B/C,5} & u_{B/C,6}
\end{array}\right),
\end{eqnarray} 
\begin{eqnarray}
e^{sp}_{X,D} &=& \left(\begin{array}{ccc}
u_{D,1} & u_{D,2} & u_{D,3}
\end{array}\right).
\end{eqnarray}
and 
\begin{eqnarray}
e^{sp}_{X,mea} &=& \left(\begin{array}{cc}
v_B & v_C
\end{array}\right).
\end{eqnarray}
The errors flip operators during the propagation according to 
\begin{eqnarray}
J^{sp}_{Z,j}{e^{sp}_X}^\mathrm{T} &=& \left(\begin{array}{c}
\mathrm{vec}\left((U_B + U'_B + U_C + U'_C) [({G_F^\mathrm{r}}^\mathrm{T})_{j,\bullet}]^\mathrm{T}\right) \\
\mathrm{vec}\left((H^D_Z U_C + H^D_Z U'_C + U_D) [({G_F^\mathrm{r}}^\mathrm{T})_{j,\bullet}]^\mathrm{T}\right)
\end{array}\right) + \left(\begin{array}{c}
\mathrm{vec}\left((U''_B + U''_C) [(E_{k_F})_{j,\bullet}]^\mathrm{T}\right) \\
\mathrm{vec}\left(H^D_Z U''_C [(E_{k_F})_{j,\bullet}]^\mathrm{T}\right)
\end{array}\right),
\label{eq:jZeX}
\end{eqnarray}
where 
\begin{eqnarray}
\mathrm{vec}(U_B) &=& (u_{B,1} + u_{B,2} + u_{B,3} + u_{B,4} + u_{C,2})^\mathrm{T}, \\
\mathrm{vec}(U'_B) &=& u_{B,5}^\mathrm{T}, \\
\mathrm{vec}(U''_B) &=& u_{B,6}^\mathrm{T}, \\
\mathrm{vec}(U_C) &=& (u_{C,3} + u_{C,4})^\mathrm{T}, \\
\mathrm{vec}(U'_C) &=& u_{C,5}^\mathrm{T}, \\
\mathrm{vec}(U''_C) &=& u_{C,6}^\mathrm{T}, \\
\mathrm{vec}(U_D) &=& (u_{D,1} + u_{D,2} + u_{D,3})^\mathrm{T}. \\
\end{eqnarray}

If the spacetime error $e^{sp}_X$ is undetectable by $H^{sp}_Z$, it satisfies the condition $H^{sp}_Z {e^{sp}_X}^\mathrm{T} = 0$. Let 
\begin{eqnarray}
\mathrm{vec}(V_{B/C}) &=& v_{B/C}^\mathrm{T},
\end{eqnarray}
the condition can be rewritten as 
\begin{eqnarray}
V_{B/C} {H^F_M}^\mathrm{T} &=& 0, \label{eq:VBC} \\
(U_B + U_C) H_F^\mathrm{T} &=& V_B + V_C \\
H^D_Z U_C H_F^\mathrm{T} + U_D H_F^\mathrm{T} &=& H^D_Z V_C.
\end{eqnarray}

When Eq.~(\ref{eq:VBC}) holds, there exists $W_{B/C} \in \mathbb{F}_2^{n_D \times n_F}$ such that $V_{B/C} = W_{B/C} H_F^\mathrm{T}$, and 
\begin{eqnarray}
\vert (W_{B/C})_{a,\bullet} \vert &\leq& \frac{n_F}{r_Fs} \vert (V_{B/C})_{a,\bullet} \vert
\end{eqnarray}
for all $a = 1,2,\ldots,n_D$; see Lemma~\ref{lem:LTC}. Then, we have 
\begin{eqnarray}
(U_B + U_C) H_F^\mathrm{T} &=& (W_B + W_C) H_F^\mathrm{T}, \\
H^D_Z U_C H_F^\mathrm{T} + U_D H_F^\mathrm{T} &=& H^D_Z W_C H_F^\mathrm{T}.
\end{eqnarray}
If the inequality (\ref{eq:espX}) is true, 
\begin{eqnarray}
\vert (U_B + U_C)_{a,\bullet} + (W_B + W_C)_{a,\bullet} \vert &\leq& \max\left\{1,\frac{n_F}{r_Fs}\right\} \times \vert e^{sp}_X \vert \leq d_F
\end{eqnarray}
and 
\begin{eqnarray}
\vert (H^D_Z U_C + U_D)_{b,\bullet} + (H^D_Z W_C)_{a,\bullet} \vert &\leq& \omega^D_Z \max\left\{1,\frac{n_F}{r_Fs}\right\} \times \vert e^{sp}_X \vert \leq d_F
\end{eqnarray}
hold for all $a = 1,2,\ldots,n_D$ and $b = 1,2,\ldots,r^D_Z$. Then, according to the F-code distance, we have 
\begin{eqnarray}
U_B + U_C &=& W_B + W_C, \\
H^D_Z U_C + U_D &=& H^D_Z W_C.
\end{eqnarray}
Substitute the above equations into Eq.~(\ref{eq:jZeX}), we obtain 
\begin{eqnarray}
J^{sp}_{Z,j}{e^{sp}_X}^\mathrm{T} &=& \left(\begin{array}{c}
\mathrm{vec}\left((W_B + W_C + U'_B + U'_C) [({G_F^\mathrm{r}}^\mathrm{T})_{j,\bullet}]^\mathrm{T}\right) \\
\mathrm{vec}\left((H^D_Z W_C + H^D_Z U'_C) [({G_F^\mathrm{r}}^\mathrm{T})_{j,\bullet}]^\mathrm{T}\right)
\end{array}\right) + \left(\begin{array}{c}
\mathrm{vec}\left((U''_B + U''_C) [(E_{k_F})_{j,\bullet}]^\mathrm{T}\right) \\
\mathrm{vec}\left(H^D_Z U''_C [(E_{k_F})_{j,\bullet}]^\mathrm{T}\right)
\end{array}\right) \notag \\
&=& H^{RS}_Z \left(\begin{array}{c}
u_{eff,B}^\mathrm{T} \\
u_{eff,C}^\mathrm{T}
\end{array}\right),
\end{eqnarray}
where 
\begin{eqnarray}
u_{eff,B/C}^\mathrm{T} &=& W_{B/C} [({G_F^\mathrm{r}}^\mathrm{T})_{j,\bullet}]^\mathrm{T} + U'_{B/C} [({G_F^\mathrm{r}}^\mathrm{T})_{j,\bullet}]^\mathrm{T} + U''_{B/C} [(E_{k_F})_{j,\bullet}]^\mathrm{T}.
\end{eqnarray}
Accordingly, the $X$ error on the prepared resource state is 
\begin{eqnarray}
e^{RS}_X &=& \left(\begin{array}{cc}
u_{eff,B} & u_{eff,C}
\end{array}\right),
\end{eqnarray}
such that $H^{RS}_Z {e^{RS}_X}^\mathrm{T} = J^{sp}_{Z,j} {e^{sp}_X}^\mathrm{T}$. Note that 
\begin{eqnarray}
\vert u_{eff,B} \vert  + \vert u_{eff,C} \vert &\leq& \max\left\{1,\frac{n_F}{r_Fs}\right\} \times \vert e^{sp}_X \vert.
\label{equ:efficentXerror}
\end{eqnarray}
The lemma has been proved. 

\subsection{Surface code in the state preparation}

\begin{figure}[htbp]
\centering
\includegraphics[width=\linewidth]{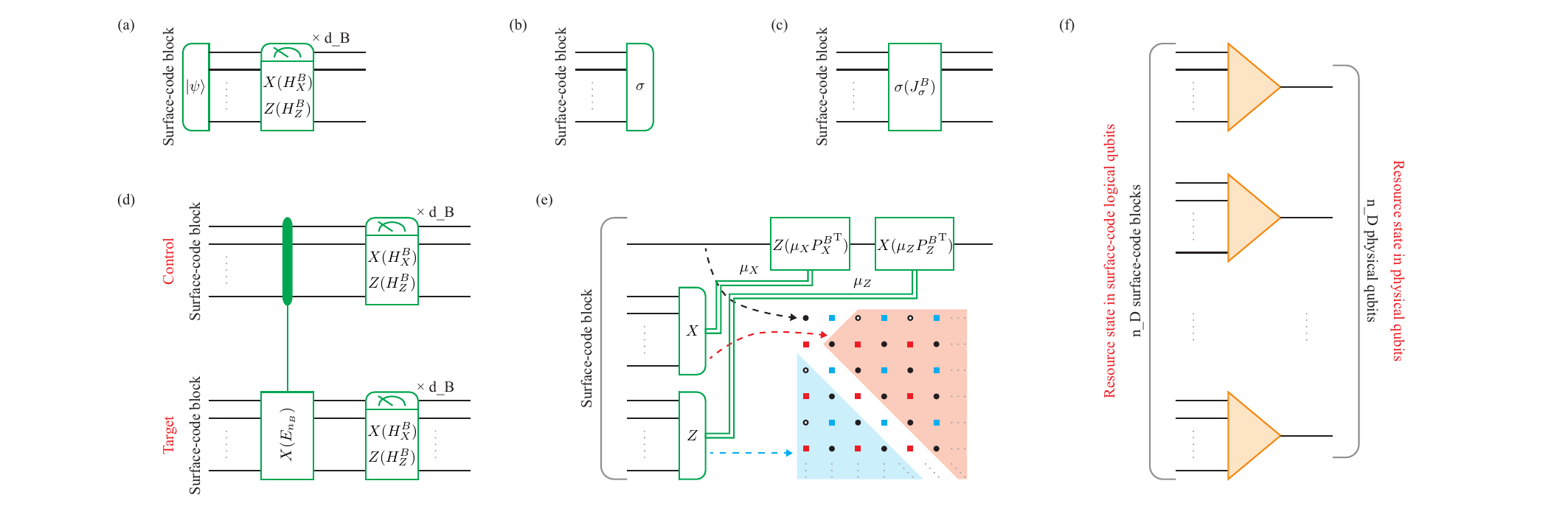}
\caption{
Operations on surface-code logical qubits. 
(a) Initialization in the state $\ket{\psi} = \ket{0},\ket{+}$. 
(b) Measurement in the basis $\sigma = X,Z$. 
(c) Logical Pauli gate $\sigma = X,Z$. Vectors $J^B_X$ and $J^B_Z$ represent $X$ and $Z$ logical operators of the surface-code logical qubit. 
(d) Logical controlled-NOT gate. 
(e) Decoding on a surface-code logical qubit. Black circles denote qubits. Red and blue squares denote $X$ and $Z$ checks, respectively; edges of the Tanner graph have been neglected. Qubits in the area marked in red (blue) are measured in the $X$ ($Z$) basis, and the measurement outcome is the vector $\mu_X$ ($\mu_Z$). In feedback gates, $P^B_X$ and $P^B_X$ are vectors, and their supports are marked with open circles: In the vector $P^B_X$ ($P^B_Z$), elements corresponding to open circles in the red (blue) area are ones, and other elements are zeros. 
(f) Decoding of the resource state. Triangles denote decoding operations on surface-code logical qubits, and each of them is realized with the circuit in (e). 
}
\label{fig:surface_code}
\end{figure}

If we prepare resource states with the circuit in Fig.~\ref{fig:state_preparation}, the F code can only correct $X$ errors, and $Z$ errors may accumulate and damage the prepared resource states. To solve this problem, we protect each qubit in the circuit with a surface code. 

We simulate the circuit in Fig.~\ref{fig:state_preparation} using the surface code. Now, each qubit in the circuit becomes a surface-code logical qubit. The circuit involves the following operations: initializations in states $\ket{0}$ and $\ket{+}$, measurements in the $X$ and $Z$ bases, Pauli gates, and controlled-NOT gate. We can implement these operations on surface-code logical qubits in the transversal way; see Figs.~\ref{fig:surface_code}(a-d). For initializations and the controlled-NOT gate, we measure surface-code stabilizer operators for $d_T = \Theta{(d_S)}$ rounds after each logical operation to correct the errors. With these logical operations, we can simulate the circuit in Fig.~\ref{fig:state_preparation} and prepare resource states encoded in surface-code logical qubits. 

Next, to obtain resource states in physical qubits, we have to decode the surface-code logical qubits. This operation can be realized by measuring physical qubits in a surface-code block in $X$ and $Z$ bases, leaving only one physical qubit unmeasured; see Fig.~\ref{fig:surface_code}(e). After decoding each surface-code logical qubit in the resource state, we obtain the eventual physical-qubit resource state; see Fig.~\ref{fig:surface_code}(f). 

Finally, we analyze the impact of using the surface code on conclusions in Lemmas~\ref{lem:spX}~and~\ref{lem:spZ}. On surface-code logical qubits, we implement a logical operation using a set of physical operations. The logical operation has errors only if physical operations have errors: Let $w_L$ and $w_P$ be the weights of logical errors and physical errors, respectively; then $w_L = 0$ if $w_P = 0$. For single-qubit logical operations, i.e.~initializations, measurements, and Pauli gates, we have $w_L \leq w_P$: Note that $w_L$ is either zero or one. For a logical controlled-NOT gate, because it acts on two surface-code logical qubits, $w_L \leq 2$, and we have $w_L \leq 2 w_P$. Multiple logical controlled-NOT gates constitute the generalized transversal controlled-NOT gates in Fig.~\ref{fig:state_preparation}. The error propagation in the second generalized transversal controlled-NOT gate may transform an error on one surface-code logical qubit into an error on multiple surface-code logical qubits, i.e.~increase the weight of logical errors. Because $H^D_Z$ has rows and columns of weight at most $\omega^D_Z$, the logical-error weight is amplified by a factor of at most $\omega^D_Z$. Therefore, for all these logical operations, we have $w_L \leq 2 \omega^D_Z w_P$. Decoding operations may also cause errors, which are the same as single-qubit logical operations; they cause extra errors, and we can take into account decoding errors in the sixth column in Fig.~\ref{fig:state_preparation}. 

\begin{lemma}[Formal version of Lemma~\ref{lem:LTSP}]
Suppose the F code is $(\omega,s)$-locally testable, and $H^D_Z$ has rows and columns of weight at most $\omega^D_Z$. Let $\bar{e}^{sp}_X$ and $\bar{e}^{sp}_Z$ be the $X$ and $Z$ spacetime (physical) errors in the state preparation circuit shown in Fig.~\ref{fig:state_preparation}, when simulated using a surface code of arbitrary code distance. If the spacetime errors are undetectable and satisfy 
\begin{eqnarray}
\vert \bar{e}^{sp}_X \vert &<& \frac{d_F}{2{\omega^D_Z}^2 \max\left\{1,\frac{n_F}{r_Fs}\right\}},
\end{eqnarray}
the following inequalities hold: 
\begin{eqnarray}
\vert e^{RS}_X \vert &\leq& 2\omega^D_Z \max\left\{1,\frac{n_F}{r_Fs}\right\} \times \vert \bar{e}^{sp}_X \vert, \\
\vert e^{RS}_Z \vert &\leq& 2\omega^D_Z \vert \bar{e}^{sp}_Z \vert.
\end{eqnarray}
Here, $\frac{n_F}{r_F s} \leq \frac{w}{s}$, i.e., $w r_F \geq n_F$; otherwise the distance reduces to $d_F = 1$.
\label{lem:sp}
\end{lemma}

\subsection{Overhead in the state preparation}

The state preparation circuit shown in Fig.~\ref{fig:state_preparation} involves $2n_D+r^D_Z$ F-code blocks, so the number of data qubits is $(2n_D+r^D_Z)n_F$. Additionally, we require $r_F$ ancilla qubits to implement parity-check measurements on each F-code block, and such measurements are applied to $2n_D$ blocks. Thus, the total number of qubits is $(2n_D+r^D_Z)n_F+2n_Dr_F = O(n_Dn_F)$. In this circuit, all operation layers have depth one, except for two layers of parity-check measurements: the layer between the second and third columns, and the layer between the fourth and fifth columns. According to Lemma~\ref{lem:depth_error}, each of these layers has depth $O(1)$ because the deformed code is a qLDPC code and the F code is an LDPC code. Therefore, the overall circuit depth is $O(1)$. 

When encoding the state preparation circuit using a surface code, each qubit in the circuit is encoded in a surface code of distance $d_S$, requiring $(2d_S - 1)^2 = O(d_S^2)$ physical qubits. Since the circuit involves $O(n_D n_F)$ qubits, the total number of physical qubits is $O(n_D n_F d_S^2)$. Each operation on surface-code logical qubits is implemented using a physical circuit of depth $O(d_T)$. Therefore, the overall depth of the physical circuit for state preparation is $O(d_T) = O(d_S)$. 

\section{Parallelized code surgery realized through locally-testable state preparation}
\label{app:PCS+LTS}

\subsection{Parity-check measurements}

\begin{figure}[htbp]
\centering
\includegraphics[width=\linewidth]{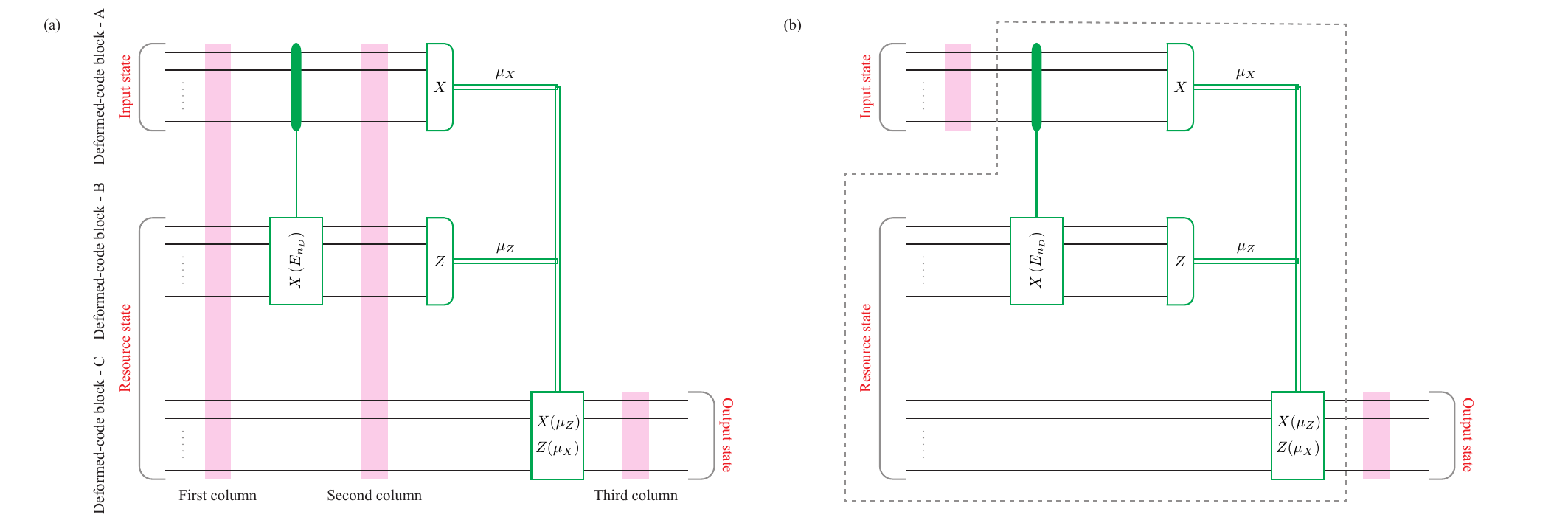}
\caption{
(a) Measurement of $Z(H^D_Z)$ based on the resource state. Outcomes of transversal measurements on blocks $A$ and $B$ are $\mu_X,\mu_Z \in \mathbb{F}_2^{n_D}$, respectively. Outcome of the $Z(H^D_Z)$ measurement is $H^D_Z \mu_Z$. 
(b) Effective error in the measurement of $Z(H^D_Z)$ based on the resource state. Operations in the dashed box are effectively error-free. 
}
\label{fig:measurement}
\end{figure}

With the resource state prepared on blocks $B$ and $C$, we can transfer the state of block-$A$ to block-$C$ and effectively realize the $Z(H^D_Z)$ measurement; see Fig.~\ref{fig:measurement}. 

\subsubsection{$X$ operators and $Z$ errors}

The propagation of $X$ operators in spacetime is described by the generator matrix 
\begin{eqnarray}
J^{pc}_X &=& \left(\begin{array}{ccc}
J^{pc}_{X,A} & J^{pc}_{X,B} & J^{pc}_{X,C}
\end{array}\right),
\end{eqnarray}
where 
\begin{eqnarray}
J^{pc}_{X,A} &=& \left(\begin{array}{cc}
H^D_X & H^D_X \\
J^D_X & J^D_X
\end{array}\right),
\end{eqnarray}
\begin{eqnarray}
J^{pc}_{X,B} &=& \left(\begin{array}{cc}
H^D_X & 0 \\
J^D_X & 0
\end{array}\right),
\end{eqnarray}
and 
\begin{eqnarray}
J^{pc}_{X,C} &=& \left(\begin{array}{ccc}
H^D_X & H^D_X & H^D_X \\
J^D_X & J^D_X & J^D_X
\end{array}\right),
\end{eqnarray}
correspond to spacetime locations on $A$, $B$, and $C$ blocks, respectively: The three columns correspond to spacetime locations as illustrated in Fig.~\ref{fig:measurement}. This generator matrix describes that operators $X(H^D_X) \otimes X(H^D_X) \otimes X(H^D_X)$ and $X(J^D_X) \otimes X(J^D_X) \otimes X(J^D_X)$ on three blocks at the beginning are transformed into operators $X(H^D_X)$ and $X(J^D_X)$ on block-$C$ at the end, respectively. Note that blocks $B$ and $C$ are prepared in the resource state, which is an eigenstate of operators $X(H^D_X) \otimes X(H^D_X)$ and $X(J^D_X) \otimes X(J^D_X)$ with the eigenvalue $+1$. Therefore, such a process effectively transforms operators $X(H^D_X)$ and $X(J^D_X)$ on block-$A$ into operators $X(H^D_X)$ and $X(J^D_X)$ on block-$C$, respectively. 

We represent $Z$ errors in spacetime with a vector 
\begin{eqnarray}
e^{pc}_Z &=& \left(\begin{array}{ccc}
e^{pc}_{Z,A} & e^{pc}_{Z,B} & e^{pc}_{Z,C}
\end{array}\right)
\end{eqnarray}
where 
\begin{eqnarray}
e^{pc}_{Z,A/B} &=& \left(\begin{array}{cc}
u_{A/B,1} & u_{A/B,2}
\end{array}\right)
\end{eqnarray}
and 
\begin{eqnarray}
e^{pc}_{Z,C} &=& \left(\begin{array}{ccc}
u_{C,1} & u_{C,2} & u_{C,3}
\end{array}\right).
\end{eqnarray}
The errors flip operators during the propagation according to 
\begin{eqnarray}
J^{pc}_X{e^{pc}_Z}^\mathrm{T} &=& \left(\begin{array}{c}
H^D_X \\
J^D_X
\end{array}\right) u_{eff}^\mathrm{T},
\end{eqnarray}
where 
\begin{eqnarray}
u_{eff} &=& u_{A,1} + u_{A,2} + u_{B,1} + u_{C,1} + u_{C,2} + u_{C,3}.
\label{eq:u_eff}
\end{eqnarray}

\begin{lemma}
Let $e^{pc}_Z$ be the $Z$ spacetime error in the $Z(H^D_Z)$ measurement based on the resource state. There exists an effective spacetime error in the form 
\begin{eqnarray}
e^{pc}_{Z,eff} &=& \left(\begin{array}{ccc}
0 & 0 & e^{pc}_{Z,C,eff}
\end{array}\right)
\end{eqnarray}
where 
\begin{eqnarray}
e^{pc}_{Z,C,eff} &=& \left(\begin{array}{ccc}
0 & 0 & u_{eff}
\end{array}\right),
\end{eqnarray}
such that the effective error $e^{pc}_{Z,eff}$ is equivalent to the error $e^{pc}_Z$, i.e.~$J^{pc}_X{e^{pc}_{Z,eff}}^\mathrm{T} = J^{pc}_X{e^{pc}_Z}^\mathrm{T}$, and $\vert e^{pc}_{Z,eff} \vert \leq \vert e^{pc}_Z \vert$. In the effective spacetime error, errors only occur on block-$C$ at the end. 
\label{lem:MeasurementZ}
\end{lemma}

\begin{proof}
In the effective spacetime error, we take $u_{eff}$ according to Eq.~(\ref{eq:u_eff}). Then, $J^{pc}_X{e^{pc}_{Z,eff}}^\mathrm{T} = J^{pc}_X{e^{pc}_Z}^\mathrm{T}$ is true, and $\vert e^{pc}_{Z,eff} \vert \leq \vert e^{pc}_Z \vert$. 
\end{proof}

\subsubsection{$Z$ operators and $X$ errors}

The propagation of unmeasured $Z$ operators in spacetime is described by the generator matrix 
\begin{eqnarray}
J^{pc}_Z &=& \left(\begin{array}{ccc}
J^{pc}_{Z,A} & J^{pc}_{Z,B} & J^{pc}_{Z,C}
\end{array}\right),
\end{eqnarray}
where 
\begin{eqnarray}
J^{pc}_{Z,A} &=& \left(\begin{array}{cc}
{(H^D_X)^\mathrm{r}}^\mathrm{T} & 0 \\
J^D_Z & 0
\end{array}\right),
\end{eqnarray}
\begin{eqnarray}
J^{pc}_{Z,B} &=& \left(\begin{array}{cc}
{(H^D_X)^\mathrm{r}}^\mathrm{T} & {(H^D_X)^\mathrm{r}}^\mathrm{T} \\
J^D_Z & J^D_Z
\end{array}\right),
\end{eqnarray}
and 
\begin{eqnarray}
J^{pc}_{Z,C} &=& \left(\begin{array}{ccc}
{(H^D_X)^\mathrm{r}}^\mathrm{T} & {(H^D_X)^\mathrm{r}}^\mathrm{T} & {(H^D_X)^\mathrm{r}}^\mathrm{T} \\
J^D_Z & J^D_Z & J^D_Z
\end{array}\right).
\end{eqnarray}
Similar to $J^{pc}_X$, the generator matrix $J^{pc}_Z$ describes the transformation of operators $Z\left({(H^D_X)^\mathrm{r}}^\mathrm{T}\right)$ and $Z(J^D_Z)$ on block-$A$ into operators $Z\left({(H^D_X)^\mathrm{r}}^\mathrm{T}\right)$ and $Z(J^D_Z)$ on block-$C$, respectively. 

The propagation of measured $Z$ operators in spacetime is described by the generator matrix 
\begin{eqnarray}
J^{pc}_{mz} &=& \left(\begin{array}{ccc}
J^{pc}_{mz,A} & J^{pc}_{mz,B} & J^{pc}_{mz,C}
\end{array}\right)
\end{eqnarray}
where 
\begin{eqnarray}
J^{pc}_{mz,A} &=& \left(\begin{array}{cc}
H^D_Z & 0
\end{array}\right),
\end{eqnarray}
\begin{eqnarray}
J^{pc}_{mz,B} &=& \left(\begin{array}{cc}
H^D_Z & H^D_Z
\end{array}\right),
\end{eqnarray}
and 
\begin{eqnarray}
J^{pc}_{mz,C} &=& \left(\begin{array}{ccc}
H^D_Z & H^D_Z & H^D_Z
\end{array}\right).
\end{eqnarray}
This generator matrix describes the transformation of operators $Z(H^D_Z)$ on block-$A$ into operators $Z(H^D_Z)$ on block-$C$. 

Outcome of the $Z(H^D_Z)$ measurement is justified by the following operator propagation, 
\begin{eqnarray}
J^{pc}_{oc} &=& \left(\begin{array}{ccc}
J^{pc}_{oc} & J^{pc}_{oc,B} & J^{pc}_{oc,C}
\end{array}\right)
\end{eqnarray}
where 
\begin{eqnarray}
J^{pc}_{oc,A} &=& \left(\begin{array}{cc}
H^D_Z & 0
\end{array}\right),
\end{eqnarray}
\begin{eqnarray}
J^{pc}_{oc,B} &=& \left(\begin{array}{cc}
H^D_Z & H^D_Z
\end{array}\right),
\end{eqnarray}
and 
\begin{eqnarray}
J^{pc}_{oc,C} &=& \left(\begin{array}{ccc}
0 & 0 & 0
\end{array}\right).
\end{eqnarray}
Note that the resource state is an eigenstate of operators $Z(H^D_Z)$ on block-$B$. This generator matrix describes the transformation of operators $Z(H^D_Z)$ on block-$A$ into operators $Z(H^D_Z)$ on block-$B$, which is subsequently measured in the transversal measurement on block-$B$. 

We represent $X$ errors in spacetime with a vector 
\begin{eqnarray}
e^{pc}_X &=& \left(\begin{array}{ccc}
e^{pc}_{X,A} & e^{pc}_{X,B} & e^{pc}_{X,C}
\end{array}\right)
\end{eqnarray}
where 
\begin{eqnarray}
e^{pc}_{X,A/B} &=& \left(\begin{array}{cc}
u_{A/B,1} & u_{A/B,2}
\end{array}\right)
\end{eqnarray}
and 
\begin{eqnarray}
e^{pc}_{X,C} &=& \left(\begin{array}{ccc}
u_{C,1} & u_{C,2} & u_{C,3}
\end{array}\right).
\end{eqnarray}
The errors flip operators and measurement outcomes during the propagation according to 
\begin{eqnarray}
J^{pc}_Z{e^{pc}_X}^\mathrm{T} &=& \left(\begin{array}{c}
{(H^D_X)^\mathrm{r}}^\mathrm{T} \\
J^D_Z
\end{array}\right) (u_{A,eff} + u_{C,eff})^\mathrm{T}, \\
J^{pc}_{mz}{e^{pc}_X}^\mathrm{T} &=& H^D_Z (u_{A,eff} + u_{C,eff})^\mathrm{T}, \\
J^{pc}_{oc}{e^{pc}_X}^\mathrm{T} &=& H^D_Z u_{A,eff}^\mathrm{T},
\end{eqnarray}
where 
\begin{eqnarray}
u_{A,eff} &=& u_{A,1} + u_{B,1} + u_{B,2}, \label{eq:u_A_eff} \\
u_{C,eff} &=& u_{C,1} + u_{C,2} + u_{C,3}.
\label{eq:u_C_eff}
\end{eqnarray}

\begin{lemma}
\label{lem:MeasurementX}
Let $e^{pc}_X$ be the $X$ spacetime error in the $Z(H^D_Z)$ measurement based on the resource state. There exists an effective spacetime error in the form 
\begin{eqnarray}
e^{pc}_{X,eff} &=& \left(\begin{array}{ccc}
e^{pc}_{X,A,eff} & 0 & e^{pc}_{X,C,eff}
\end{array}\right)
\end{eqnarray}
where 
\begin{eqnarray}
e^{pc}_{X,A,eff} &=& \left(\begin{array}{ccc}
u_{A,eff} & 0
\end{array}\right), \\
e^{pc}_{X,C,eff} &=& \left(\begin{array}{ccc}
0 & 0 & u_{C,eff}
\end{array}\right),
\end{eqnarray}
such that the effective error $e^{pc}_{X,eff}$ is equivalent to the error $e^{pc}_X$, i.e. 
\begin{eqnarray}
\left(\begin{array}{c}
J^{pc}_Z \\
J^{pc}_{mz} \\
J^{pc}_{oc}
\end{array}\right) {e^{pc}_{X,eff}}^\mathrm{T} &=& \left(\begin{array}{c}
J^{pc}_Z \\
J^{pc}_{mz} \\
J^{pc}_{oc}
\end{array}\right) {e^{pc}_X}^\mathrm{T},
\label{eq:effective_X_error}
\end{eqnarray}
and $\vert e^{pc}_{X,eff} \vert \leq \vert e^{pc}_X \vert$. In the effective spacetime error, errors only occur on block-$A$ at the beginning and block-$C$ at the end. 
\label{lem:MZ}
\end{lemma}

\begin{proof}
In the effective spacetime error, we take $u_{A,eff}$ and $u_{C,eff}$ according to Eqs.~(\ref{eq:u_A_eff})~and~(\ref{eq:u_C_eff}). Then, Eq.~(\ref{eq:effective_X_error}) is true, and $\vert e^{pc}_{X,eff} \vert \leq \vert e^{pc}_X \vert$. 
\end{proof}

\subsection{Code surgery}

\begin{figure}[htbp]
\centering
\includegraphics[width=\linewidth]{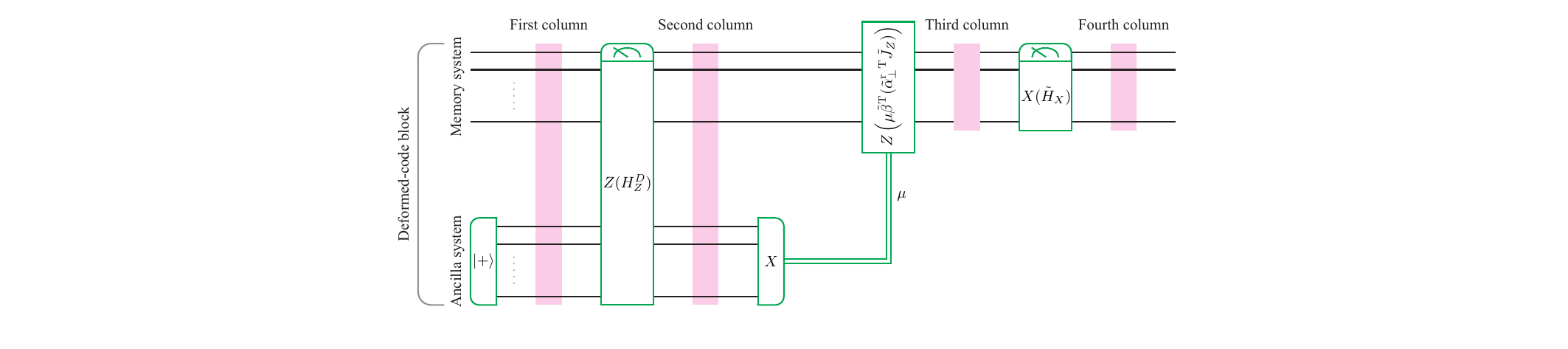}
\caption{
Code surgery. 
}
\label{fig:lattice_surgery}
\end{figure}

The circuit of code surgery is illustrated in Fig.~\ref{fig:lattice_surgery}, which is applied on $k_F$ blocks of the code $(H_X,H_Z,J_X,J_Z)$, called memory system. Check and generator matrices of the memory system are $(\tilde{H}_X,\tilde{H}_Z,\tilde{J}_X,\tilde{J}_Z)$. This code surgery realizes the measurement of $Z(\tilde{\alpha} \tilde{J}_Z)$. 

The input state of the memory system may carry errors. We model them as errors occurring at locations on the memory system in the first column; see Fig.~\ref{fig:lattice_surgery}. Accordingly, the effective input state (the state of the memory system before the first column) is in the codeword subspace defined by stabilizer operators $X(\tilde{H}_X)$ and $Z(\tilde{H}_Z)$. 

We realize the measurements of $Z(H^D_Z)$ and $X(\tilde{H}_X)$ with corresponding resource states. According to Lemmas~\ref{lem:MeasurementZ}~and~\ref{lem:MZ}, errors effectively occur on the input and output states of these measurements, i.e.~outcomes of these measurements are effectively error-free; see Fig.~\ref{fig:measurement}(b). 

\subsubsection{$X$ operators and $Z$ errors}

The propagation of $X$ operators in spacetime is described by the generator matrix 
\begin{eqnarray}
J^{cs}_X &=& \left(\begin{array}{ccc}
J^{cs}_{X,M} & J^{cs}_{X,A} & J^{cs}_{X,mea}
\end{array}\right),
\end{eqnarray}
where 
\begin{eqnarray}
J^{cs}_{X,M} &=& \left(\begin{array}{cccc}
\tilde{\alpha}_\perp \tilde{J}_X & \tilde{\alpha}_\perp \tilde{J}_X & \tilde{\alpha}_\perp \tilde{J}_X & \tilde{\alpha}_\perp \tilde{J}_X
\end{array}\right)
\end{eqnarray}
and 
\begin{eqnarray}
J^{cs}_{X,A} &=& \left(\begin{array}{cc}
\tilde{\beta} & \tilde{\beta}
\end{array}\right)
\end{eqnarray}
correspond to spacetime locations on the memory and ancilla systems, respectively, and 
\begin{eqnarray}
J^{cs}_{X,mea} &=& 0
\end{eqnarray}
corresponds to the measurement of $X(\tilde{H}_X)$. The columns in $J^{cs}_{X,M}$ and $J^{cs}_{X,A}$ correspond to spacetime locations as illustrated in Fig.~\ref{fig:lattice_surgery}. This generator matrix describes that operators $X(\tilde{\alpha}_\perp \tilde{J}_X) \otimes X(\tilde{\beta})$ on two systems at the beginning are transformed into operators $X(\tilde{\alpha}_\perp \tilde{J}_X)$ on the memory system at the end. Note that operators $X(\tilde{\alpha}_\perp \tilde{J}_X) \otimes X(\tilde{\beta}) = X(J^D_X)$ commute with the measurement of $Z(H^D_Z)$. Because qubits in the ancilla system are prepared in the state $\ket{+}$, such a process effectively implements a trivial transformation on operators $X(\tilde{\alpha}_\perp \tilde{J}_X)$. 

We detect errors in the following way. Let $\tilde{\nu} \in \mathbb{F}_2^{k_F r_X}$ be the measurement outcome of $X(\tilde{H}_X)$. We introduce the vector 
\begin{eqnarray}
\nu &=& \left(\begin{array}{cc}
\tilde{\nu} + \mu \tilde{T}^\mathrm{T} & \mu \tilde{H}_M^\mathrm{T}
\end{array}\right).
\end{eqnarray}
Then, $(-1)^{\nu_j}$ is the eigenvalue of the operator $X\left((H^D_X)_{j,\bullet}\right)$. Note that operators $X(\tilde{H}_X)$ [$X(H^D_X)$] commute with the feedback gate [measurement of $Z(H^D_X)$]. Therefore, we have $\nu = 0$ if the circuit is error-free. Accordingly, the check matrix is 
\begin{eqnarray}
H^{cs}_X &=& \left(\begin{array}{ccc}
H^{cs}_{X,M} & H^{cs}_{X,A} & H^{cs}_{X,mea}
\end{array}\right),
\end{eqnarray}
where 
\begin{eqnarray}
H^{cs}_{X,M} &=& \left(\begin{array}{cccc}
\tilde{H}_X & \tilde{H}_X & \tilde{H}_X & 0 \\
0 & 0 & 0 & 0
\end{array}\right), \\
H^{cs}_{X,A} &=& \left(\begin{array}{cc}
\tilde{T} & \tilde{T} \\
\tilde{H}_M & \tilde{H}_M
\end{array}\right),
\end{eqnarray}
and 
\begin{eqnarray}
H^{cs}_{X,mea} &=& \left(\begin{array}{c}
E_{k_F r_X} \\
0
\end{array}\right).
\end{eqnarray}

We represent $Z$ errors in spacetime with the vector 
\begin{eqnarray}
e^{cs}_Z &=& \left(\begin{array}{ccc}
e^{cs}_{Z,M} & e^{cs}_{Z,A} & 0
\end{array}\right)
\end{eqnarray}
represents errors before the $X(\tilde{H}_X)$ measurement, where 
\begin{eqnarray}
e^{cs}_{Z,M} &=& \left(\begin{array}{cccc}
u_{M,1} & u_{M,2} & u_{M,3} & u_{M,4}
\end{array}\right)
\end{eqnarray}
and 
\begin{eqnarray}
e^{cs}_{Z,A} &=& \left(\begin{array}{cc}
u_{A,1} & u_{A,2}
\end{array}\right).
\end{eqnarray}
We decompose this spacetime error into two components: 
\begin{eqnarray}
e^{cs-b}_Z &=& \left(\begin{array}{ccccccc}
u_{M,1} & u_{M,2} & u_{M,3} & 0 & u_{A,1} & u_{A,2} & 0
\end{array}\right)
\end{eqnarray}
represents errors occurring before the $X(\tilde{H}_X)$ measurement, and 
\begin{eqnarray}
e^{cs-a}_Z &=& \left(\begin{array}{ccccccc}
0 & 0 & 0 & u_{M,4} & 0 & 0 & 0
\end{array}\right)
\end{eqnarray}
represents errors occurring after the $X(\tilde{H}_X)$ measurement. Note that $e^{cs}_Z = e^{cs-b}_Z + e^{cs-a}_Z$. The errors flip operators during the propagation according to 
\begin{eqnarray}
J^{cs}_X {e^{cs}_Z}^\mathrm{T} &=& J^D_X u_{eff}^\mathrm{T} + (\tilde{\alpha}_\perp \tilde{J}_X) u_{res}^\mathrm{T},
\end{eqnarray}
where 
\begin{eqnarray}
u_{eff} &=& \left(\begin{array}{cc}
u_{M,1} + u_{M,2} + u_{M,3} & u_{A,1} + u_{A,2}
\end{array}\right), \\
u_{res} &=& u_{M,4}.
\end{eqnarray}

If the spacetime error $e^{cs}_Z$ is undetectable through $H^{cs}_X$, it satisfies the condition $H^{cs}_X {e^{cs}_Z}^\mathrm{T} = 0$. This condition can be rewritten as $H^D_X u_{eff}^\mathrm{T} = 0$. 

\begin{lemma}
Let $e^{cs-b}_Z$ and $e^{cs-a}_Z$ be the spacetime error vectors representing $Z$ errors that occur before and after the $X(\tilde{H}_X)$ measurement, respectively, during code surgery. If errors are undetectable and satisfy $\vert e^{cs-b}_Z \vert < d_D$, then the weight of the residual $Z$ errors on the output state of the memory system is upper bounded by $\vert e^{cs-p}_Z \vert$. 
\label{lem:csmZ}
\end{lemma}

\begin{proof}
Note that $d_D$ is the distance of the deformed code. Because $\vert u_{eff} \vert \leq \vert e^{cs-b}_Z \vert < d_D$ and $H^D_X u_{eff}^\mathrm{T} = 0$, we have $J^D_X u_{eff}^\mathrm{T} = 0$. Then, 
\begin{eqnarray}
J^{cs}_X {e^{cs}_Z}^\mathrm{T} &=& (\tilde{\alpha}_\perp \tilde{J}_X) u_{res}^\mathrm{T},
\end{eqnarray}
i.e.~the vector $u_{res}$ represents the residual errors, and $\vert u_{res} \vert = \vert e^{cs-a}_Z \vert$. 
\end{proof}

\subsubsection{$Z$ operators and $X$ errors}

The propagation of unmeasured $Z$ operators in spacetime is described by the generator matrix 
\begin{eqnarray}
J^{cs}_Z &=& \left(\begin{array}{ccc}
J^{cs}_{Z,M} & J^{cs}_{Z,A} & J^{cs}_{Z,mea}
\end{array}\right),
\end{eqnarray}
where 
\begin{eqnarray}
J^{cs}_{Z,M} &=& \left(\begin{array}{cccc}
\tilde{\alpha}_\perp^\mathrm{r}{}^\mathrm{T} \tilde{J}_Z & \tilde{\alpha}_\perp^\mathrm{r}{}^\mathrm{T} \tilde{J}_Z & \tilde{\alpha}_\perp^\mathrm{r}{}^\mathrm{T} \tilde{J}_Z & \tilde{\alpha}_\perp^\mathrm{r}{}^\mathrm{T} \tilde{J}_Z
\end{array}\right)
\end{eqnarray}
and 
\begin{eqnarray}
J^{cs}_{Z,A} &=& \left(\begin{array}{cc}
0 & 0
\end{array}\right)
\end{eqnarray}
correspond to spacetime locations on the memory and ancilla systems, respectively, and 
\begin{eqnarray}
J^{cs}_{Z,mea} &=& 0
\end{eqnarray}
corresponds to the measurement of $Z(H^D_Z)$. Note that operators $Z(\tilde{\alpha}_\perp^\mathrm{r}{}^\mathrm{T} \tilde{J}_Z)$ commute with the measurement of $X(\tilde{H}_X)$. 

Similarly, the propagation of measured $Z$ operators in spacetime is described by the generator matrix 
\begin{eqnarray}
J^{cs}_{mz} &=& \left(\begin{array}{ccc}
J^{cs}_{mz,M} & J^{cs}_{mz,A} & J^{cs}_{mz,mea}
\end{array}\right),
\end{eqnarray}
where 
\begin{eqnarray}
J^{cs}_{mz,M} &=& \left(\begin{array}{cccc}
\tilde{\alpha} \tilde{J}_Z & \tilde{\alpha} \tilde{J}_Z & \tilde{\alpha} \tilde{J}_Z & \tilde{\alpha} \tilde{J}_Z
\end{array}\right), 
\end{eqnarray}
\begin{eqnarray}
J^{cs}_{mz,A} &=& \left(\begin{array}{cc}
0 & 0
\end{array}\right), 
\end{eqnarray}
and 
\begin{eqnarray}
J^{cs}_{mz,mea} &=& 0.
\end{eqnarray}

Outcome of the $Z(\tilde{\alpha} \tilde{J}_Z)$ measurement is described by the following generator matrix, 
\begin{eqnarray}
J^{cs}_{oc} &=& \left(\begin{array}{ccc}
J^{cs}_{oc,M} & J^{cs}_{oc,A} & J^{cs}_{oc,mea}
\end{array}\right),
\end{eqnarray}
where 
\begin{eqnarray}
J^{cs}_{oc,M} &=& \left(\begin{array}{cccc}
\tilde{\alpha} \tilde{J}_Z & 0 & 0 & 0
\end{array}\right), 
\end{eqnarray}
\begin{eqnarray}
J^{cs}_{oc,A} &=& \left(\begin{array}{cc}
0 & 0
\end{array}\right), 
\end{eqnarray}
and 
\begin{eqnarray}
J^{cs}_{oc,mea} &=& \tilde{\alpha} \tilde{J}_Z \tilde{R} \gamma_2.
\end{eqnarray}
Here, 
\begin{eqnarray}
\gamma_1 = \left(\begin{array}{cc}
E_{k_R r_Z} & 0
\end{array}\right)
\end{eqnarray}
and 
\begin{eqnarray}
\gamma_2 = \left(\begin{array}{cc}
0 & E_{r^D_Z - k_R r_Z}
\end{array}\right).
\end{eqnarray}
Note that the measurement of $Z(H^D_Z)$ measures operators $Z(\tilde{\alpha} \tilde{J}_Z)$ because 
\begin{eqnarray}
\tilde{\alpha} \tilde{J}_Z \tilde{R} \gamma_2 H^D_Z = \left(\begin{array}{cc}
\tilde{\alpha} \tilde{J}_Z & 0
\end{array}\right).
\end{eqnarray}
Let $\nu \in \mathbb{F}_2^{r^D_Z}$ be the outcome of the $Z(H^D_Z)$ measurement. Then, $\nu \gamma_{1/2}^\mathrm{T}$ is the outcome of $Z(\gamma_{1/2} H^D_Z)$ operators, and $\nu (\tilde{\alpha} \tilde{J}_Z \tilde{R} \gamma_2)^\mathrm{T}$ is the outcome of $Z(\tilde{\alpha} \tilde{J}_Z)$ operators. 

If the circuit is error-free, we have $\nu \gamma_1^\mathrm{T} = 1$. Accordingly, the check matrix is 
\begin{eqnarray}
H^{cs}_Z &=& \left(\begin{array}{ccc}
H^{cs}_{Z,M} & H^{cs}_{Z,A} & H^{cs}_{Z,mea}
\end{array}\right),
\end{eqnarray}
where 
\begin{eqnarray}
H^{cs}_{Z,M} &=& \left(\begin{array}{cccc}
\tilde{H}_Z & 0 & 0 & 0
\end{array}\right), \\
H^{cs}_{Z,A} &=& \left(\begin{array}{cc}
0 & 0
\end{array}\right),
\end{eqnarray}
and 
\begin{eqnarray}
H^{cs}_{Z,mea} &=& \gamma_1.
\end{eqnarray}

We represent $X$ errors in spacetime with the vector 
\begin{eqnarray}
e^{cs}_X &=& \left(\begin{array}{ccc}
e^{cs}_{X,M} & e^{cs}_{X,A} & 0
\end{array}\right)
\end{eqnarray}
represents errors before the $X(\tilde{H}_X)$ measurement, where 
\begin{eqnarray}
e^{cs}_{X,M} &=& \left(\begin{array}{cccc}
u_{M,1} & u_{M,2} & u_{M,3} & u_{M,4}
\end{array}\right)
\end{eqnarray}
and 
\begin{eqnarray}
e^{cs}_{X,A} &=& \left(\begin{array}{cc}
u_{A,1} & u_{A,2}
\end{array}\right).
\end{eqnarray}
Similar to $Z$ errors, we decompose this spacetime error into two components: 
\begin{eqnarray}
e^{cs-b}_X &=& \left(\begin{array}{ccccccc}
u_{M,1} & 0 & 0 & 0 & u_{A,1} & 0 & 0
\end{array}\right)
\end{eqnarray}
represents errors occurring before the $Z(H^D_Z)$ measurement, and 
\begin{eqnarray}
e^{cs-a}_X &=& \left(\begin{array}{ccccccc}
0 & u_{M,2} & u_{M,3} & u_{M,4} & 0 & u_{A,2} & 0
\end{array}\right)
\end{eqnarray}
represents errors occurring after the $Z(H^D_Z)$ measurement. The errors flip operators and measurement outcomes during the propagation according to 
\begin{eqnarray}
J^{cs}_Z {e^{cs}_X}^\mathrm{T} &=& (\tilde{\alpha}_\perp^\mathrm{r}{}^\mathrm{T} \tilde{J}_Z) u_{eff}^\mathrm{T} + (\tilde{\alpha}_\perp^\mathrm{r}{}^\mathrm{T} \tilde{J}_Z) u_{res}^\mathrm{T}, \\
J^{cs}_{mz} {e^{cs}_X}^\mathrm{T} &=& (\tilde{\alpha} \tilde{J}_Z) u_{eff}^\mathrm{T} + (\tilde{\alpha} \tilde{J}_Z) u_{res}^\mathrm{T}, \\
J^{cs}_{oc} {e^{cs}_X}^\mathrm{T} &=& (\tilde{\alpha} \tilde{J}_Z) u_{eff}^\mathrm{T}, \\
\end{eqnarray}
where 
where 
\begin{eqnarray}
u_{eff} &=& u_{M,1}, \\
u_{res} &=& u_{M,2} + u_{M,3} + u_{M,4}.
\end{eqnarray}

If the spacetime error $e^{cs}_X$ is undetectable through $H^{cs}_Z$, it satisfies the condition $H^{cs}_Z {e^{cs}_X}^\mathrm{T} = 0$. This condition can be rewritten as $\tilde{H}_Z u_{eff}^\mathrm{T} = 0$. 

\begin{lemma}
Let $e^{cs-b}_X$ and $e^{cs-a}_X$ be the spacetime error vectors representing $X$ errors that occur before and after the $Z(H^D_Z)$ measurement, respectively, during code surgery. If errors are undetectable and satisfy $\vert e^{cs-b}_X \vert < d$, then the measurement outcome of $Z(\tilde{\alpha} \tilde{J}_Z)$ is correct, and the weight of the residual $X$ errors on the output state of the memory system is upper bounded by $\vert e^{cs-p}_X \vert$. 
\label{lem:csmX}
\end{lemma}

\begin{proof}
Note that $d$ is the distance of the code $(H_X,H_Z,J_X,J_Z)$, which is also the distance of the code $(\tilde{H}_X,\tilde{H}_Z,\tilde{J}_X,\tilde{J}_Z)$. Because $\vert u_{eff} \vert \leq \vert e^{cs-b}_X \vert < d$ and $\tilde{H}_Z u_{eff}^\mathrm{T} = 0$, we have $\tilde{J}_Z u_{eff}^\mathrm{T} = 0$. Then, 
\begin{eqnarray}
J^{cs}_Z {e^{cs}_X}^\mathrm{T} &=& (\tilde{\alpha}_\perp^\mathrm{r}{}^\mathrm{T} \tilde{J}_Z) u_{res}^\mathrm{T}, \\
J^{cs}_{mz} {e^{cs}_X}^\mathrm{T} &=& (\tilde{\alpha} \tilde{J}_Z) u_{res}^\mathrm{T}, \\
J^{cs}_{oc} {e^{cs}_X}^\mathrm{T} &=& 0, \\
\end{eqnarray}
i.e.~the measurement outcome is correct, and the vector $u_{res}$ represents the residual errors. Regarding the weight, $\vert u_{res} \vert \leq \vert e^{cs-a}_X \vert$. 
\end{proof}

{\section{Threshold theorem}}
\label{app:threshold}

In this section, we provide the proof of the threshold theorem (Theorem~\ref{the:FTQC}) for our protocol. We begin by defining the underlying stochastic noise model. Subsequently, we analyze the primitive circuits by leveraging the technical lemmas established in Appendices~\ref{app:LTSP} and~\ref{app:PCS+LTS}. Building on these results, we finally establish and prove the existence of a fault-tolerant threshold. 

\subsection{Noise model}

\begin{definition}
{\bf Standard local stochastic noise model.} Let $C$ be the set of operations constituting a quantum circuit. Each element in $C$ is also referred to as a \textbf{location}. A fault configuration is described by a random subset $F \subseteq C$, representing locations where operations are replaced by faulty versions. Such a set $F$ is referred to as a \textbf{fault path}. The local stochastic noise model specifies that for any subset of locations $S \subseteq C$, the probability that $S$ is contained within the random fault path $F$ satisfies: 
\begin{eqnarray}
\Pr[S \subseteq F] \leq \prod_{j \in S} p_j,
\end{eqnarray}
where each location $j \in C$ has an associated error probability $p_j$. 
\end{definition}

We establish the fault-tolerance theorem under the local stochastic noise model. To enhance clarity, we first present the derivation using an intuitive, simplified noise model before demonstrating that the theorem remains valid under the standard local stochastic model. 

\subsubsection{Generalized local stochastic noise model}

As a proof technique, we introduce a \textit{generalized local stochastic noise model} to capture correlated errors and the propagation of faults. In this framework, we introduce \textit{virtual locations} where faults occur and subsequently propagate to \textit{real locations}. 

\begin{definition}
{\bf $\nu$-confined generalized local stochastic noise model.} The noise model is specified by the tuple $(A, B, f, p_\bullet, \delta, \mathfrak{con})$, where $A$ denotes the set of real locations, $B$ is the set of virtual locations, and $f: A \to \mathcal{P}(B)$ is a mapping to the power set of $B$ that describes the propagation of faults. We define the following notations to characterize this mapping:
\begin{itemize}
    \item For any subset $S \subseteq A$, let $f_\cup(S) \equiv \bigcup_{a \in S} f(a)$ denote the set of all virtual locations that can influence the real locations in $S$.
    \item Let $\|f\|_\to \equiv \max_{a \in A} |f(a)|$ be the maximum number of virtual locations associated with any single real location.
    \item Let $\|f\|_\gets \equiv \max_{b \in B} |\{a \in A \mid b \in f(a)\}|$ be the maximum number of real locations influenced by any single virtual location.
    \item A mapping $f$ is said to be \textit{disjoint} if for any $a, a' \in A$ with $a \neq a'$, the sets $f(a)$ and $f(a')$ are disjoint, i.e.,~$f(a) \cap f(a') = \emptyset$.
\end{itemize}
The model is said to be \textit{$\nu$-confined} if the mapping $f$ satisfies $\|f\|_\leftrightarrow \leq \nu$, where $\|f\|_\leftrightarrow \equiv \max\{\|f\|_\to, \|f\|_\gets\}$. 

Faults at virtual locations are distributed according to a standard local stochastic noise model, with error probabilities specified by the function $p_\bullet: B \to [0, p]$. Faults at real locations arise from two sources: 
\begin{enumerate}
\item \textbf{Intrinsic faults:} Conditioned on the event $\mathfrak{con}$, at most $\delta$ faults may occur directly at real locations. 
\item \textbf{Propagated faults:} A virtual fault at location $b \in B$ may induce a fault at any real location $a \in A$ for which $b \in f(a)$. 
\end{enumerate}

Formally, let $F_A \subseteq A$ and $F_B \subseteq B$ denote the real and virtual fault paths, respectively. For any subset of virtual locations $S \subseteq B$, the probability that $S$ is contained within the virtual fault path $F_B$ satisfies: 
\begin{eqnarray}
\Pr[S \subseteq F_B] \leq \prod_{b \in S} p_b \leq p^{|S|}.
\end{eqnarray}
The real fault path $F_A$ is always contained within the union of two sub-paths, $F_A \subseteq F_{A,I} \cup F_{A,P}$, where $F_{A,I}$ and $F_{A,P}$ represent the intrinsic and propagated fault paths, respectively. Conditioned on the event $\mathfrak{con}$, the intrinsic fault path satisfies 
\begin{eqnarray}
|F_{A,I}| \leq \delta.
\end{eqnarray}
The propagated fault path satisfies 
\begin{eqnarray}
F_{A,P} \subseteq \{a \in A \st f(a) \cap F_B \neq \emptyset\}.
\end{eqnarray}

For brevity, we often denote the generalized noise model as $(A, B, f, p, \delta, \mathfrak{con})$, replacing the mapping $p_\bullet$ with its upper bound $p$. 
\end{definition}

The standard local stochastic noise model is a $1$-confined generalized local stochastic noise model given by $(A=C,B=C,f=\mathrm{identity},p_\bullet,\delta=0,\mathfrak{con}=\mathfrak{trivial})$. 

\begin{lemma}
{\bf Probability bound under the generalized noise model.} Let $(A, B, f, p, \delta, \mathfrak{con})$ be a $\nu$-confined generalized local stochastic noise model. Let $F_A \subseteq A$ be the real fault path. Conditioned on the event $\mathfrak{con}$, for any subset of real locations $S \subseteq A$ with cardinality $s = |S| \geq \delta$, the probability that $S$ contains at least $t \geq \delta$ faults satisfies: 
\begin{align}
    \Pr[|S \cap F_A| \geq t \mid \mathfrak{con}] \leq \left( \frac{\nu s\,e}{l} p \right)^l,
\end{align}
where $l = \lceil(t-\delta)/\nu\rceil$. Furthermore, if the mapping $f$ is disjoint, the same upper bound holds with $l = \lceil t-\delta \rceil$. 
\label{lem:probability_bound}
\end{lemma}

\begin{proof}
By definition, the real fault path is partitioned such that $F_A \subseteq F_{A,I} \cup F_{A,P}$, where $F_{A,I}$ and $F_{A,P}$ represent the intrinsic and propagated fault paths, respectively. Conditioned on $\mathfrak{con}$, the intrinsic component contributes at most $\delta$ faults to the set $S$, i.e.,~$|S \cap F_{A,I}| \leq \delta$. 

Consequently, to satisfy the condition $|S \cap F_A| \geq t$, the propagated component must account for the remaining faults: 
\begin{equation}
    |S \cap F_{A,P}| \geq t - |S \cap F_{A,I}| \geq t - \delta.
\end{equation}
Since each virtual fault in $B$ can influence at most $\nu$ real locations ($\|f\|_\gets \leq \nu$), achieving at least $t - \delta$ propagated faults in $S$ requires at least $l = \lceil (t - \delta) / \nu \rceil$ distinct faults in the virtual set $f_\cup(S)$. 

The probability of having at least $l$ faults in the virtual set is upper bounded by 
\begin{equation}
    \Pr[|f_\cup(S) \cap F_B| \geq l] \leq \binom{|f_\cup(S)|}{l} p^l \leq \left( \frac{|f_\cup(S)|\,e}{l} p \right)^l.
\end{equation}
Using the confinement property $\|f\|_\to \leq \nu$, we have $|f_\cup(S)| \leq \nu s$. Substituting this into the inequality yields the inequality in the lemma. 

Now, consider the case where the mapping $f$ is disjoint. Because the images $\{f(a) \st a \in S\}$ are mutually disjoint, each virtual fault can influence at most one real location in $S$. Achieving at least $t - \delta$ propagated faults in $S$ necessitates at least $l' = \lceil t - \delta \rceil$ virtual faults, which must be distributed across $l'$ distinct subsets $f(a)$ with $a \in S$. This probability is upper-bounded by:
\begin{equation}
    \Pr[|\{f(a) \cap F_B \neq \emptyset \st a \in S\}| \geq l'] \leq \binom{s}{l'} \nu^{l'} p^{l'} \leq \left( \frac{s\,e}{l'} \nu p \right)^{l'}.
\end{equation}
\end{proof}

\begin{lemma}
{\bf Merging of generalized noise models.} Consider two generalized noise models $(A_j, B_j, f_j, p_{j,\bullet}, \delta_j, \mathfrak{con}_j)$ for $j \in \{1, 2\}$, where the $j$th model is $\nu_j$-confined. Suppose two virtual location sets are disjoint, i.e.,~$B_1\cap B_2 = \emptyset$. Let $A$ be a set of locations, and assume there exist merging mappings $g_j: A \to \mathcal{P}(A_j)$ such that an arbitrary location $a \in A$ is faulty only if $g_1(a)\cup g_2(a)$ contains at least one fault. 

Then, the effective noise on $A$ is described by a $\nu$-confined model $(A, B, f, p_\bullet, \delta, \mathfrak{con})$, where: 
\begin{itemize}
\item $B = B_1 \cup B_2$ and $f(a) = f_{1,\cup} \circ g_1(a) \cup f_{2,\cup} \circ g_2(a)$; 
\item $p_\bullet|_{B_j} = p_{j,\bullet}$ for $j \in \{1, 2\}$, and $p = \max\{p_1, p_2\}$; 
\item $\delta = \delta_1 \|g_1\|_\gets + \delta_2 \|g_2\|_\gets$; 
\item $\mathfrak{con} = \mathfrak{con}_1 \cap \mathfrak{con}_2$; 
\item $\nu = \nu_1 \|g_1\|_\leftrightarrow + \nu_2 \|g_2\|_\leftrightarrow$. 
\end{itemize}
If $f_1$, $f_2$, $g_1$, and $g_2$ are disjoint, $f$ is disjoint. 
\label{lem:merging_generalized_models}
\end{lemma}

\begin{proof}
Let $F_{j,A}$, $F_{j,B}$, $F_{j,I}$, and $F_{j,P}$ denote the real, virtual, intrinsic, and propagated fault paths of model $j \in \{1, 2\}$, respectively. We use the same notation without the subscript $j$ to represent the fault paths of the resulting effective model on $A$. 

Based on the merging mappings, a location $a \in A$ is faulty only if at least one of its associated locations in $A_1$ or $A_2$ is faulty. Formally: 
\begin{eqnarray}
F_A \subseteq \{a \in A \st g_1(a) \cap F_{1,A} \neq \emptyset \text{ or } g_2(a) \cap F_{2,A} \neq \emptyset\}.
\end{eqnarray}
Since $F_{j,A} \subseteq F_{j,I} \cup F_{j,P}$, we can decompose the effective fault path into intrinsic and propagated components. We define $F_I$ as the set of real locations in $A$ affected by intrinsic faults in $A_1$ and $A_2$: 
\begin{eqnarray}
F_I \subseteq \{a \in A \st g_1(a) \cap F_{1,I} \neq \emptyset \text{ or } g_2(a) \cap F_{2,I} \neq \emptyset\}.
\end{eqnarray}
Conditioned on $\mathfrak{con} = \mathfrak{con}_1 \cap \mathfrak{con}_2$, the cardinality of the intrinsic fault set $F_I$ is bounded by the total number of errors propagated from the intrinsic sets $F_{1,I}$ and $F_{2,I}$, i.e.,~$|F_I| \leq \delta_1 \|g_1\|_{\gets} + \delta_2 \|g_2\|_{\gets} = \delta$. 

Next, we characterize the propagated fault path $F_P$. Using the definition of propagated faults in the constituent models, $F_{j,P} \subseteq \{a_j \in A_j \st f_j(a_j) \cap F_{j,B} \neq \emptyset\}$, we have: 
\begin{eqnarray}
F_P \subseteq \{a \in A \st f_{1,\cup} \circ g_1(a) \cap F_{1,B} \neq \emptyset \text{ or } f_{2,\cup} \circ g_2(a) \cap F_{2,B} \neq \emptyset\}.
\end{eqnarray}
This expression justifies the definitions of the virtual location set, mapping function, and error probability for the effective model. Specifically, the virtual fault path is given by $F_B = F_{1,B} \cup F_{2,B}$. Finally, the confinement parameter $\nu$ is obtained by calculating the norms $\|f\|_{\to}$ and $\|f\|_{\gets}$ based on the individual norms of $f_j$ and $g_j$, yielding the result stated in the lemma. 
\end{proof}

\begin{remark}
This result can be generalized to the merging of $N$ models with disjoint virtual location sets. For instance, if $\|g_j\|_\leftrightarrow = 1$ for all $j \in \{1, \dots, N\}$, the parameters of the resulting effective model are: 
\begin{align*}
p &= \max\{p_j \st j=1,2,\ldots,N\}, & \delta &= \sum_{j=1}^N \delta_j, & \mathfrak{con} &= \bigcap_{j=1}^N \mathfrak{con}_j, & \nu &= \sum_{j=1}^N \nu_j.
\end{align*}
\end{remark}

\begin{remark}
In all generalized noise models presented in Secs.~\ref{app:fault-tolerant}, \ref{app:noisy}, and~\ref{app:theorem_proof}, the propagation and merging mapping, i.e.,~$f$ and $g$, are all disjoint. Note that a sufficient condition for $g$ to be disjoint is that $\|g\|_\gets = 1$, which is satisfied by all merging mappings utilized in our analysis. For a $1$-confined model, $f$ is always disjoint. 
\end{remark}

\begin{lemma}
{\bf Splitting of a generalized noise model.} Consider a $\nu$-confined generalized noise model $(A, B, f, p_\bullet, \delta, \mathfrak{con})$. If the mapping $f$ is disjoint, then for any arbitrary partition of the real locations $A = A_1 \cup A_2$, the noise in each subset $A_j$ ($j \in \{1, 2\}$) is described by a $\nu$-confined model $(A_j, B_j, f_j, p_{j,\bullet}, \delta_j, \mathfrak{con}_j)$, characterized by:
\begin{itemize}
    \item $B_j = f_\cup(A_j) \subseteq B$ is the set of associated virtual locations;
    \item $f_j = f|_{A_j}$ is the restriction of the mapping $f$ to the domain $A_j$;
    \item $p_{j,\bullet} = p_{\bullet}|_{B_j}$ is the restriction of the probability mapping to $B_j$, with $p_j = p$;
    \item $\delta_j = \delta$, $\mathfrak{con}_j = \mathfrak{con}$, and $\nu_j = \nu$.
\end{itemize}
Furthermore, the resulting propagation mappings $f_1$ and $f_2$ are disjoint, and the virtual location sets $B_1$ and $B_2$ are disjoint. 
\label{lem:splitting_generalized_models}
\end{lemma}

\begin{proof}
Given an arbitrary fault path $F_A \subseteq A$ and its decomposition into intrinsic and propagated sub-paths $F_A \subseteq F_I \cup F_P$, this induces a corresponding decomposition of the fault path restricted to each $A_j$. Specifically, we define $F_j = F_A \cap A_j = (F_I \cap A_j) \cup (F_P \cap A_j)$, where $F_I \cap A_j$ and $F_P \cap A_j$ represent the intrinsic and propagated components of the $j$-th sub-model, respectively. Under the condition $\mathfrak{con}$, the cardinality of the intrinsic component satisfies $|F_I \cap A_j| \leq |F_I| \leq \delta$. Furthermore, the distribution of the propagated component $F_P \cap A_j$ satisfies the model of $A_j$ described in the lemma. 
\end{proof}

\subsubsection{Simplified local stochastic noise model}

In our protocol, universal quantum computation is implemented through the following primitive circuits: 
\begin{itemize}
\item \textbf{Resource state preparation:} The preparation of gate-teleportation resource states, as depicted in Fig.~\ref{fig:state_preparation}; 
\item \textbf{Parity-check measurement:} Measurement of stabilizer operators using the prepared resource states, as depicted in Fig.~\ref{fig:measurement}(a); 
\item \textbf{Code surgery:} Measurement of logical operators performed via code surgery, as depicted in Fig.~\ref{fig:lattice_surgery}; 
\item \textbf{Additional transversal operations:} Transversal initialization, Pauli and $S$ gates, and measurement of code blocks, as shown in Figs.~\ref{fig:operations} and~\ref{fig:injections}; the parity-check measurement and code surgery circuits apply specifically to $Z$-type operators, and $X$-type operators can be measured by applying transversal Hadamard gates to the code block both before and after a $Z$-type measurement. 
\end{itemize}

\begin{definition}
{\bf Simplified local stochastic noise model.} In this model, we assume that error events occur between transversal operations. To simplify the analysis, we effectively subsume the noise from all additional constant-depth transversal operations into the state preparation and parity-check measurement circuits. Errors in the code surgery circuits are likewise accounted for within these two primary circuits. Specifically, each location is either a qubit-time tuple $(q, t)$ or a measurement operation, defined as follows: 
\begin{itemize}
\item \textbf{State preparation (Fig.~\ref{fig:state_preparation}):} The locations include the qubit-time positions in the six indicated columns and all parity-check measurement operations shown in the circuit. We partition the locations into two disjoint subsets: the \textit{preparation locations}, consisting of qubit-time locations in the first five columns and all measurement locations, and the \textit{decoding locations}, consisting of qubit-time locations in the sixth column. 
\item \textbf{Parity-check measurement (Fig.~\ref{fig:measurement}(a)):} The locations include the qubit-time positions in block A of the first column and blocks A and B of the second column. Errors occurring at other qubit-time positions in the first and second columns are effectively subsumed into the preceding state preparation circuit, while errors in the third column is attributed to the subsequent parity-check measurement circuit. 
\end{itemize}

A fault at a qubit-time location indicates that an error is applied to the qubit at that time; a fault at a measurement location indicates that the measurement outcome is flipped. In this model, faults are distributed according to a standard local stochastic noise model characterized by a physical error parameter $p$ as follows: 
\begin{itemize}
\item \textbf{In the state preparation circuits:} 
    \begin{enumerate}
    \item For preparation locations, the error probability is $p_{L,S}(p)$. This represents the logical error probability of a surface-code logical qubit. We assume a threshold $p_{th,S}$ exists such that for $p < p_{th,S}$, this probability is $p_{L,S}(p) \leq e^{-\Theta(d_S)}$, where $d_S$ is the surface-code distance. 
    \item For decoding locations, the error probability is $p_{D,S} \leq \lambda_S\,p$. This represents the error probability of the decoding operation that transfers the state of a surface-code logical qubit to a physical qubit, where $\lambda_S$ is a positive constant. We define $\lambda = \max\{1,\lambda_S\}$. 
    \end{enumerate}
\item \textbf{In the stabilizer measurement circuits:} The error probability for every location is simply $p$. 
\end{itemize}
\end{definition}

This simplified noise model is consistent with the error weight analysis presented in Appendices~\ref{app:LTSP} and~\ref{app:PCS+LTS}. However, one specific simplification warrants clarification: we have neglected error propagation during the generalized transversal controlled-NOT gates (including those used for parity-check measurements). We will demonstrate that omitting explicit error propagation does not impact the validity of our final threshold conclusions (see Sec.~\ref{app:justification}). This is justified by our exclusive use of (q)LDPC codes, where the sparse connectivity of the operations ensures that any error spread remains strictly bounded. 

Regarding the surface-code assumptions, we acknowledge that this simplified model neglects potential temporal correlations of surface-code logical errors. However, as the state preparation circuit has constant depth, these correlations remain strictly bounded and do not affect the existence of a finite fault-tolerant threshold (see Sec.~\ref{app:justification}). 

\subsection{Fault-tolerant operations}
\label{app:fault-tolerant}

\begin{lemma}
{\bf Resource state preparation for fault-tolerant operations.} Let $d$ be the distance of the memory code. Assume the factory code is an LDPC-type LTC with parameters as specified in Table~\ref{tab:codes}. Assume that the circuit noise follows the simplified local stochastic noise model. There exists a decoding algorithm such that the state preparation circuit in Fig.~\ref{fig:state_preparation}, using a factory-code distance $d_F = \Theta(d)$ and surface-code distance $d_S = O(\mathrm{polylog}(n_D n_F))$, satisfies the following: 
\begin{itemize}
\item The circuit outputs $k_F$ copies of the resource state $H_Z^D$. 
\item Under the simplified local stochastic noise model with physical error rate $p < p_{S,th}$, the circuit fails with probability $\Pr(\mathfrak{fail}_{SP}) \leq \exp(-\Omega(d))$. 
\item If the circuit succeeds, each copy of the output state is subject to noise characterized as follows: 
    \begin{enumerate}
    \item Let $A_{RS}$ be the set of physical qubits for a given copy of the state. 
    \item The noise on the output state is effectively characterized by a $1$-confined model $(A_{RS}, B_{RS}, f_{RS}, p_{D,S}, t_{RS} = d/16, \mathfrak{success}_{SP})$ for some $B_{RS}$ and disjoint $f_{RS}$. 
    \end{enumerate}
\end{itemize}
The virtual location sets associated with the $k_F$ resource state copies are mutually disjoint; furthermore, the virtual location sets associated with resource states generated by different state preparation circuits are also disjoint. 
\label{lem:state_preparation}
\end{lemma}

\begin{proof}
Let $C$ denote the total set of locations, partitioned into the subsets of preparation locations, $C_P$, and decoding locations, $C_D$. Consequently, any fault path $F \subseteq C$ can be partitioned into $F = F_P \cup F_D$, where $F_P = F \cap C_P$ and $F_D = F \cap C_D$ represent preparation and decoding faults, respectively. 

The check matrix of the state preparation circuit, $H_Z^{sp}$, only detects $X$-type errors at preparation locations. Following Lemma~\ref{lem:spZ}, the $X$-type error vector is $e_X^{sp} \in \mathbb{F}_2^{C_X}$, where $C_X \subseteq C$ is the subset of locations relevant to $X$-type errors. The support of this vector, restricted to preparation locations, satisfies $\mathrm{supp}_{C_X \cap C_P}(e_X^{sp}) \subseteq C_X \cap F_P$. 

Given the observed syndrome $v^\mathrm{T} = H_Z^{sp} {e_X^{sp}}^\mathrm{T}$, the decoding algorithm identifies the minimum weight vector $c \in \mathbb{F}_2^{C_X}$ such that $v^\mathrm{T} = H_Z^{sp} c^\mathrm{T}$. Since $H_Z^{sp}$ only monitors locations in $C_X \cap C_P$, the weight-minimizing solution satisfies $\mathrm{supp}(c) \subseteq C_X \cap C_P$ and $\vert c \vert \leq \vert \mathrm{supp}_{C_X \cap C_P}(e_X^{sp}) \vert \leq \vert C_X \cap F_P \vert \leq \vert F_P \vert$. 

Error correction is performed by effectively applying $X$ gates on the support of $c$. Let $F' = F_P' \cup F_U$ be the post-correction fault path, where $F_P'$ denotes post-correction preparation faults. Since $F_P' \subseteq \mathrm{supp}(c) \cup F_P$, $\vert F_P' \vert \leq \vert c \vert + \vert F_P \vert \leq 2 \vert F_P \vert$. 

The circuit is defined to \textit{succeed} if the total number of preparation faults is smaller than $t$ ($\vert F_P \vert < t$), where 
\begin{eqnarray}
t = \left\lceil \min\left\{\frac{t_{RS}+1}{4 \max\left\{1, \frac{n_F}{r_F s}\right\}}, \frac{d_F}{2 \omega^D_Z \max\left\{1, \frac{n_F}{r_F s}\right\}} \right\} \right\rceil.
\label{equ:errorbound}
\end{eqnarray}
Under this success condition, $\vert F_P' \vert < 2t$. By Lemma~\ref{lem:spX} and Lemma~\ref{lem:spZ}, these preparation faults induce $X$ (or $Z$) errors on at most $t_{RS}/2$ physical qubits in each copy of the output state: the post-correction errors are undetectable (the corresponding error vector is in the kernel of $H_Z^{sp}$), and the potential occurrence of decoding faults does not interfere with the propagation of errors originating from the preparation locations. Let $F_{RS,P} \subseteq A_{RS}$ be the set of physical qubits having errors from preparation faults. It follows that $\vert F_{RS,P} \vert \leq t_{RS}$. 

The total number of preparation locations is $O(n_Dn_F)$. Using a tail bound for the binomial distribution, the probability that the circuit fails (i.e.,~$\vert F_P \vert \geq t$) is: 
\begin{eqnarray}
\Pr(\mathfrak{fail}_{SP}) \leq \binom{O(n_Dn_F)}{\lceil t \rceil} p_{L,S}(p)^{\lceil t \rceil} \leq \left( \frac{O(n_Dn_F) e}{\lceil t \rceil} p_{L,S}(p) \right)^{\lceil t \rceil}.
\label{eq:logicalerror}
\end{eqnarray}
By choosing a surface-code distance $d_S = O(\mathrm{polylog}(n_Dn_F))$ such that $p_{L,S}(p)$ is sufficiently small, we ensure the term inside the parentheses is less than 1. Given $t_{RS} = d/16$ and $d_F = \Theta(d)$, it follows that $t = \Theta(d)$, yielding a failure probability suppressed exponentially in the memory code distance: $\Pr(\mathfrak{fail}_{SP}) \leq e^{-\Omega(d)}$. 

Each fault at a decoding location maps directly to a physical qubit in the output state. This correspondence is described by a bijective mapping $f_{RS}: A_{RS} \to C_D'$, where $C_D' \subseteq C_D$ denotes the subset of decoding locations associated with the physical qubits of a given copy of the output state. Accordingly, we take $B_{RS} = C_D'$ as the set of virtual locations for the resulting output state. 

Let $F_{RS,D} \subseteq A_{RS}$ denote the set of output faults originating from the decoding locations. The total output fault path satisfies $F_{RS} \subseteq F_{RS,P} \cup F_{RS,D}$. Within this framework, the preparation faults $F_{RS,P}$ constitute the \textit{intrinsic faults}, satisfying $|F_{RS,P}| \leq t_{RS}$ conditioned on the event of circuit success. The decoding faults $F_{RS,D}$ represent the \textit{propagated faults} induced by the virtual fault path $F_D \cap C_D'$. Consequently, the resulting noise on the output state precisely satisfies the generalized stochastic noise model. 

Since the virtual locations of each output state copy are the decoding locations on the relevant qubit set, the virtual location sets associated with different copies are mutually disjoint. These virtual locations correspond to the physical components of each preparation circuit; consequently, the virtual location sets for different circuits are disjoint. 
\end{proof}

\begin{lemma}
{\bf Fault-tolerant parity-check measurement.} Consider the circuit in Fig.~\ref{fig:measurement}(a). Let $A_{in}$ and $A_{out}$ denote the sets of qubits in the input and output states, respectively. Assume that the circuit noise follows the simplified local stochastic noise model. Suppose the resource state is generated by a fault-tolerant resource state preparation circuit depicted by Lemma~\ref{lem:state_preparation}.

Under these conditions, all errors introduced by the noise sources are equivalent to effective errors on the input and output states, characterized as follows:
\begin{itemize}
\item The circuit is treated as effectively error-free.
\item The effective noise on the input state is described by a $4$-confined model $(A_{in}, B_{in}, f_{in}, \lambda p, t_{RS}, \mathfrak{success}_{SP})$ for some $B_{in}$ and disjoint $f_{in}$.
\item The effective noise on the output state is described by a $2$-confined model $(A_{out}, B_{out}, f_{out}, \lambda p, t_{RS}, \mathfrak{success}_{SP})$ for some $B_{out}$ and disjoint $f_{out}$.
\end{itemize}
The virtual location sets $B_{in}$ and $B_{out}$ are disjoint; furthermore, the virtual location sets associated with different parity-check measurements are also disjoint.
\label{lem:stabilizer_measurement}
\end{lemma}

\begin{proof}
Let $C_{\alpha,j}$ denote the set of qubit-time locations for block $\alpha$ in the $j$-th column of the circuit, and define $C_{in} = \bigcup_{\phi \in \Phi} C_\phi$, where the index set is $\Phi = \{(A,1), (B,1),(B,2)\}$.
By Lemmas~\ref{lem:MeasurementZ} and~\ref{lem:MeasurementX}, every single-qubit fault at a location in $C_{\mathrm{in}}$ is equivalent to a single-qubit error on the input state, and every single-qubit fault at a location in $C_{A,2}$ is equivalent to a single-qubit error on the output state. Hence every single-qubit fault within the circuit is equivalent to a single-qubit error on either the input or output state.

Partition the resource-state as $A_{RS}=A_{RS,1}\cup A_{RS,2}$, where errors on $A_{RS,1}$ are equivalent to errors at locations $C_{B,1}$ and hence to input errors, while errors on $A_{RS,2}$ are equivalent to errors at locations $C_{C,1}$ and hence to output errors. By Lemma~\ref{lem:splitting_generalized_models}, this induces corresponding disjoint virtual location sets $B_{RS,1}$ and $B_{RS,2}$.

According to the above analysis, the error equivalences are described by the following merging mappings: $g_{\phi\in \Phi}: A_{in} \rightarrow C_{\phi}$, $g_{A,2}: A_{out} \rightarrow C_{A,2}$,  $g_{RS,1}: A_{in} \to A_{RS,1}$ and $g_{RS,2}: A_{out} \to A_{RS,2}$. All these mappings satisfy \(\|g\|_\leftrightarrow=1\). The input and output virtual location sets are $B_{in}=C_{\mathrm{in}}\cup B_{RS,1}$ and $B_{out}=C_{A,2}\cup B_{RS,2}$, which are disjoint. By applying Lemma~\ref{lem:merging_generalized_models} to these parameters, the proof is complete. Note that all the mappings are disjoint. 

The sets $B_{in}$ and $B_{out}$ are comprised of locations within the current circuit or virtual locations of the resource state. Since the virtual locations for distinct resource states are mutually disjoint, they are disjoint from the virtual locations for all other parity-check measurements.
\end{proof}

\begin{remark}
While we only explicitly analyze the deformed code parity-check matrix $H_Z^D$ in Lemmas~\ref{lem:state_preparation} and~\ref{lem:stabilizer_measurement}, our results and conclusions generalize naturally to all other check matrices utilized in the protocol. 
\end{remark}

\begin{lemma}
{\bf Lemma 1 in \cite{Gottesman2014}.} Consider a set $S$ of $t$ nodes in a graph where each node has degree at most $z$. Let $M_z(\mathsf{s}, \mathsf{S})$ be the number of sets containing $\mathsf{S}$ and in total $\mathsf{s}$ nodes ($\mathsf{s}-t$ nodes beyond those in $\mathsf{S}$), which form a union of connected clusters, each containing at least one node from $S$. Then $M_z(\mathsf{s}, \mathsf{S}) \leq e^{t-1}(ze)^{\mathsf{s}-t}$, where $e$ is the base of the natural logarithm. 
\label{lem:cluster_bound}
\end{lemma}

\begin{lemma}
{\bf Fault-tolerant code surgery.} Let $d$ be the distance of the memory code. Consider the code surgery circuit illustrated in Fig.~\ref{fig:lattice_surgery}. We assume that the deformed code is a qLDPC code constructed such that $n_D = O(\mathrm{poly}(d))$ and $d_D = d$. Let $A_{in}$ and $A_{out}$ denote the sets of qubits in the input and output states, respectively. Assume that the circuit noise follows the simplified local stochastic noise model. Furthermore, assume the noise on the input state is characterized by a $\nu$-confined model $(A_{in},B_{in},f_{in},\lambda p,\eta t_{RS},\mathfrak{success}_{SP})$ for some $B_{in}$ and disjoint $f_{in}$, where $\nu = O(1)$, $B_{in}$ is disjoint from the virtual location sets of all other noise sources in the circuit, $\eta < 8$, and $\mathfrak{success}_{SP}$ denotes the event that all preceding fault-tolerant resource state preparation circuits have been successful. Suppose the parity-check measurements are fault-tolerant as depicted by Lemma~\ref{lem:stabilizer_measurement}, and that the condition $\mathfrak{success}_{SP}$ is satisfied. 

Under these conditions, there exists a decoding algorithm such that:
\begin{itemize}
    \item There exists a finite threshold $\epsilon > 0$ such that when $p < \epsilon$, the error correction fails with probability $\Pr(\mathfrak{fail}_{EC}) \leq \exp(-\Omega(d))$. If the error correction succeeds, the errors undergoing correction, combined with the application of recovery gates, result in trivial operations on the encoded state, ensuring that the logical measurement outcomes are correct. 
    \item Errors that bypass the correction process result in effective noise on the output state characterized by an $8$-confined model $(A_{out},B_{out},f_{out},\lambda p,4t_{RS},\mathfrak{success}_{SP})$ for some $B_{out}$ and disjoint $f_{out}$. 
\end{itemize}
Furthermore, $B_{out}$ is disjoint from the virtual location sets of all other code surgery circuits and parity-check measurements, with the exception of the two specifically employed in the current circuit. 
\label{lem:code_surgery}
\end{lemma}

\begin{proof}
Let $C_j$ denote the set of qubit-time locations in the $j$th column of the circuit. Following the proofs of Lemmas~\ref{lem:csmZ} and~\ref{lem:csmX}, $X$-type errors at $C_1$ undergo correction through the $Z$-type parity-check measurement, while $X$ errors at $C_2, C_3$, and $C_4$ bypass the correction procedure. Similarly, $Z$-type errors at $C_1, C_2$, and $C_3$ undergo correction through the $X$-type parity-check measurement, whereas $Z$ errors at $C_4$ bypass the correction procedure. 

We now analyze the effective noise model for each subset $C_j$. Errors on the input state and the effective input errors of the $Z$-type parity-check measurement are equivalent to errors at $C_1$. Taking these two sources of errors into account, the effective noise at $C_1$ is characterized by a $(\nu+4)$-confined model with $(C_1,B_1,f_1,\lambda p,\eta t_{RS},\mathfrak{success}_{SP})$ for some $B_1$ and disjoint $f_1$. The errors at $C_2$ and $C_4$ are the effective output errors of corresponding parity-check measurements, therefore, characterized by a $2$-confined model with $(C_j,B_j,f_j,\lambda p,t_{RS},\mathfrak{success}_{SP})$ for some $B_j$ and disjoint $f_j$, where $j = 2,4$; similarly,the effective errors at $C_3$ are characterized by a $4$-confined model with $(C_3,B_3,f_3,\lambda p,t_{RS},\mathfrak{success}_{SP})$ for some $B_3$ and disjoint $f_3$. Note that virtual location sets in these models are mutually disjoint. 

{\bf Error correction.} Our derivation follows the threshold proof presented in Ref.~\cite{Gottesman2014}, adapted for our $\nu$-confined generalized noise model. 

Since the parity-check measurements are effectively error-free , the logical measurement outcomes are always correct provided that no logical $X$ errors occur prior to the $Z$-type parity-check measurement. Consequently, the threshold proof reduces to analyzing the probability of logical errors. 

First, we consider $X$ errors acting on the memory system in the first column, which are corrected by the check matrix $\tilde{H}_Z$ (see Lemma~\ref{lem:csmX}): Although the circuit measures stabilizers for the complete $H_Z^D$ check matrix, only the $\tilde{H}_Z$ component can detect errors. This is because only the stabilizers of $\tilde{H}_Z$ have well-defined values in the input state, enabling a direct syndrome comparison to identify faults. Note that measurement outcomes for these stablizers are treated as effectively error-free. 

The proof is established on the \textit{adjacency graph} of $\tilde{H}_Z$. Let $C_1' \subseteq C_1$ be the subset of qubits in the memory system. Each column of $\tilde{H}_Z$ corresponds to a physical qubit at a location in $C_1'$. In this graph, qubits are represented as vertices, and an edge exists between two vertices if and only if they share a common parity check. For qLDPC codes, the maximum degree of this graph is bounded by a constant $z = O(1)$. Under the minimum-weight decoding algorithm, error correction fails only if there exists a connected cluster $S \subseteq C_1'$ of $s \geq d$ qubits in the adjacency graph such that $S$ contains at least $\lceil s/2 \rceil$ errors. 

Given $S \subseteq C_1$, we apply Lemma~\ref{lem:probability_bound}. The probability that $S$ contains at least $\lceil s/2 \rceil$ errors is bounded by:
\begin{equation}
P_s = \left[ \frac{(\nu+4)\,s\,e}{l} \lambda p \right]^l, \text{ where } l = \left\lceil \lceil s/2 \rceil - \eta t_{RS} \right\rceil.
\end{equation}
Utilizing $t_{RS} = d/16$, we observe that for $s \geq d$, $l \geq \kappa s$, where $\kappa = 1/2 - \eta/16$. For a sufficiently small physical error rate $p < \epsilon$, where 
\begin{equation}
\epsilon = \frac{\kappa}{\lambda(\nu+4)e(ze)^{1/\kappa}},
\end{equation}
the term inside the parentheses is bounded by $(ze)^{1/\kappa} > 1$. Then, we obtain:
\begin{equation}
P_s \leq \left[ \frac{(\nu+4)\,e}{\kappa} \lambda p \right]^{\kappa s}.
\end{equation}

The number of connected clusters of size $s$ containing a fixed vertex in the adjacency graph is at most $(ze)^{s-1}$ (Lemma~\ref{lem:cluster_bound}). Summing over all $n_D$ possible starting vertices and all cluster sizes $s \geq d$, the total probability of failure for $X$-type error correction is bounded by: 
\begin{align}
\Pr(\mathfrak{fail}_{EC,X}) &\leq \sum_{s \geq d} n_D (ze)^{s-1} P_s \leq \frac{n_D}{ze} \sum_{s \geq d} \left( \frac{p}{\epsilon} \right)^{\kappa s} = \exp(-\Omega(d)).
\end{align}
This result establishes the existence of a finite fault-tolerant threshold in the correction of $X$ errors. 

Next, we consider $Z$ errors acting on the first three columns, which are corrected by the check matrix $H_X^D$ (see Lemma~\ref{lem:csmZ}). Although the circuit directly measures only the stabilizers in $\tilde{H}_X$, the eigenvalues of the full $H_X^D$ stabilizer set can be reconstructed from the outcomes of the transversal $X$-basis measurements on the ancilla system. 

The $Z$-type errors occurring in the first three columns are equivalent to effective $Z$ errors localized at $C_1$. By merging the noise models for these three columns, we obtain a $(\nu+10)$-confined generalized noise model $(C_1,B_1',f_1',\lambda p,(\eta+3) t_{RS},p\},\mathfrak{success}_{SP})$ for some $B_1'$ and $f_1'$. Note that the distance of the code used for decoding is $d_D = d$. Following an analysis analogous to that for $X$-type errors, we can establish the existence of a finite threshold $\epsilon_Z > 0$. Specifically, for any physical error rate $p < \epsilon_Z$, the total failure probability for $Z$-type error correction is exponentially suppressed: $\Pr(\mathfrak{fail}_{EC,Z}) = \exp(-\Omega(d))$. 

{\bf Faults on the output state.} $X$-type errors at $C_2$ and $C_3$ may cause faults on the output state: each single-qubit $X$ error on the memory system can result in fault on a physical qubit in the output state; and $X$ errors on the ancilla system are trivial because of the transversal $X$ measurement. Errors at $C_4$ directly act on physical qubits in the output state. Consequently, the corresponding merging mappings $g$ each satisfy $\|g\|_\leftrightarrow = 1$, which implies they are disjoint. By applying Lemma~\ref{lem:merging_generalized_models} to these parameters, we have the effective noise model of the output state. 

The set $B_{out}$ is comprised of locations within the current circuit and the virtual locations of its two constituent parity-check measurements. Since the virtual sets for distinct parity-check measurements are mutually disjoint, they are disjoint from the virtual locations of all other parity-check measurement and code surgery circuits. 
\end{proof}

\begin{remark}
In Appendix~\ref{app:FToperations}, we establish that the input noise model parameters are $\nu = 8$ and $\eta = 4$, resulting in a threshold $\epsilon = \frac{\kappa}{12 \lambda e (ze)^{1/\kappa}}$, where $\kappa = 1/2 - 4/16 = 1/4$. Numerical studies in Ref.~\cite{Li2014AMS} demonstrate that magic states can be injected into surface codes with effective error rates potentially lower than those of individual physical gates, suggesting a parameter value of $\lambda \approx 1$. The factor of $1/16$ in the definition of $\kappa$ originates from the specific definition of $t_{RS}$. By logarithmically increasing the surface-code distance $d_S$, this factor can be reduced such that $\kappa$ approaches $1/2$. In the limit $\kappa = 1/2$, this threshold coincides with the bound reported in Theorem~3 of Ref.~\cite{Gottesman2014}, subject to an additional factor of $(6e)^{-1}$. This discrepancy is comprised of a factor of $2/e$ arising from the specific combinatorial bounding techniques employed, and a factor of $1/12$ which accounts for the accumulation of errors throughout the circuit.
\end{remark}

The final fault-tolerant operations to be analyzed are the initializations of a code block in the logical $\ket{0}$ or $\ket{+}$ state. We focus our discussion on the $\ket{+}$ state, as the results are symmetrically applicable to the $\ket{0}$ state. 

\begin{lemma}
{\bf Fault-tolerant block initialization.} Apply the parity-check measurement of $H_X^M$ to a memory-code block where each physical qubit is initialized in the $\ket{+}$ state, the circuit yields a block with all logical qubits prepared in the logical state $\ket{+}$. Let $A_{out}$ denote the set of physical qubits in the resulting output state. Assume that the circuit noise follows the simplified local stochastic noise model. Suppose the parity-check measurement is fault-tolerant as depicted by Lemma~\ref{lem:stabilizer_measurement}. 

Under these conditions, the output state is subject to effective noise described by a $6$-confined model $(A_{out},B_{out},f_{out},\lambda p,t_{RS},\mathfrak{success}_{SP})$ for some $B_{out}$ and disjoint $f_{out}$. Furthermore, $B_{out}$ is disjoint from the virtual location sets of all other initialization circuits, code surgery circuits and parity-check measurements, with the exception of the one specifically employed in the current circuit. 
\label{lem:initialization}
\end{lemma}

\begin{proof}
By considering the logical $X$ operators, one can verify that the logical qubits are initialized in the logical $\ket{+}$ state. Specifically, if the entire circuit is error-free, the output state is an eigenstate of both the logical $X$ operators and the $X$-type stabilizers with eigenvalue $+1$. Furthermore, the state is projected into an eigenstate of the $Z$-type stabilizers, with eigenvalues determined by the outcomes of the parity-check measurement. 

The noise in the parity-check measurement is effectively described by the input and output noise models established in Lemma~\ref{lem:stabilizer_measurement}. Input $X$ errors are trivial in this context, as they act on physical qubits already in the $\ket{+}$ state. Conversely, input $Z$ errors are equivalent to $Z$ errors on the final output state, as they commute with the parity-check measurement. Therefore, the effective noise model for the fault-tolerant initialization operation is obtained by merging the two effective noise models of the parity-check measurement according to Lemma~\ref{lem:merging_generalized_models}. Note that merging mappings $g$ are disjoint. 

The set $B_{out}$ consists of the virtual locations from its constituent parity-check measurement; consequently, it is disjoint with virtual location sets for all other fault-tolerant operations. 
\end{proof}

\subsection{Noisy preparation of encoded magic states}
\label{app:noisy}

\begin{remark}
Throughout the following discussion, referring to certain operations as ``noisy'' does not necessarily imply that their physical fidelity is low. Rather, it signifies that we have adopted a modified success condition: while the probability of failure may be higher, the weight of the residual errors is guaranteed to be low, conditioned on success. 

To distinguish between the events of fault-tolerant operations and noisy operations, we let $\mathfrak{success}'$ and $\mathfrak{fail}'$ denote events regarding noisy operations, whereas $\mathfrak{success}$ and $\mathfrak{fail}$ denote events relevant to fault-tolerant operations. 

In our protocol, surface codes are employed in two distinct roles. In the resource state preparation circuits, they protect qubits against both $X$- and $Z$-type errors. In the context of magic state preparation, they are utilized for the initial encoding of the magic states. For clarity, we denote the distances of these two types of surface codes as $d_S$ and $d_S'$, respectively. 

In this section, we adopt the notation and conditions established in Sec.~\ref{app:fault-tolerant}; thus, they will not be redundantly restated. For simplicity, we omit explicit statements regarding the disjointness of propagation mappings and virtual location sets, although these properties remain valid for all noisy operations. 
\end{remark}

\begin{lemma}
{\bf Resource state preparation for noisy operations.} Under the same assumptions as Lemma~\ref{lem:state_preparation}, but with a modified success condition, the following statements hold: 
\begin{itemize}
\item If $p < p_{S,th}$, the circuit fails with a probability $\Pr(\mathfrak{fail}_{SP}') \leq p/16$. 
\item Conditioned on the circuit succeeding, each copy of the output state is subject to noise characterized by a $1$-confined model $(A_{RS},B_{RS},f_{RS},p_{D,S},0,\mathfrak{success}_{SP}')$. 
\end{itemize}
\label{lem:noisy_state_preparation}
\end{lemma}

\begin{proof}
The proof follows the same logic as that of Lemma~\ref{lem:state_preparation}, by setting $t_{RS} = 0$. Given that $\lceil t \rceil \geq 1$, the failure probability bound $\Pr(\mathfrak{fail}_{SP}') \leq p/16$ can be satisfied by selecting a surface-code distance $d_S = O(\mathrm{polylog}(n_Dn_F))$, provided the physical error rate is below the surface-code threshold. 
\end{proof}

\begin{lemma}
{\bf Noisy parity-check measurement.} Under the same assumptions as Lemma~\ref{lem:stabilizer_measurement}, but suppose the resource state is generated by a noisy resource state preparation circuit depicted by Lemma~\ref{lem:noisy_state_preparation}. The circuit is effectively error-free, with effective input and output noise described by a $4$-confined model $(A_{in}, B_{in}, f_{in}, \lambda p, 0, \mathfrak{success}_{SP}')$ and a $2$-confined model $(A_{out},B_{out},f_{out},\lambda p,0,\mathfrak{success}_{SP}')$, respectively. 
\label{lem:noisy_stabilizer_measurement}
\end{lemma}

\begin{proof}
The proof is the same as that of Lemma~\ref{lem:stabilizer_measurement}. 
\end{proof}

\begin{lemma}
{\bf Noisy code surgery.} Consider the code surgery circuit illustrated in Fig.~\ref{fig:lattice_surgery}. We assume a surface-code distance $d_S' = O(\mathrm{polylog}(d))$ and a deformed code distance $d_D = d_S'$. Assume the noise on the input state is characterized by a $\nu$-confined model $(A_{in},B_{in},f_{in},\lambda p,0,\mathfrak{success}_{SP}')$ with $\nu = O(1)$. Suppose the parity-check measurements are noisy as depicted by Lemma~\ref{lem:noisy_stabilizer_measurement}. 

Under these conditions, there exists a decoding algorithm such that:
\begin{itemize}
    \item There exist a finite threshold $\epsilon' > 0$ and a constant $\lambda_S' > 0$ such that for $p < \epsilon'$, the error correction fails with probability $\Pr(\mathfrak{fail}_{EC}') \leq \lambda_S' p$. 
    \item Errors that bypass the correction process result in effective noise on the output state characterized by an $8$-confined model $(A_{out},B_{out},f_{out},\lambda p,0,\mathfrak{success}_{SP}')$. 
\end{itemize}
\label{lem:noisycode_surgery}
\end{lemma}

\begin{proof}
The difference between fault-tolerant and noisy parity-check measurements lies in the effective output noise. Specifically, in the fault-tolerant (noisy) measurement, the intrinsic fault parameter is $\delta = t_{RS}$ ($\delta = 0$). Consequently, the proof for Lemma~\ref{lem:code_surgery} applies here, yielding a finite threshold $\epsilon'$, and if $p < \epsilon'$, the failure probability is $\Pr(\mathfrak{fail}_{EC}') \leq n_D \exp(-\Omega(d_S'))$. By choosing the surface-code distance $d_S' = O(\mathrm{polylog}(d))$, we can ensure that this probability is suppressed below $\lambda_S' p$ for any fixed physical error rate $p < \epsilon'$ and positive constant $\lambda_S'$. 
\end{proof}

\begin{lemma}
{\bf Noisy block initialization.} Under the same conditions as in Lemma~\ref{lem:initialization} with the exception that the parity-check measurement is noisy as depicted by Lemma~\ref{lem:noisy_stabilizer_measurement}. Then, the output state is subject to effective noise described by a $6$-confined model $(A_{out},B_{out},f_{out},\lambda p,0,\mathfrak{success}_{SP}')$. 
\label{lem:noisyinitialization}
\end{lemma}

\begin{proof}
The proof is the same as that of Lemma~\ref{lem:initialization}. 
\end{proof}

\begin{lemma}
{\bf Noisy surface-code state injection.} Consider a surface-code block of distance $d_S'$. Suppose each physical qubit is initialized according to the circuit in Fig.~\ref{fig:injections}(b), following the layout pattern given in Fig.~\ref{fig:surface_code}(e). Upon applying parity-check measurements to extract $X$- and $Z$-type stabilizers, the circuit prepares the surface-code block in the logical state $\ket{\psi}$. Assume that the circuit noise follows the simplified local stochastic noise model. Suppose the parity-check measurements are noisy as depicted by Lemma~\ref{lem:noisy_stabilizer_measurement}. 

Under these conditions, there exists a decoding algorithm such that:
\begin{itemize}
    \item There exist a finite threshold $\epsilon' > 0$ and a constant $\lambda_S' > 0$ such that for $p < \epsilon'$, the error correction fails with probability $\Pr(\mathfrak{fail}_{EC}') \leq \lambda_S' p$. 
    \item Errors that bypass the correction process result in effective noise on the output state characterized by an $8$-confined model $(A_{out},B_{out},f_{out},\lambda p,0,\mathfrak{success}_{SP}')$. 
\end{itemize}
\label{lem:injection}
\end{lemma}

\begin{proof}
Every parity-check measurement is subject to both input and output effective noise. The error correction procedure specifically addresses $X$ errors preceding the $Z$-stabilizer measurement and $Z$ errors preceding the $X$-stabilizer measurement. Consequently, the effective noise on the output state is described by a model formed by merging one input noise model and two output noise models. Notably, in the merging of models, the merging mappings $g$ are disjoint. 

The noise addressed during error correction is characterized by a model formed by merging two input noise models and one output noise model, which is a $10$-confined model $(A_{in},B_{in},f_{in},\lambda p,0,\mathfrak{success}_{SP}')$. Error correction is performed by comparing the stabilizer measurement outcomes with expected eigenvalues derived from the initialization pattern shown in Fig.~\ref{fig:surface_code}(e). The initialization pattern is partitioned into two triangular regions, where qubits are initialized in either the $\ket{0}$ or $\ket{+}$ state. In the $\ket{0}$-state region, $Z$-type stabilizer eigenvalues can be deduced and utilized for error correction, while in the $\ket{+}$-state region, $X$-type stabilizer eigenvalues are similarly available. Consequently, this initialization scheme provides only partial error-correction capability compared to the standard surface code. 

Henceforth, we focus on the $\ket{+}$-state region, defined by the top horizontal boundary and the diagonal line of the surface-code lattice [see Fig.~\ref{fig:surface_code}(e)]. We observe that a logical $Z$ error occurs if a connected cluster of $Z$ errors spans the distance from the top boundary to the diagonal line. Let $x$ denote the horizontal coordinate of a qubit on the top boundary, indexed from left to right. There exists a positive constant $L$ such that any such cluster containing qubit-$x$ must have a minimum size of $x/L$; otherwise, the cluster cannot span the geometric distance between the top boundary and the diagonal line. 

The number of connected clusters of size $s$ containing a fixed qubit on the top boundary is at most $(ze)^{s-1}$ (see Lemma~\ref{lem:cluster_bound}). Since a cluster of size $s$ can only reach the diagonal line if its starting point $x$ on the top boundary satisfies $x \leq Ls$, there are at most $Ls$ possible starting positions for such a cluster. Therefore, the total number of connected size-$s$ clusters spanning from the top boundary to the diagonal line is upper-bounded by $Ls(ze)^{s-1}$. 

Following an analysis analogous to the proof of Lemma~\ref{lem:code_surgery}, the probability of having sufficient errors within a cluster of size $s$ to cause a error correction failure is upper bounded: there exist positive constants $\epsilon'$ and $\kappa'$ such that for $p < \epsilon'$, 
\begin{equation}
P_s \leq \left( \frac{10\,e}{\kappa'} \lambda p \right)^{\max\{\kappa' s,1\}}.
\end{equation}
Note that we have used $l \geq 1$. Consequently, the probability of an error correction failure can be bounded by splitting the sum over cluster sizes at a constant $s_0 = \lfloor 1/\kappa' \rfloor$: 
\begin{align}
\Pr(\mathfrak{fail}_{EC}') &\leq \frac{L}{ze} \sum_{s = 1}^{s_0} s \left( \frac{p}{\epsilon'} \right)^1 + \frac{L}{ze} \sum_{s > s_0} s \left( \frac{p}{\epsilon'} \right)^{\kappa' s} 
\leq \frac{L}{ze} \left[ \frac{s_0(s_0+1)}{2\epsilon'} p + \frac{(s_0+1)(p/\epsilon')^{\kappa'(s_0+1)}}{[1 - (p/\epsilon')^{\kappa'}]^2} \right].
\end{align}
When $(p/\epsilon')^{\kappa'} \leq 1/2$, which defines the threshold, there exists a positive constant $\lambda_S'$ such that $\Pr(\mathfrak{fail}_{EC}') \leq \lambda_S' p$. 

It is similar for the correction of $X$ errors. 
\end{proof}

\subsection{Proof of the theorem under the simplified noise model}
\label{app:theorem_proof}

\subsubsection{Fault-tolerant operations}
\label{app:FToperations}

The lemmas in Sec.~\ref{app:fault-tolerant} establish the fault tolerance of the protocol's operations, specifically including all operations illustrated in Figs.~\ref{fig:operations} and~\ref{fig:injections}(a) (taking $d_C = \Theta(d)$). These operations, which consist of parity-check measurements and code surgery, each fail with a probability upper-bounded by $e^{-\Omega(d)}$, conditioned on the success of all preceding fault-tolerant operations. Consequently, the probability that any failure occurs within a set of $L_{FT}$ fault-tolerant operations is at most $L_{FT} e^{-\Omega(d)}$. This total failure probability can be efficiently suppressed by increasing the code distance $d$.

The final remaining condition to examine is the input noise model for code surgery operations; see Lemma~\ref{lem:code_surgery}. In a general code surgery, the input memory system may be composed of multiple code blocks, each originating from a preceding operation---such as a parity-check measurement, code surgery, or block initialization. A single preceding operation may produce multiple output code blocks, with noise collectively described by a unified noise model. Since the propagation mapping $f$ in all our output noise models is inherently disjoint, we can apply Lemma~\ref{lem:splitting_generalized_models} to decompose the collective model into independent sub-models for each block. Consequently, each input memory-code block participating in a code surgery operation possesses noise characterized by an $8$-confined model $(A,B,f,\lambda p,4t_{RS},\mathfrak{success}_{SP})$. 

However, a critical issue arises: a single code surgery operation may involve as many as $\Theta(k_R)$ code blocks. Since each constituent input block contributes an intrinsic fault path of weight $4t_{RS}$, the cumulative weight of the intrinsic faults across the entire input memory system scales as $\Theta(k_R) t_{RS}$. Consequently, the joint noise model for the memory system is described by an $8$-confined model $(A', B', f', \lambda p, \Theta(k_R)t_{RS}, \mathfrak{success}_{SP})$, which is an invalid input when $k_R$ is large. 

We can effectively control the accumulation of intrinsic faults through joint error correction. Whenever multiple code blocks are provided as input to a code surgery operation, a preceding round of joint error correction must be performed. This is implemented by executing a code surgery operation with the check matrix configured as $H_Z^D = \tilde{H}_Z$. While the parity-check measurements are executed jointly across multiple blocks, the subsequent error correction is performed individually for each code block. Although the total intrinsic error weight may reach $\Theta(k_R)t_{RS}$ across the full input memory system, the intrinsic weight localized to any single block remains bounded by $4t_{RS}$. Consequently, the joint error correction round maintains a failure probability upper-bounded by $e^{-\Omega(d)}$. Following this procedure, the noise on the resulting memory system is characterized by an $8$-confined model $(A'',B'',f'',\lambda p,4t_{RS},\mathfrak{success}_{SP})$. This provides a valid input for the subsequent code surgery operation. 

\subsubsection{Magic states}

The noisy encoded magic states are prepared using the circuit illustrated in Fig.~\ref{fig:injections}(b). Most constituent operations within this circuit have been previously analyzed in Sec.~\ref{app:noisy}. Specifically, each operation yields an effective output noise model that serves as a valid input for the subsequent step, with logical failure probabilities bounded by $O(p)$.

The remaining operation, the transversal $X$-basis measurement on the surface-code block, corrects $Z$-type errors to ensure that the feedback gate applied to the memory-code block is logically correct. Since the surface code possesses a finite fault-tolerant threshold, the probability of a logical failure at this stage is upper-bounded by $e^{-\Omega(d_S')}$. By choosing a surface-code distance $d_S' = O(\mathrm{polylog}(d))$, this probability can be efficiently suppressed to $O(p)$.

In summary, the circuit produces an encoded magic state with a logical error probability of $O(p)$ and residual physical errors characterized by an $8$-confined model $(A_{out},B_{out},f_{out},\lambda p,0,\mathfrak{success}_{SP}')$. Hence, the prepared state is a valid noisy input for the subsequent fault-tolerant stages of the protocol.

\begin{theorem}[Theorem 1.4 in \cite{Nguyen2024}]
Let $\ket{\mathrm{CCZ}} = \mathrm{CCZ} \ket{+}^{\otimes 3}$. There exists a distillation protocol (using noiseless Clifford operations and measurements) $\mathrm{MSD}(p_L)$ and a constant noise threshold $p_{\mathrm{noisy}}$ such that the following holds. Upon receiving $N$ (independent) noisy states whose infidelity with $\ket{\mathrm{CCZ}}$ is $p \le p_{\mathrm{noisy}}$, where $N=O(\log^{1+o(1)}(1/p_L))$ is a sufficiently large number, $\mathrm{MSD}(p_L)$ produces $\Theta\!\left( \frac{N}{\log^{o(1)}(1/p_L)} \right)$ states each of which has infidelity at most $p_L$ with $\ket{\mathrm{CCZ}}$. Furthermore, the quantum depth and classical depth of $\mathrm{MSD}(p_L)$ are both polyloglog$(1/p_L)$. Here, the $o(1)$ exponent scales as $O\!\left( \frac{1}{\log \log\log(1/p_L)} \right)$.
\label{the:constantmagic}
\end{theorem}

Although our protocol is formulated in terms of $T$-gate magic states, it can be straightforwardly adapted to CCZ-gate magic states. Alternatively, one could utilize the $CCZ$ state distillation protocol (described in Theorem~\ref{the:constantmagic}) to distill $T$ states, as these states can be converted into each other with constant overhead~\cite{Gidney2019efficientmagicstate}. Given that the error rate of the injected magic state is $O(p)$, there exists a finite threshold for $p$ below which the conditions of Theorem~\ref{the:constantmagic} are satisfied. Consequently, magic states can be distilled to the target error rate with almost constant overhead.

\subsubsection{Summary}

The overhead analysis is provided in Secs.~\ref{sub:cost analysis} and Appendices~\ref{app:cost}. Our analysis concludes that the protocol achieves strictly constant qubit overhead and a time overhead of $O(d^{a+o(1)})$, provided that the code distances $d,d_R,d_F=O(\mathrm{polylog}(Mk))$.

We consider the usual setting where classical information is provided as input, processed by the quantum computation, and finally returned as classical output. We use $C_{\mathrm{FT}}$ to simulate the quantum circuit $C$, and measure the simulation error by the discrepancy between their resulting output distributions, with target accuracy $\epsilon_L$.

To ensure that the output distributions of $C$ and $C_{\mathrm{FT}}$ differ by at most $\epsilon_L$, it suffices to require that the logical failure probability per operation satisfy $p_L\le \epsilon_L/|C|$. In our construction, $p_L=e^{-\Omega(d)}$, so we choose $d=d_R=d_F=\Theta(\log(|C|/\epsilon_L))$. Given that the circuit size is polynomial in the width, $|C|=O(\mathrm{poly}(W))$, and the total number of logical qubits satisfies $Mk=\Theta(W)$, the polylogarithmic condition assumed in the overhead analysis is fully justified. Consequently, the total time overhead of the protocol scales as $O(\log^{a+o(1)}(|C|/\epsilon_L))$.

\subsection{Justification of the simplified noise model}
\label{app:justification}

In this section, we demonstrate that the threshold theorem remains valid under the standard local stochastic noise model, assuming each operation in the circuit can be faulty. In this framework, every circuit operation is treated as a potential fault location, and the probability of a specific fault path $F$ is exponentially suppressed by the number of faulty operations, i.e.,~$\Pr[F] \leq p^{|F|}$. 

\textbf{Error propagation.} In our simplified model, we assume that noise is concentrated at specific locations that subsume all errors generated by the preceding or subsequent constant-depth transversal operations. These locations are treated under a $1$-confined model with $\delta = 0$, where the probability of $s$ faults is exponentially suppressed as $p^s$. 

However, in a more realistic model where every operation acts as a potential error source, a single faulty operation may propagate to multiple locations (of the simplified model), and conversely, a single location may be influenced by multiple operations. This error propagation is naturally captured by our generalized local stochastic noise model, where each physical operation corresponds to a virtual location and the locations in the simplified model represent real locations. 

To account for this propagation, we transition from a $1$-confined model to a $\nu$-confined model (with $\delta = 0$). The confinement parameter $\nu$ remains a constant due to the (q)LDPC nature of the circuits: the sparsity of the constant degree of the gate network ensures that each faulty operation affects at most a constant number of qubits, and each qubit is affected by at most a constant number of operations. While increasing the confinement parameter $\nu$ necessarily reduces the numerical value of the fault-tolerant threshold, it does not invalidate its existence; the threshold remains finite as long as $\nu$ is bounded. 

\textbf{Surface-code logical errors.} We now examine our model for surface-code logical errors. The existence of a finite fault-tolerant threshold for the surface code is well-established (see Ref.~\cite{Gottesman2014}). If the physical error rate is below this threshold, the logical error probability is exponentially suppressed with the code distance $d_S$. However, because our protocol employs a small $d_S$, the logical error rate is non-negligible. 

In the following, we demonstrate that although these logical errors are potentially \textit{temporally non-local}, their impact can be effectively mitigated. Specifically, since the surface code is utilized within a constant-depth state preparation circuit, any temporal non-locality is strictly bounded. This limited temporal range ensures that such correlations can be successfully managed during the error correction process without compromising the threshold theorem. 

In the fault-tolerance framework of Ref.~\cite{Gottesman2014}, the primary objective is to bound both the probability of error-correction failure and the weight of the residual errors. When error correction succeeds, the residual noise can be faithfully characterized by a local stochastic model, which ensures that the logical error probability for subsequent operations remains suppressed by $e^{-\Omega(d_S)}$. Conversely, if an error-correction failure occurs, this exponential suppression can no longer be unconditionally guaranteed for subsequent stages of the circuit.

In the fault-tolerance proof presented in Ref.~\cite{Gottesman2014}, the essential task is to bound both the probability of error-correction failure and the weight of the residual errors following correction. If error correction succeeds, the residual errors are effectively described by a (standard) local stochastic noise model, which ensures that the logical error probability for subsequent operations remains suppressed by $e^{-\Omega(d_S)}$. Conversely, if an error-correction failure occurs, this exponential suppression can no longer be guaranteed for subsequent operations on the surface-code block.

This necessitates a more careful treatment of logical errors: an error-correction failure at a qubit-time position $(q, t)$ may induce faults throughout its \textit{causal cone}, which consists of all subsequent positions reachable through the circuit's operations. We can naturally incorporate these effects into our generalized noise model by treating error-correction failures as \textit{virtual faults}. Each such virtual fault is then mapped to a set of \textit{real faults} representing the logical errors at the affected qubit-time positions. 

Critically, since the surface code is utilized within a constant-depth state preparation circuit, each causal cone contains at most a constant number of locations. Consequently, while this propagation increases the confinement parameter $\nu$, it remains an $O(1)$ constant, thereby preserving the existence of a finite fault-tolerant threshold. 

Lastly, we address the errors associated with decoding operations. Analogous to the analysis in Lemma~\ref{lem:noisy_state_preparation}, if the noise on a surface-code block follows a (standard) local stochastic model, the decoding error rate is upper-bounded by $\lambda_S p$. However, due to potential temporal correlations of logical errors, this $\lambda_S p$ bound is strictly applicable only to qubits lying outside the causal cones of any prior error-correction failures. 

If $t$ such error-correction failures occur within the circuit, they may propagate to $O(t)$ qubits at the decoding locations. To incorporate this effect, we can model the decoding-operation errors as a $1$-confined model with an intrinsic fault parameter $\delta = O(t)$. Since the number of logical failures $t$ is exponentially suppressed by increasing the surface-code distance $d_S$, the parameter $\delta$ can be rigorously controlled in the same manner as $t_{RS}$. Including a finite intrinsic error parameter $\delta$ in the decoding-operation noise model merely results in an enlarged value for the parameter $t_{RS}$ in the generalized noise model of the output resource state in Lemma~\ref{lem:state_preparation}. 

{\section{Performance at finite distance}}
\label{app:performance}

While our protocol is optimized for asymptotic performance, its feasibility for moderate code distances is of practical interest. This section evaluates performance of the proposed methods in the finite-distance regime.

\subsection{The surface-code overhead}
\label{app:SC_overhead}

Although the surface-code distance $d_S$ is sub-polynomial in the large-distance limit, it may significantly influence the physical qubit requirements and the overall performance in the finite-distance regime. In this section, we provide a more explicit analytical estimation of $d_S$.

The resource-state preparation circuit, prior to surface-code encoding, comprises $(n_D + r_Z^D)(n_F + r_F)$ qubits and possesses a depth of $\omega_Z^D + \omega_F + 3$. Here, $n_D$ and $n_F$ denote the lengths of the deformed and factory codes, respectively; $r_Z^D$ and $r_F$ are the number of rows in the $Z$-type check matrix of the deformed code and the check matrix of the factory code, respectively. The parameters $\omega_Z^D$ and $\omega_F$ represent the maximum row/column weights of these respective matrices. To ensure that the logical error contribution in Eq.~\eqref{eq:logicalerror} is suppressed below $p/16$, it is sufficient to satisfy the condition:
\begin{equation}
    p_{L,S}(p) < \frac{\lceil t \rceil p}{16(\omega_Z^D + \omega_F + 3)(n_D + r_Z^D)(n_F + r_F)e},
\end{equation}
where $\lceil t \rceil \geq 1$. Using the standard phenomenological formula for the surface-code logical error rate, $p_{L,S}(p) \approx p_0 (p/p_{th,S})^{\lfloor (d_S+1)/2 \rfloor}$~\cite{PhysRevA.86.032324}, where $p_0$ is a positive constant and $p_{th,S}$ is the surface-code threshold, we can derive an estimate for the required surface-code distance:
\begin{equation}
    d_S \approx 2\frac{\ln\left( \frac{p/p_0}{16 (\omega_Z^D + \omega_F + 3)(n_D + r_Z^D)(n_F + r_F)e} \right)}{\ln(p/p_{th,S})}.
\end{equation}
Since $n_D, r_Z^D = O(d^{a+a_R})$ and $n_F, r_F = O(d^{a_F})$, while the weights $\omega_Z^D$ and $\omega_F$ remain constant, it follows that the required distance scales logarithmically with the memory-code distance, $d_S = O(\log d)$.

\subsection{Application of PCS to finite-length codes}
\label{app:PCS_finite}

In this section, we analyze the application of PCS to quantum codes with finite length. We demonstrate its performance advantage for implementing certain logical operations, even when the underlying code is not drawn from a good qLDPC family.

Consider the following logical circuit layer. We have $M$ memory-code blocks, each encoding $k$ logical qubits. Within each block, we label the logical qubits as a one-dimensional array indexed from $1$ to $k$. The layer consists of $Z \otimes Z$ logical measurements applied to nearest-neighboring logical qubit pairs. Assuming $k$ is even, the measurement is performed on each pair of logical qubits with indices $(2m-1, 2m)$, where $m$ ranges from $1$ to $k/2$.

Before illustrating the benefits of PCS for implementing such a layer, we first demonstrate its utility within practical quantum algorithms using a concrete example. Consider the simulation of time evolution under a one-dimensional Heisenberg model, where each spin is represented by a logical qubit. According to the Trotter decomposition, this simulation requires the implementation of two-qubit unitary gates of the form $e^{-i Z \otimes Z \theta}$ on nearest-neighboring spins. When decomposed into standard universal gates, the $e^{-i Z \otimes Z \theta}$ interaction can be realized using two controlled-NOT gates and an $e^{-i Z \theta}$ gate on the target qubit. Consequently, to implement this interaction across all spin pairs, we must perform controlled-NOT gates on each pair of nearest neighbors simultaneously. In the code surgery framework, each controlled-NOT gate requires an ancilla logical qubit and $Z \otimes Z$ measurements on the control and ancilla qubits. Therefore, to implement the $Z \otimes Z$ measurement layer on a one-dimensional logical qubit array, we allocate the logical qubits as follows: logical qubits with odd indices are used to represent the spins, and logical qubits with even indices are used as ancilla qubits. In this arrangement, the $e^{-i Z \otimes Z \theta}$ interaction between each spin pair involves the simultaneous execution of the specified operations across all blocks. The use of ancilla qubits leads to a scaling in the number of memory-code blocks $M = \lceil 2N/k \rceil$, where $N$ is the number of spins.

To estimate the qubit and time overhead for this $Z \otimes Z$ measurement layer under PCS, we recall that the deformed code has a length bounded by $n_D \leq (k_R+2n_R)n$. Given that PCS simultaneously implements the operation layer on $k_R$ blocks, the qubit overhead is given by:
\begin{equation}
    \frac{n_D}{k_R k} \leq \frac{\rho_R+2}{\rho_R\rho},
\end{equation}
where $\rho = k/n$ and $\rho_R = k_R/n_R$ are encoding rates for the memory code and R code, respectively. Here, we have neglected ancilla qubits for implementing stabilizer measurements. Absent the use of LTSP, the time overhead for PCS is simply $O(d)$. When the codes have constant encoding rates, the overall spacetime overhead is $O(1)$, outperforms code surgery methods listed in Table~\ref{tab:schemes}. 

To provide concrete numbers, we take HGP codes as an example. Assume the memory code is an HGP code generated by a linear code with rate $\rho_R$ (the same code used for the R code). Assuming the check matrix of the linear code is full-rank, the HGP rate is $\rho = \rho_R^2 / (1 + (1-\rho_R)^2)$. Substituting this into the PCS overhead formula yields:
\begin{equation}
    \frac{n_D}{k_R k} \leq \frac{(\rho_R + 2)[1 + (1-\rho_R)^2]}{\rho_R^3}.
\end{equation}
Notably, the qubit overhead remains constant, resulting in a total spacetime overhead of $O(d)$. Table~\ref{tab:hgp_examples_all} provides explicit qubit overhead values for several example linear codes.

\begin{table}[t]
\centering
\begin{tabular}{c c c c c}
\hline
Code family & Classical code & HGP code & $(\rho_R,\rho)$ & PCS qubit overhead \\
\hline
Hamming code
& $[7,4,3]$
& $[[58,16,3]]$
& $\left(\frac47,\frac{8}{29}\right)$
& $\le 17$ \\ 
\hline
\multirow{6}{*}{$(3,4)$-regular Gallager codes}
& $[12,3,4]$
& $[[225,9,4]]$
& \multirow{6}{*}{$\left(\frac14,\frac{1}{25}\right)$}
& \multirow{6}{*}{$\le 225$} \\
& $[20,5,6]$
& $[[625,25,6]]$ & & \\
& $[28,7,8]$
& $[[1225,49,8]]$ & & \\
& $[40,10,12]$
& $[[2500,100,12]]$ & & \\
& $[60,15,16]$
& $[[5625,225,16]]$ & & \\
& $[80,20,18]$
& $[[10000,400,18]]$ & & \\
\hline
\multirow{4}{*}{$(3,5)$-regular quasi-cyclic codes}
& $[80,34,12]$
& $[[8704,1160,12]]$
& \multirow{4}{*}{$\left(\ge \frac{2}{5},\ge \frac{2}{17}\right)$}
& \multirow{4}{*}{$\le 51$} \\
& $[105,44,16]$
& $[[14994,1940,16]]$ & & \\
& $[150,62,20]$
& $[[30600,3844,20]]$ & & \\
& $[210,86,24]$
& $[[59976,7400,24]]$ & & \\
\hline
\end{tabular}
\caption{Examples of classical codes, the corresponding symmetric HGP memory codes, and the resulting PCS qubit-overhead bounds. The $(3,4)$-regular Gallager code examples and the $(3,5)$-regular quasi-cyclic code examples are taken from Refs.~\cite{Xu2024,Fossorier2004}.}
\label{tab:hgp_examples_all}
\end{table}

\subsection{LTSP with more general factory codes}
\label{app:LTSP_finite}

To achieve constant qubit and time overheads for resource state preparation, we have primarily focused on the asymptotic regime where the F code is selected from a family of LTCs with constant encoding rate and constant soundness. Although theoretical constructions for such codes exist~\cite{Panteleev2022,lin2022c3LTC}, identifying concrete code instances optimized for the finite-distance regime remains a subject that requires further investigation.

Critically, neither the constant encoding rate nor the constant soundness requirement is strictly necessary for finite-distance applications. Any LDPC code with a non-zero soundness parameter can be utilized as an F code. For example, one could employ a repetition code with a check matrix defined by a bounded-degree expander graph, where vertices represent bits and edges represent checks enforcing equality between adjacent bits. In this case, the soundness of the code is controlled by the expansion properties of the graph. Specifically, the soundness satisfies $s=\frac{n_F}{r_F}h$, where $h$ denotes the edge expansion of the graph. In particular, for a constant-degree expander family with $h=\Omega(1)$, this yields $s=\Omega(1)$.

Using such a construction for the F code allows for a significant reduction in time overhead at the expense of additional qubits. For an F code based on a length-$d$ repetition code, the circuit depth for state preparation remains almost constant at $O(d_S)$, while the qubit overhead increases to $O(d d_S^2)$, where the factor $d$ originates from the vanishing encoding rate of the repetition code. This allows for a trade-off between qubit overhead and temporal efficiency, illustrating the flexibility of the LTSP approach, where even codes that do not satisfy strictly constant asymptotic requirements remain viable for applications in fault-tolerant quantum computation.

\end{widetext}

\bibliography{references.bib}

\end{document}